\documentclass{amsart}

\title[Regularizations in the perturbed KP solitons]
{Regularizations for shock and rarefaction waves in the perturbed solitons of the KP equation}

\author{Guangfu Han$^1$, Yuji Kodama$^{1,2}$, Chuanzhong Li$^1$, Lin Sun$^1$}
\date{\today}
\address{$^1$ College of Mathematics and Systems Science, Shandong University of Science and Technology, Qingdao, 266590, China}
\address{$^2$Department of Mathematics, The Ohio State University,
	Columbus, OH 43210}
\subjclass[2000]{}

\usepackage{url}
\usepackage{color}
\usepackage{rotating}

\countdef\x=23
\countdef\y=24
\countdef\z=25
\countdef\t=26

\def\tbox(#1,#2)#3{
	\x=#1 \y=#2
	\multiply\x by 12
	\multiply\y by 12
	\z=\x \t=\y
	\advance\z by 12
	\advance\t by 12
	\psline(\x,\y)(\x,\t)(\z,\t)(\z,\y)(\x,\y)
	\advance\x by 6
	\advance\y by 6
	\rput(\x,\y){{\bf #3}}}

\usepackage{paralist}
\usepackage{graphics} 
\usepackage{epsfig} 
\usepackage{graphicx}
\usepackage{epstopdf}
\usepackage{subfigure}
\usepackage{array}

\usepackage{float}
\usepackage{caption}
\usepackage{CJK}
\usepackage{amssymb}
\usepackage{graphicx}
\usepackage{epstopdf}
\usepackage{amscd}
\usepackage{appendix}
\usepackage{graphics}
\usepackage{tikz}
\usepackage{mathdots}
\usepackage{xcolor}
\usepackage[colorlinks=true, linkcolor=red]{hyperref}
\def\proof{\par{\it Proof}. \ignorespaces}
\def\endproof{{\ \vbox{\hrule\hbox{%
				\vrule height1.3ex\hskip0.8ex\vrule}\hrule }}\par}

\theoremstyle{definition}

\theoremstyle{remark}

\numberwithin{equation}{section}

\def\Wr{\mathop{\mathrm{Wr}}\nolimits}
\let\trueint=\int
\let\truesum=\sum
\def\int{\mathop{\textstyle\trueint}\limits}
\def\sum{\mathop{\textstyle\truesum}\limits}
\let\<=\langle
\let\>=\rangle

\def\sech{\mathop{\rm sech}\nolimits}

\def\Wr{\text{Wr}}

\def\t{\mathbf{t}}
\def\0{\mathbf{0}}

\def\edge{\ar@{-}}
\def\dedge{\ar@{.}}

\usepackage{amsmath, amsthm, amssymb, amsbsy,amsfonts,amscd}
\usepackage[mathscr]{eucal}
\usepackage{epsfig}
\usepackage{young}
\newtheorem{theorem}{Theorem}[section]

\newtheorem{definition}[theorem]{Definition}
\newtheorem{proposition}[theorem]{Proposition}
\newtheorem{lemma}[theorem]{Lemma}

\newtheorem{example}[theorem]{Example}

\newtheorem{remark}[theorem]{Remark}

\usepackage{amsfonts}
\usepackage{xy}
\usepackage{amssymb}
\usepackage{amsmath}

	\newcommand{\thmrefer}[1]{\renewcommand\thetheorem
		{\protect\ref{#1}}\addtocounter{theorem}{-1}}
	
	\xyoption{all}
	\CompileMatrices

\begin{document}

\begin{abstract}
      Using an asymptotic perturbation method, we study the initial value problem for the KP equation with initial data consisting of parts of exact line-soliton solutions. We consider a slow modulation of the soliton parameters, described by a dynamical system obtained via the perturbation method. {The dynamical system is given by a $2$-component quasi-linear system.} In particular, we show that a singular solution (\emph{shock wave}) {of the system} leads to the generation of a new soliton as a result of the resonant interaction of solitons. We also show that a regular solution corresponding to a rarefaction wave {of the system} can be described by a parabola (which we call a \emph{parabolic soliton}). We then perform numerical simulations of the initial value problem and show that they are in excellent agreement with the results obtained by the perturbation method.
\end{abstract}
      
\maketitle
    
      \begingroup
  \hypersetup{linkcolor=black}
  \tableofcontents
  \endgroup 
	
     
\section{Introduction}
\par The KP equation is a two-dimensional nonlinear dispersive wave equation given by
	\begin{eqnarray}\label{0}
		(4u_{t}+6uu_{x}+u_{xxx})_{x}+3u_{yy}=0,
	\end{eqnarray}
where $x$, $y$, and $t$ are the spatial coordinates and time, $u=u(x,y,t)$ represents the (normalized) wave amplitude, and the subscripts denote partial derivatives. It is well-known that the KP equation admits a line soliton solution $u(x,y,t)=\phi(\xi-x_0;\kappa_i,\kappa_j)$ with two constants $\{\kappa_i,\kappa_j\}$, the soliton parameters (see for example \cite{Kodama3}),
	\begin{eqnarray}\label{01}
		\phi(\xi-x_0;\kappa_i,\kappa_j)=A_{[i,j]}\sech^2\Big(\sqrt{\dfrac{A_{[i,j]}}{2}}(\xi-x_0)\Big),\quad\text{with}\quad \xi=x+\tan\varPsi_{[i,j]} y-C_{[i,j]}t,
	\end{eqnarray}
where $x_0$ gives a constant phase. Here
the amplitude $A_{[i,j]}$ and the soliton inclination from $y$-axis $\tan\varPsi_{[i,j]}$, 
and the velocity $C_{[i,j]}$ are expressed in terms of $\kappa_i$ and $\kappa_j$, 
\[
A_{[i,j]}=\frac{1}{2}(\kappa_i-\kappa_j)^2,\qquad \tan\Psi_{[i,j]}=\kappa_i+\kappa_j,\qquad C_{[i,j]}=\kappa_i^2+\kappa_i\kappa_j+\kappa_j^2.
\]
Note that $C_{[i,j]}>0$, which implies every line-soliton propagates in the positive $x$-direction.
We call the soliton solution \eqref{01} line-soliton of $[i,j]$-type (or simply \emph{$[i,j]$-soliton}).
The KP equation is a two-dimensional generalization of the Korteweg-deVries (KdV) equation, and the KdV soliton is recovered when $\kappa_{i}=-\kappa_{j}$ in (\ref{01}). 
It is also well-known that Eq.~(\ref{0}) admits a resonant soliton solution.
The resonant solution is observed in the Mach reflection problem of shallow water waves (see, e.g. Chapter 8 in \cite{Kodama3}). In \cite{Miles2}, Miles showed that two obliquely interacting line solitons become resonant at a certain critical interaction angle. As a result of the resonance, the phase shift between these line solitons becomes infinity, and the resonance generates an $additional$ soliton(s).
The resonant solution forms a $Y$-shape soliton, simply called \emph{$Y$-$soliton$} \cite{KY:16} (also see section \ref{sec:Y} for the details). In Appendix \ref{A-KP}, we also provide a brief review of the general soliton solutions of the KP equation, referred to as \emph{KP solitons}, and their classification
(see \cite{Kodama3}).

One should note here that the stability problem of these solitons is widely open except for the case of one line-soliton (see \cite{Mizumachi}, also  Remarks 6.1 and 6.2 in \cite{Kodama3}). 
It was shown in \cite{Mizumachi} that a small perturbation generates local phase shifts propagating along the line-soliton, but asymptotically the soliton parameters $\{\kappa_i,\kappa_j\}$ remain unchanged, i.e., the amplitude $A_{[i,j]}$ and the slope $\Psi_{[i,j]}$ of the soliton remain the same. More precisely, for one line-soliton with a small perturbation, it was shown that  as $t\to\infty$,
\begin{equation}\label{1-solST}
	\iint_{D_t}|u(x,y,t)-\phi(x-x_0;\kappa_i,\kappa_j)|^2dxdy\longrightarrow 0,\qquad\text{for some}\quad x_0,
\end{equation}
where $D_t\subset \mathbb{R}^2$ is any compact domain including the line soliton and it depends on $t$ (see also Chapter 6 in \cite{Kodama3} for the details).
We emphasize that the parameters $\{\kappa_i,\kappa_j\}$ for one-soliton stay the same, unlike the case of the KdV soliton, whose parameters change under even a small perturbation in general.
 
Recently, there have been several publications on the initial value problems of Eq.~(\ref{0}) with certain classes of initial data, which include \cite{Kao,McDowell,Ryskamp1,Ryskamp} for numerical and semi-analytical studies,   \cite{LYK:11} for shallow water experiments and \cite{Wang,Yuan2} ocean simulations. 

Their work demonstrates that solutions to the initial value problem with specific types of initial condition approaches to certain KP soliton solutions. These results may be stated as follows, which is an extension of \eqref{1-solST} (see Chapter 6 in \cite{Kodama3}). For this type of initial data, there exists a KP soliton so that 
\begin{equation}\label{localstability}
	\iint_{D_t}|u(x,y,t)-u_0(x,y,t)|^2dxdy\longrightarrow 0,\qquad\text{as}\quad t\to \infty,
\end{equation}
where the integration domain $D_t$ may be taken to cover the ``main part'' (or a central part of the interaction patterns) of the solution, and $u_0(x,y,t)$ is an exact soliton solution, KP soliton.
In the present paper, we study this type of stability for some explicit initial data.
The initial data we consider are those in \cite{Kao} (also see \cite{Kodama3}),
which include $V$-type initial value waves (see Figure \ref{fig96} below).
It should be noted that this problem with some initial data was first numerically studied in \cite{PTLO} for the Mach reflection phenomena. The phenomena were later explained in terms of the KP solitons in \cite{Kodama} (see also Chapter 8 in \cite{Kodama3}).
In \cite{Kao}, the initial value problems with several initial data are studied numerically, and their result leads to a conjecture that the asymptotic solution of the perturbed problem converges to a certain KP soliton in the sense of \eqref{localstability}. Our main result {in this paper} is to confirm the conjecture by analytically solving a \emph{quasi-linear} system describing the dynamics of the soliton parameters $(\kappa_i,\kappa_j)$, which depend on the slowly varying variables $(Y=\epsilon y, T=
\epsilon t)$ for some small parameter $0<\epsilon \ll 1$.
We provide an elementary derivation of the system in Appendix \ref{A-kappa}, and it is given by
\begin{eqnarray}{\label{6}}		
	\frac{\partial}{\partial T}\left(\begin{array}{cc}
		\kappa_{1} \\
		\kappa_{2} \\
	\end{array}\right)+\left(\begin{array}{cc}
		2\kappa_{1}+\kappa_{2}& 0\\
		0& \kappa_{1}+2\kappa_{2}\\
	\end{array}\right)
	\frac{\partial}{\partial Y}\left(\begin{array}{cc}
		\kappa_{1} \\
		\kappa_{2} \\
	\end{array}\right)=0.
\end{eqnarray}
Note here that the soliton parameters $(\kappa_1,\kappa_2)$ are the Riemann invariants of the system. This system has also been derived in  \cite{Ryskamp1,Ryskamp} using the Whitham modulation theory \cite{Wh:74} for  $A_{[i,j]}$ and $\tan\varPsi_{[i,j]}$. 

In general, a quasi-linear system admits a singular solution, called a shock wave. 
To obtain a global solution, we regularize the initial data in a similar way as in the KdV-Whitham theory in \cite{BK:94, Kodama5} (see also Appendix \ref{A-KW}). We then show 
that a shock wave in the $\kappa$-system generates a soliton as a result of resonant interaction of solitons. The main result of the present paper is to provide an analytical explanation for the asymptotic stability in the sense of \eqref{localstability} \cite{Kao, Kodama3}.

The paper is organized as follows. In Section \ref{Sec:KPsolitons}, we give some details on
line-solitons and Y-soliton as the necessary background for our study. In particular, we discuss
the resonance phenomena of two solitons of so-called \emph{O-type} following \cite{Miles2}
(see also \cite{Kodama3}). Here, we introduce the \emph{colored $\kappa$-graph} (Definition \ref{CG}) to describe some of the KP solitons, which will play the main role in the paper.
In particular, we show that a Y-soliton can be described by a ``singular'' colored $\kappa$-graph, which is obtained by a limit of the soliton parameters in O-type soliton. This limit corresponds to the resonance found in \cite{Miles2}.
In Section \ref{Sec:kappa}, we discuss some properties of the $\kappa$-system \eqref{6}. In particular, we give a condition for the global existence of the solution (Lemma \ref{simple}). Moreover, we set up the initial value problem of the $\kappa$-system \eqref{6} with particular set of initial data (see Section \ref{sec:IVP}). Then in Section \ref{sec:H}, we study simple but important examples,
where the initial data consist of a semi-infinite line-soliton, referred to as a \emph{half-soliton}.
We show, in particular, that the rarefaction wave can be described by a perturbed soliton whose
peak trajectory has a parabolic shape (we call it a \emph{parabolic-soliton}).
These results provide part of the building blocks for the solutions we study in the paper.
In Section \ref{Sec:IVP-V}, we study the initial value problem of the $\kappa$-system \eqref{6} with V-shape initial data consisting with two half-solitons. The initial data for the $\kappa$-system is
then given by step functions. The main result of this section is to regularize the step
initial data, so that the initial value problem of the $\kappa$-system admits a global solution.
In particular, we find that the shock singularity can be regularized by adding a new soliton
(Section \ref{sec:d}).
This regularization is due to the resonant interaction of the KP solitons.
Then we find that the asymptotic solution consists of line-solitons and parabolic-solitons,
and the solution converges locally to some exact KP solitons in the sense of the local stability \eqref{localstability}.
In Section \ref{summary}, we give a summary of the results of the initial value problems
with V-shape initial data (Theorem \ref{main}).

We also provide a brief review of the KP solitons in Appendix \ref{A-KP}, 
an elementary derivation of the $\kappa$-system \eqref{6} in Appendix \ref{A-kappa}, and a brief review of the regularization in the KdV-Whitham equation in Appendix \ref{A-KW}.

\section{Background}\label{Sec:KPsolitons}
In this section, we briefly review soliton solutions of the KP equation, particularly some details of the one soliton solution, two solitons solution and a resonant soliton solution. Here, we fix the notations of those solutions and introduce the chord diagram (permutation diagram) to describe the asymptotic structure of the solitons (also see Appendix \ref{A-KP}).

\subsection{One line-soliton}
The KP equation admits a steady propagating wave of the KP equation \eqref{0} in the form,
	\begin{eqnarray}\label{92}
			u(x,y,t)=A\sech^2\dfrac{1}{2}({\bf K} \cdot ({\bf x}-{\bf x}_{0})-\Omega t),
	\end{eqnarray}
where $A$ is the amplitude, ${\bf K}=(K^{x},K^{y})$ is the wave vector with ${\bf x} =(x,y)$, $\Omega$ is the frequency, and ${\bf x}_{0}=(x_{0},y_{0})$ is a constant vector. This solution
 is localized along the line ${\bf K} \cdot({\bf x}-{\bf x}_{0})-\Omega t=0$ (the wave crest) and decays exponentially away from the line. For  $x\to +\infty$, the solution \eqref{92} has the asymptotic form,
\begin{eqnarray*}
u(x,y,t)~\longrightarrow~ A\exp{({\bf K}\cdot {\bf x}_{0})}\cdot\exp{(-{\bf K}\cdot {\bf x}+\Omega t)},
\end{eqnarray*}
where we have assumed ${K^{x}}>0$. Then, from the KP equation (\ref{0}), we see that the constants $({\bf K},\Omega)$ satisfy the (soliton) \emph{dispersion relation},
	\begin{eqnarray}
		-4\Omega K^{x}+(K^{x})^4+3(K^{y})^2=0.
	\end{eqnarray}
The dispersion relation can be \emph{parametrized} by a pair of arbitrary constants $\{\kappa_{i},\kappa_{j}\}$, called \emph{soliton parameters}, such that
\begin{eqnarray}\label{93}
{\bf{K}}=(K^{x},K^{y})=\left(\kappa_{j}-\kappa_{i},\kappa_{j}^2-\kappa_{i}^2\right),
\qquad\Omega=\kappa_{j}^3-\kappa_{i}^3.
\end{eqnarray}
Note that the condition $K^{x}>0$ implies $\kappa_{i}<\kappa_{j}$, and the amplitude $A$ is given by $A=\frac{1}{2}(\kappa_{i}-\kappa_{j})^2$. We call the solution (\ref{92}) with \eqref{93} $[i,j]$-soliton, and we write $A = A_{[i,j]}$, ${\bf K} = {\bf K}_{[i,j]}$ and $\Omega= \Omega_{[i,j]}$. The slope of the crest and  the velocity in the $x$-direction of $[i,j]$-soliton are given by
	\begin{eqnarray}\label{90}
\tan \Psi_{[i,j]}=\dfrac{{K_{[i,j]}^{y}} }{{K_{[i,j]}^{x}}} =\kappa_{i}+\kappa_{j},\qquad C^{x}_{[i,j]}=\dfrac{\Omega_{[i,j]}}{{K_{[i,j]}^{x}}}=\kappa_{i}^2+\kappa_{i}\kappa_{j}+\kappa_{j}^{2},
	\end{eqnarray}
where the angle $\Psi_{[i,j]}$ is measured in the counter-clockwise direction form the $y$-axis. One should note that there is no soliton parallel to the $x$-axis (i.e., $-\frac{\pi}{2}<\Psi_{[i,j]}< \frac{\pi}{2}$). 

As shown in Appendix \ref{A-KP}, the solution is expressed by $u(x,y,t)=2(\ln\tau(x,y,t))_{xx}$, and
the $\tau$-function of the line-soliton \eqref{92} is given by
\begin{equation}\label{One-sol}
\tau(x,y,t)=E_i(x,y,t)+aE_j(x,y,t)\qquad \text{with}\qquad E_i(x,y,t)=\exp(\kappa_ix+\kappa_i^2y-\kappa_i^3t),
\end{equation}
where $a$ is a positive constant. Then the \emph{peak trajectory} (wave crest) is given by the line $L_{[i,j]}$,
\begin{equation}\label{crest}
L_{[i,j]}~:~{\bf K}\cdot ({\bf x}-{\bf x}_0)-\Omega t=(\kappa_j-\kappa_i)(x+\tan\Psi_{[i,j]}y-C_{[i,j]}t+x_{[i,j]}^0)=0,
\end{equation}
where $x_{[i,j]}^0=\frac{1}{\kappa_j-\kappa_i}\ln a$.

In this paper, we study an \emph{adiabatic} deformation of solitons under some perturbations,
which can be described by changes in the soliton parameters on slow time scales, i.e., \eqref{6}.
We then define the \emph{colored $\kappa$-graph} to illustrate the dynamics of the parameters for $[i,j]$-soliton.

\begin{definition}\label{CG}
	Along any line \begin{tikzpicture}
		\draw[color=blue, thick] (0,0.1) -- (0.5,0.1);
		\draw[color=red, thick] (0,-0.1) -- (0.5,-0.1);		
	\end{tikzpicture} parallel to the $y$-axis (There is a gap between the two lines), we will obtain the amplitude variation of the line soliton on this line, with a peak value of $A_{[i,j]}$. This means that line solitons can be represented as two lines of different colours (red and blue) paired together. Therefore, we can represent this soliton on the $y$-$\kappa$ plane, which is called colored $\kappa$-graph or $\kappa$-graph. See Figure \ref{fig1}.
\end{definition} 
\begin{figure}[H]
	\begin{minipage}[htb]{1\linewidth}
		\centering
		\includegraphics[height=3.75cm,width=11.92cm]{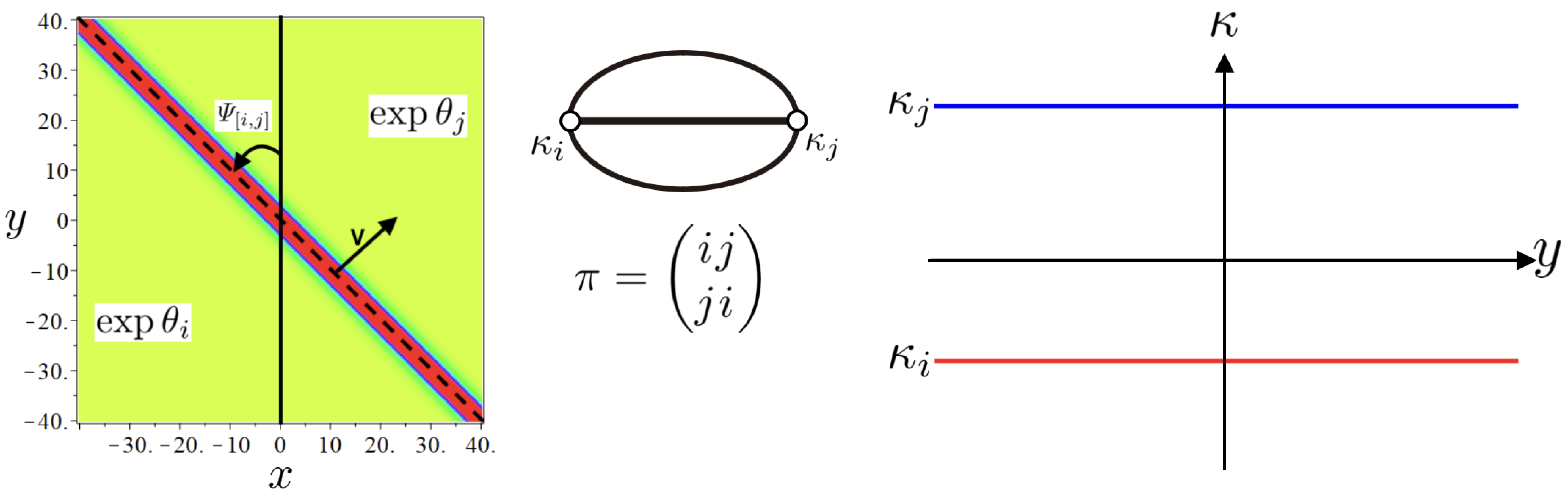}
	\end{minipage}%
	\caption{The left panel shows the contour plot of the $[i,j]$-soliton solution (\ref{92}) with $\kappa_{i}=-1$, $\kappa_{j}=2$ at $t=0$. The dotted line is the crest of the soliton. The middle panel shows the corresponding permutation (transposition $i\leftrightarrow j$), which we call the chord diagram of the soliton. The diagram indicates the asymptotic structure of KP soliton, that is, 
		the upper (lower) part of the diagram shows the $[i,j]$-soliton for $y\gg0$ ($y\ll 0$). The right panel shows the corresponding colored $\kappa$-graphs.}
	\label{fig1}
\end{figure}

The soliton parameters are, of course, constants without perturbations, and the $\kappa$-graph
for line-soliton is trivial. {The main tool of our study is the colored $\kappa$-graph, which we use to describe the dynamics of the parameters in the presence of perturbations.}
We give the following remarks about the colored $\kappa$-graph and adiabatic deformation of $[i,j]$-soliton.
\begin{itemize}
	\item[1.] The $blue$ line represents the larger parameter $\kappa_{j}$, and the $red$ line represents the smaller parameter $\kappa_{i}$ in the pair $\{\kappa_{i},\kappa_{j}\}$. Then if the blue and red lines coincide  \begin{tikzpicture}
		\draw[color=blue, thick] (0,0.00) -- (0.5,0.00);
		\draw[color=red, thick] (0,-0.05) -- (0.5,-0.05);		
	\end{tikzpicture} , we have $\kappa_{i}=\kappa_{j}$ and $A_{[i,j]}=0$, i.e., there is no soliton.	 
	\item[2.]  The adiabatic deformation of the line-soliton can be described by the small scales $(X=\epsilon x, Y=\epsilon y,T=\epsilon t)$.   
	Then the {peak trajectory} of the $[i,j]$-soliton is described by the ``curve''
	\begin{equation}\label{Trajectory}
		X+\tan\Psi_{[i,j]}Y-C_{[i,j]}T=0,
	\end{equation}
	which is given by the limit $\epsilon\to 0$ for $\epsilon L_{[i,j]}$ in \eqref{crest} {with $X,Y,T \sim \mathcal{O}(1)$.}  
    {The phase part $x_{[i,j]}^0$ is ignored in \eqref{crest}, so that the line-soliton \eqref{Trajectory} intersects the origin at $T=0$.}
	Then we see that the slow scale $X$ should be considered as a function of $(Y,T)$, i.e., $X=X(Y,T)$.
	Then taking the variation of $X$, i.e., $dX+\tan\Psi_{[i,j]}dY-C_{[i,j]}dT=0$,  we obtain
	\begin{equation}\label{dX}
		\frac{\partial X}{\partial Y}=-\tan\Psi_{[i,j]}=-(\kappa_i+\kappa_j),\qquad \frac{\partial X}{\partial T}=C_{[i,j]}=\kappa_i^2+\kappa_i\kappa_j+\kappa_j^2,
	\end{equation}
	which gives the curve of the peak trajectory. This is the main object that we study in the present paper.
\end{itemize}

\subsection{O-type soliton and Y-soliton}\label{sec:OY}
Here, we review some particular KP solitons such as O-type soliton and Y-soliton. 
The main purpose of this section is to show that a Y-soliton is generated as a result of \emph{resonant} interaction of two line-solitons (this was first discovered by Miles in \cite{Miles2}). As will be explained in the following Sections \ref{sec:H} and \ref{Sec:IVP-V}, the resonance plays an important role in our regularization of a shock singularity.

\subsubsection{O-type soliton}\label{sec:O}
We recall that two solitons, say $[i,j]$-soliton and $[k,l]$-soliton, are of O-type, if the soliton parameters of these solitons satisfy
\[
\kappa_i<\kappa_j<\kappa_k<\kappa_l.
\]
In this case, the totally nonnegative matrix $A$ and the exponential matrix $E$ in the $\tau$-function $\tau=|AE^T|$ in \eqref{26} are given by
\[
A=\begin{pmatrix}
1 & a & 0 & 0\\
0&0&1&b
\end{pmatrix},\qquad\text{and}\qquad 
E=\begin{pmatrix}
E_1 &E_2 & E_3& E_4\\
\kappa_1 E_1&\kappa_2E_2&\kappa_3E_3&\kappa_4E_4
\end{pmatrix},
\]
where $a$ and $b$ are positive constants, and $E_{i}=\exp{(\kappa_{i}x+\kappa_{i}^2y-\kappa_{i}^3t)}$.
Then the $\tau$-function in the form (\ref{26}) is 
\begin{align}\label{150}
		\tau=|AE^T|=&=E_{1,3}+bE_{1,4}+aE_{2,3}+abE_{2,4}, 
	\end{align}  
where $E_{i,j}=(\kappa_{j}-\kappa_{i})E_{i}E_{j}$. Take the following values of the parameters $a$ and $b$, so that two line-solitons in the O-type soliton intersect at the origin $(0,0)$ at $t=0$ (see  \cite{CK:09}),
\[
a=\sqrt{\frac{(\kappa_3-\kappa_1)(\kappa_4-\kappa_1)}{(\kappa_3-\kappa_2)(\kappa_4-\kappa_2)}},\qquad\text{and}\qquad b=\sqrt{\frac{(\kappa_3-\kappa_1)(\kappa_3-\kappa_2)}{(\kappa_4-\kappa_1)(\kappa_4-\kappa_2)}}.
\]
Then the $\tau$-function becomes
\begin{equation}\label{O-soliton}
\tau=(\kappa_3-\kappa_1)\left(E_1E_3+\Delta E_2E_3+\Delta E_1E_4+E_2E_4\right),
\end{equation}
where the coefficient $\Delta$ is given by
\[
\Delta=\sqrt{\frac{(\kappa_3-\kappa_2)(\kappa_4-\kappa_1)}{(\kappa_3-\kappa_1)(\kappa_4-\kappa_2)}}~<~1,
\]
{which gives the phase shift resulting from the nonlinear interaction between these solitons.}
Figure \ref{fig101} shows an example of O-type soliton solution. For $y\gg0$, the O-type soliton has two solitons of $[1,2]$-, and $[3,4]$-type, which are given by
\begin{align*}
u(x,y,0)~\approx~&A_{[1,2]}\sech^2\frac{\kappa_2-\kappa_1}{2}\left(x+(\kappa_1+\kappa_2)y-\frac{1}{\kappa_2-\kappa_1}\ln\Delta\right)\\
+&A_{[3,4]}\sech^2\frac{\kappa_4-\kappa_3}{2}\left(x+(\kappa_3+\kappa_4)y+\frac{1}{\kappa_4-\kappa_3}\ln \Delta\right).
\end{align*}
These two solitons intersect at $(x,y)=(x_+,y_+)$ with
\begin{equation}\label{y+}
\left\{
\begin{array}{lll}
\displaystyle{x_+=-(\kappa_1+\kappa_2)y_++\frac{1}{\kappa_2-\kappa_1}\ln\Delta},\\[2.0ex]
\displaystyle{y_+=\frac{(\kappa_4-\kappa_1)-(\kappa_3-\kappa_2)}{(\kappa_4-\kappa_1)+(\kappa_3-\kappa_2)}\frac{-\ln \Delta}{(\kappa_4-\kappa_3)(\kappa_2-\kappa_1)}>0}.
\end{array}\right.
\end{equation}
It should be noted that the middle section of the O-type soliton in Figure \ref{fig101} represents the phase shift. In this figure, we give a relatively large phase shift by taking special values of the parameters to explain a generation of the Y-soliton as the result of the resonance interaction of two solitons (see also \cite{CK:09} for the choice of the parameters).
\begin{figure}[htbp]
	\begin{minipage}[htb]{1\linewidth}
		\centering
		\includegraphics[width=14cm,height=4cm]{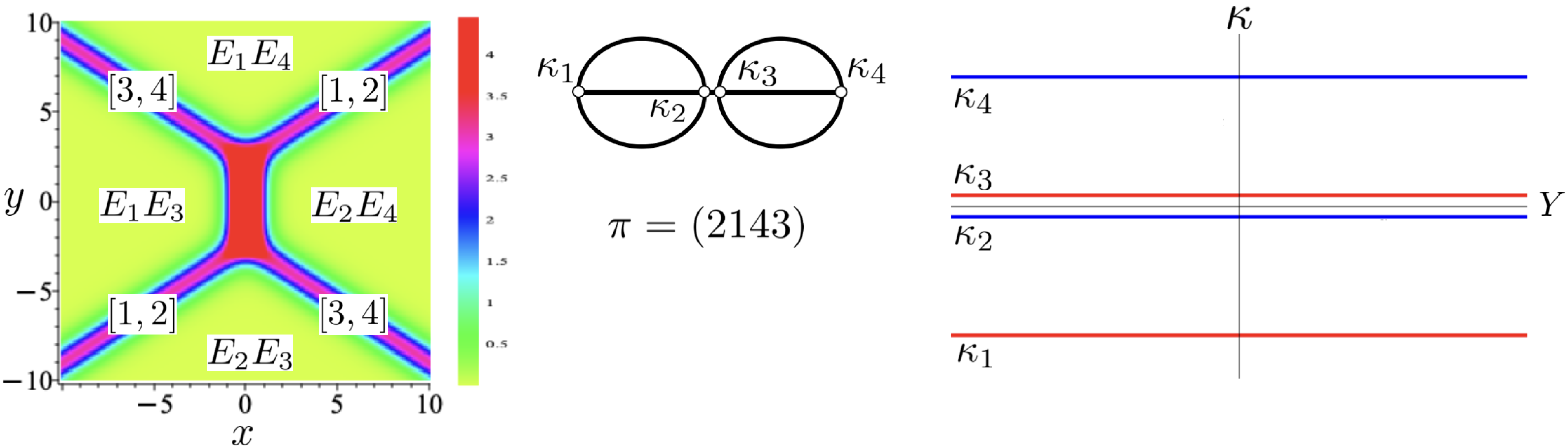}
	\end{minipage}
	\caption{The parameters in the matrix $A$ are $a=\tfrac{1}{a'}=\sqrt{\tfrac{3\cdot 10^5}{2}}$, $b=\tfrac{1}{b'}=\sqrt{\tfrac{2}{3\cdot 10^{5}}}$ (i.e., $ab=1$). The $\kappa$-parameters are given by $(\kappa_{1},\kappa_{2},\kappa_{3},\kappa_{4})=(-\frac{3}{2},-10^{-5},10^{-5},\frac{3}{2})$. 
	The left panel shows the contour plot of the solution $u(x,y,0)$. The middle panel is the chord diagram for the O-type soliton. The right panel shows the colored $\kappa$-graph in the slow scale $Y=\epsilon y$, and note that the phase shift in the left figure is ignored in this scale.}
	\label{fig101}
\end{figure}

Now we consider the limit $\kappa_3\to\kappa_2$. Then, the middle part (phase shift)
becomes $[1,4]$-soliton, which can be 
easily seen from the $\tau$-function \eqref{O-soliton}, that is, noting $\Delta\to0$
and $E_3\to E_2$, we have
\[
\tau~\longrightarrow~ (\kappa_3-\kappa_1)E_3(E_1+E_4).
\]
The solution $u=2(\ln\tau)_{xx}$ with the parameter $\kappa_1=-\kappa_4$ gives the $[1,4]$-soliton parallel to the $y$-axis, i.e.,
\[
u(x,y,0)=A_{[1,4]}\sech^2(\frac{\kappa_4-\kappa_1}{2}x).
\]
In the colored $\kappa$-diagram, this implies that the limit of $\kappa_3\to \kappa_2$ leads to the cancellation of the red line of $[3,4]$-soliton with the blue line of $[1,2]$-soliton, and generates the $[1,4]$-soliton (called the Mach stem \cite{Miles2}). With two solitons $[1,2]$- and $[3,4]$-solitons in the asymptotic regions $|y|\gg0$, the limit induces  a three wave resonance among $[1,2]$-, $[2,4]$-, and $[1,4]$-solitons, that is, we have  the resonant triad in the wave number space,
\[
{\bf K}_{[1,4]}={\bf K}_{[1,2]}+{\bf K}_{[2,4]}.
\]
Figure \ref{fig111} shows the resonant interaction at $y=y_+$ in \eqref{y+} for $y\gg0$ in
the limit $\kappa_3\to\kappa_2$. Similarly, we have the resonant interaction at $y=y_-=-y_+$
as shown also in Figure \ref{fig101}. We also represent the corresponding resonant solitons (Y-solitons) and the colored $\kappa$-graphs. 
\begin{figure}[htbp]
	\begin{minipage}[H]{1\linewidth}
		\centering
		\includegraphics[width=10.51cm,height=8cm]{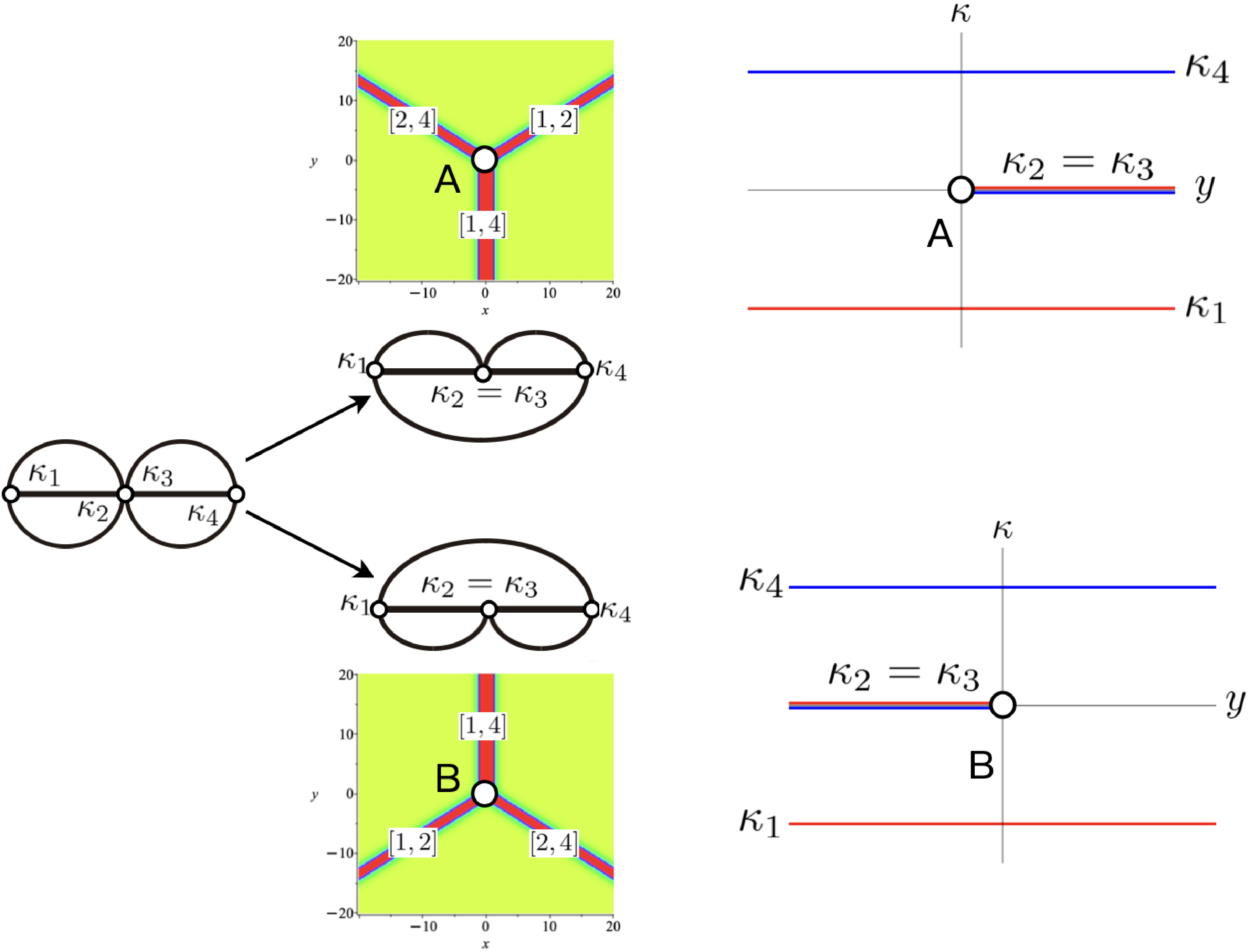}
	\end{minipage}
	\renewcommand\figurename{Figure}
	\caption{Y-solitons as the result of the resonant interactions of two line-solitons
	of O-type in the limit $\kappa_3\to\kappa_2$. In this limit,
	the chord diagram becomes singular, and the middle panel shows that the singular  chord diagram splits into two asymptotic diagrams for $y\gg 0$ and $y\ll 0$. Then the singularity can be represented by three wave resonant interaction, i.e., the generation of $[1,4]$-soliton. The colored $\kappa$-graphs around the points $A$  and $B$ are on the shifted coordinates $y-y_+$ and $y+y_+$, respectively.}
	\label{fig111}
\end{figure}

\subsubsection{Y-soliton: resonant solution}\label{sec:Y}
Here, we review the soliton resonance and Y-soliton as an exact solution of the KP equation. We first recall that each soliton, say $[i,j]$-soliton, is parametrized by a pair of the numbers $\{\kappa_{i},\kappa_{j}\}$ with (\ref{93}), i.e.,
	\begin{eqnarray*}
		{\bf K}_{[i,j]}=(\kappa_{j}-\kappa_{i},\kappa_{j}^2-\kappa_{i}^2),\qquad \Omega_{[i,j]}=\kappa_{j}^3-\kappa_{i}^3.
	\end{eqnarray*}
It is then immediate to see the following relation among Y-soliton of $[1,2]$-, $[2,3]$-, and $[1,3]$-solitons with arbitrary ordered parameters ${\kappa_{1}<\kappa_{2}<\kappa_{3}}$,
	\begin{eqnarray}\label{41}
		{\bf K}_{[1,3]}={\bf K}_{[1,2]}+{\bf K}_{[2,3]},\qquad \Omega_{[1,3]}=\Omega_{[1,2]}+\Omega_{[2,3]},
	\end{eqnarray}
which is called the \emph{three wave resonant relations}.
There are two types of resonances, and they correspond to the permutations $\pi= (312)$ and $\pi=(231)$, as shown in Figure \ref{fig111} (see \cite{KY:16}). For the case $\pi=(312)$, we have the $\tau$-function
in \eqref{26} with
\[
A=(1,1,1),\quad\text{and}\quad E=(E_1,E_2,E_3).
\]
Here, note that we choose the specific $A$ so that the intersection point is located at the origin $(0,0)$. For the case $\pi=(2,3,1)$, we have
\[
A=\begin{pmatrix}
1 & 0 &-a\\
0&1&b
\end{pmatrix},\quad\text{and}\quad
E=\begin{pmatrix}
E_1&E_2&E_3\\
\kappa_1E_1&\kappa_2E_2&\kappa_3E_3
\end{pmatrix},
\]
where $a= \frac{\kappa_1-\kappa_2}{\kappa_1-\kappa_3}$ and $b=\frac{\kappa_1-\kappa_2}{\kappa_2-\kappa_3}$, which gives the intersection point at the origin $(0,0)$.

Then the time evolution of the intersection point for both cases is given by the following lemma.
\begin{lemma}
	The intersection point $(x_{0}(t),y_{0}(t))$ of those Y-solitons is given by
		\begin{eqnarray}\label{89}
			\left\{\begin{array}{ll} 
				\displaystyle{x_{0}(t)=-(\kappa_1\kappa_2+\kappa_1\kappa_3+\kappa_2\kappa_3)t},\\[1.0ex]
				\displaystyle{y_{0}(t)=(\kappa_1+\kappa_2+\kappa_3)t}.
			\end{array} \right.
		\end{eqnarray}
\end{lemma}

In the next section, we consider a perturbation problem under the assumption of an adiabatic change of the $\kappa$-parameters. 
\begin{example} Consider the Y-soliton with the parameters $(\kappa_1,\kappa_2,\kappa_3)=(-\frac{3}{2},\frac{1}{2},\frac{3}{2})$. Figure \ref{figY} shows the contour plot of the solution $u=2(\ln \tau)_{xx}$. The intersection point $C$ has the coordinates $(x_0=\frac{9}{4}t, y_0=\frac{1}{2}t)$.
\begin{figure}[htbp]
		\begin{minipage}[t]{1\linewidth}
			\centering
			\includegraphics[width=8.51cm,height=3cm]{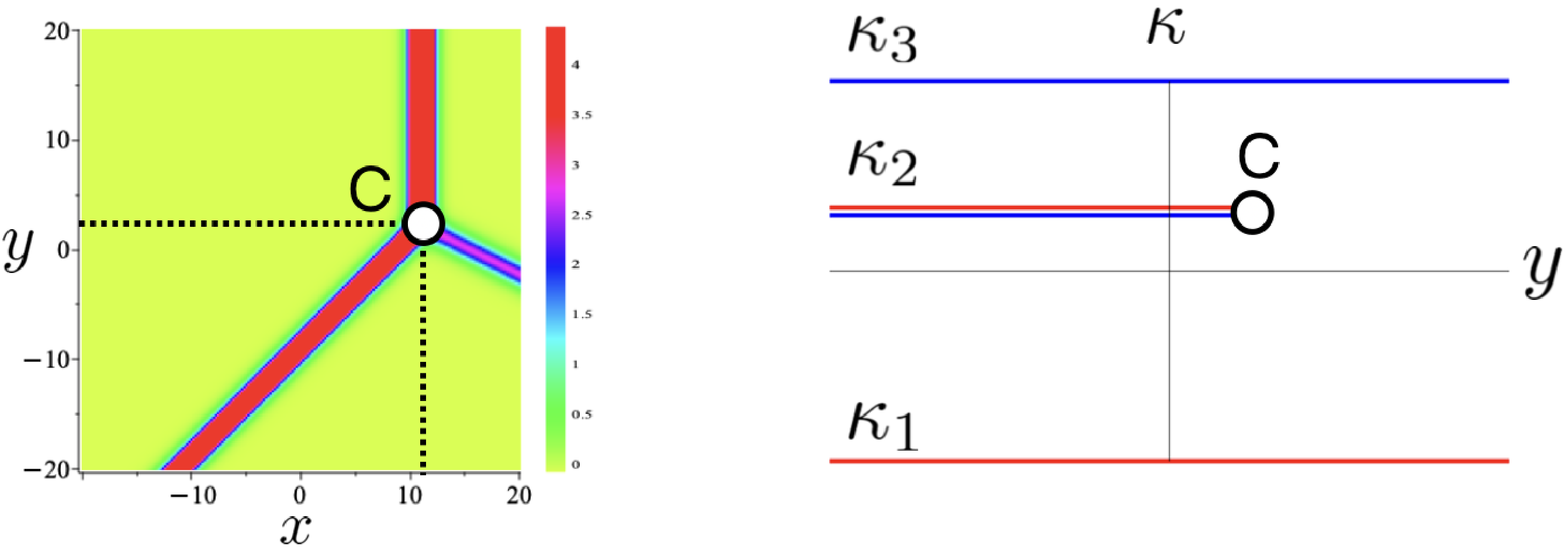}
		\end{minipage}
		\renewcommand\figurename{Figure}
		\caption{Y-soliton and the colored $\kappa$-graph with the singular point $C$.\label{figY}}
	\end{figure}
\end{example}

\section{The $\kappa$-system for soliton perturbations}\label{Sec:kappa}
The main purpose of our study is to describe the evolution of the KP soliton under certain classes of perturbations. Specifically, we assume an adiabatic change of the solution, given by the soliton $\kappa$-parameters, in the slow scales $(Y=\epsilon y, T=\epsilon t)$. {Within this framework, regular KP soliton solutions can be parametrized by chord diagrams (see Appendix \ref{A-KP}). Under perturbations, the soliton parameters evolve slowly in time, which induces an evolution of the corresponding incomplete chord diagram (see Section 6.3 in \cite{Kodama3}). We refer to this induced evolution as the \emph{chord dynamics}.}  The dynamical system of the $\kappa$-parameter has been derived in \cite{GKP:18, Ryskamp1} (also see Appendix \ref{A-kappa} for a simple derivation using a standard asymptotic perturbation theory), and it is given by \eqref{6}, i.e.,
\begin{equation}\label{kappa}
	\frac{\partial}{\partial T}
	\begin{pmatrix} \kappa_1\\ \kappa_2\end{pmatrix}~+~
	\begin{pmatrix} 2\kappa_1+\kappa_2 & 0\\ 0&\kappa_1+2\kappa_2\end{pmatrix}
	\frac{\partial}{\partial Y}\begin{pmatrix}\kappa_1\\ \kappa_2\end{pmatrix}
	~=~\begin{pmatrix}0\\0\end{pmatrix}.
\end{equation}
We remark that the system obtained in \cite{GKP:18, Ryskamp1} is given in terms of the amplitude
$A_{[1,2]}$ and the slope $\tan\Psi_{[1,2]}$, and note that the $\kappa$-parameter gives the Riemann invariants of the system. We call the system \eqref{kappa} ``$\kappa$-system'', {which is the main equation in this paper.}

{Most of our analysis is based on a simple wave case, that is, one of the parameters takes a constant value. In this case, we have the following lemma.}
\begin{lemma}\label{simple}
Assume that the system depends only on one parameter, say $\kappa_1$ and $\kappa_2=c=$constant, i.e.,
\begin{equation}\label{QLE1}
\frac{\partial\kappa_1}{\partial T}+(2\kappa_1+c)\frac{\partial \kappa_1}{\partial Y}=0.
\end{equation}
Then, if the initial data is monotonically increasing, the system has a global solution given 
in a hodograph form,
\begin{equation}\label{SQLE}
\kappa_1(Y,T)=f(Y-(2\kappa_1+c)T), \quad\text{for}\quad T>0,
\end{equation}
where $f(Y)$ is the initial data, i.e., $\kappa_1(Y,0)=f(Y)$.
\end{lemma}
\begin{proof}
The characteristic line of Eq. \eqref{QLE1} is given by
\[
\frac{dY}{dT}=2\kappa_1+c,\qquad \frac{d\kappa_1}{dT}=0.
\]
This implies that $\kappa_1$ is constant along the characteristic,
\[
Y=(2\kappa_1+c)T+Y_0, 
\]
where $Y_0$ is a constant. That is, we have $\kappa_1(Y,T)=f(Y_0)$ along the characteristic line $Y=(2f(Y_0)+c)T+Y_0$. Then the solution is given by \eqref{SQLE}, which is a rarefaction wave, i.e.,
\[
\kappa_1(Y,T)=\frac{Y-Y_0-cT}{2T}.
\]
This completes all the proofs.
\end{proof}
In the proof, one should note that the solution from \eqref{SQLE} is a general solution for \eqref{QLE1}. Then, taking the derivatives, we have
\[
\frac{\partial \kappa_1}{\partial Y}=\frac{f'(Y)}{1+2f'(Y)T},\quad\text{and}\quad
\frac{\partial\kappa_1}{\partial T}=\frac{-(2\kappa_1+c)f'(Y)}{1+2f'(Y)T},
\]
where $f'(Y)=\frac{df(Y)}{dY}$. This shows that if the initial data $f(Y)$ decreases, i.e., $f'(Y)<0$, in some region,
then the solution develops a shock wave at a finite time $T>0$, i.e., both derivatives of $\kappa_1$ become singular. Section \ref{sec:d} shows that the singularity corresponds to a resonant interaction of solitons, and it can be regularized by generating a soliton.

\subsection{The initial value problems for the {$\kappa$}-system}\label{sec:IVP}
{In this paper, we study the initial value problem of the $\kappa$-system with the following initial data}
	\begin{eqnarray}\label{45}
		u(x,y,0)=u^+_0(x,y)H(y)+u^-_0(x,y)H(-y),
	\end{eqnarray}
where $u_0^{\pm}(x,y)$ are some KP solitons (see \eqref{01}) at $t=0$, and $H(y)$ is the unit step function, 
$H(y)=1$ for $y>0$ and $H(y)=0$ for $y<0$. In particular, we consider the two cases (see \cite{Kao, Kodama3}) as shown in Figure \ref{fig96}.
 \begin{figure}[htbp]
  	\begin{minipage}[t]{1\linewidth}
  	\centering
  	\includegraphics[width=11.03cm,height=3cm]{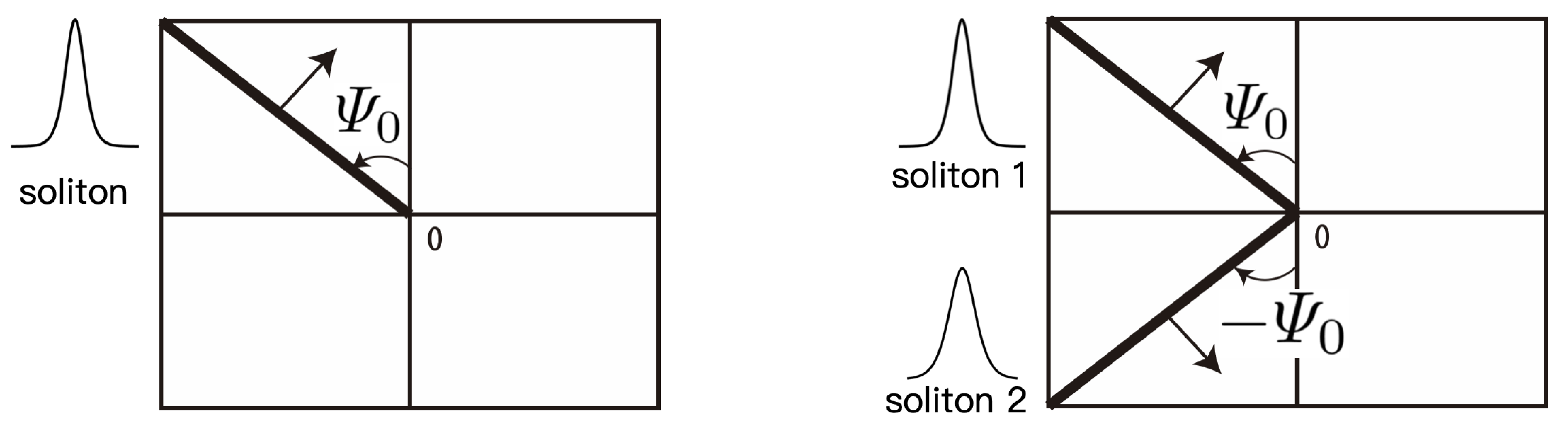}
  	\end{minipage}
  	\caption{The initial data \eqref{45}. Each bold face line shows a semi-infinite line-soliton (half-soliton). In the right panel, the amplitude of \emph{soliton 2} ($u_0^-$) is fixed to be 2, and that of \emph{soliton 1}
    ($u_0^+$) is a variable $A_0$. {The angle satisfies $-\tfrac{\pi}{2}<\Psi_{0}<\tfrac{\pi}{2}$.}} 
  	\label{fig96}
  \end{figure}
{The initial value problem \eqref{45} considered in this paper involves two types of initial data: the half-soliton initial data and the V-shaped initial data (formed by joining two different half-solitons), which are studied in Section 4 and Section 5, respectively. Although the half-soliton initial data can be regarded as a special case of the V-shaped initial data, i.e., when $u_0^{+}(x,y)=0$ or $u_0^{-}(x,y)=0$ in \eqref{45}, it serves as the fundamental building block for the analysis of the V-shaped case.}

{We also note that the ``bent soliton'' problem studied in \cite{Ryskamp1} (see (4.14) therein) requires $A_0=2$ and $0<\Psi_0<\tfrac{\pi}{2}$ in \eqref{45}. In contrast, the initial condition \eqref{45} allows more general V-shaped configurations, i.e., $A_0$ is arbitrary and $-\tfrac{\pi}{2}<\Psi_{0}<\tfrac{\pi}{2}$.}

{In the following section, we also perform numerical simulations of the KP equation \eqref{0}. These simulations were originally implemented in \cite{Kao}, using a Fourier discretization in $x$ together with a super-Gaussian window function to handle the non-periodic data in $y$. The same numerical scheme was also adopted in \cite{McDowell,Ryskamp1,Yuan2}.}

\section{The initial value problems with a half line-soliton initial data}\label{sec:H}
We first consider the initial data consisting of a half line-soliton for $y>0$,
\begin{equation}\label{ID-H}
u(x,y,0)=u_0(x,y)H(y),
\end{equation}
where $u_0(x,y)$ is a line-soliton with the parameters $\{\kappa_1^0,\kappa_2^0\}$, i.e.,
\begin{equation}\label{u0}
u_0(x,y)=A_0\sech^2\sqrt{\frac{A_0}{2}}(x+\tan\Psi_0 y)=\frac{(\kappa^0_1-\kappa^0_2)^2}{2}\sech^2\frac{\kappa^0_2-\kappa^0_1}{2}(x+(\kappa^0_1+\kappa^0_2)y).
\end{equation}
{This initial data \eqref{ID-H} corresponding to the case in the left panel of Figure \ref{fig96}, i.e., $u_0^{+}:=u_0$ and $u_0^{-}=0$ in \eqref{45}.} We show an example in Figure \ref{fig12}, in which we also show the initial
data for the $\kappa$-system and the ``incomplete'' chord diagram. What we mean by ``incomplete'' is that the half-soliton for $y>0$ represents just the upper part of the chord diagram of the ``full'' line-soliton corresponding to the permutation $\pi=(2,1)$. Also note in this figure, the coordinates are the slow scales $(X,Y)$.
\begin{figure}[H]
	\begin{minipage}[t]{1\linewidth}
		\centering
		\includegraphics[width=11.5cm,height=3.44cm]{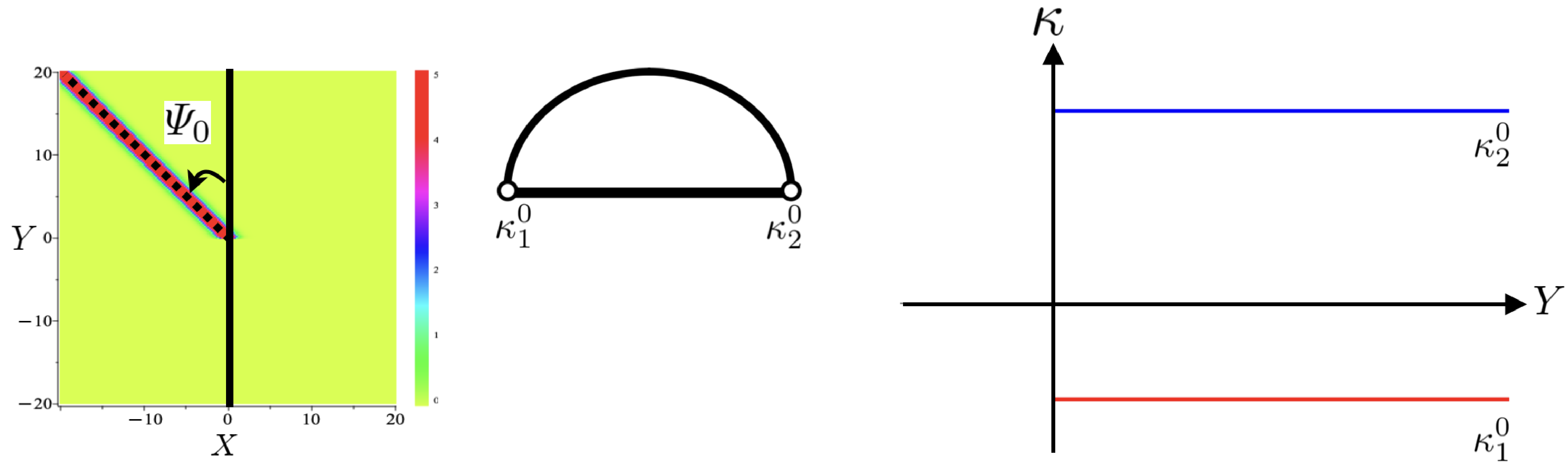}
	\end{minipage}%
	\caption{Example of a half-soliton initial data, incomplete chord diagram, and the $\kappa$-graph. In the right panel, the red line and the blue line represent the initial values of $\kappa_1$ and $\kappa_2$, i.e., $\kappa_1(Y,0)=\kappa_{1}^{0}$ and $\kappa_2(Y,0)=\kappa_{2}^{0}$ for $Y>0$.
	Here we take $\kappa_1^0=-1$ and $\kappa_2^0=2$, so that we have $A_{0}=\tfrac{9}{2}$ and $\tan\Psi_0=1$.}
	\label{fig12}
\end{figure} 

Extending the $\kappa$-graph for $Y<0$ as in Figure \ref{fig113}, we consider the initial value problem for $\kappa_2$ with $\kappa_1=\kappa_1^0=$constant, i.e.,
\begin{equation}\label{kappa2}
\frac{\partial \kappa_2}{\partial T}+(\kappa_1^0+2\kappa_2)\frac{\partial \kappa_2}{\partial Y}=0\qquad\text{with}\quad \kappa_2(Y,0)=
\left\{\begin{array}{ll}
\kappa_1^0,&\quad Y<0,\\
\kappa_2^0,&\quad Y>0.
\end{array}\right.
\end{equation}
\vskip -0.2cm
\begin{figure}[htbp]
	\begin{minipage}[htb]{1\linewidth}
	\centering
	\includegraphics[width=5.03cm,height=3cm]{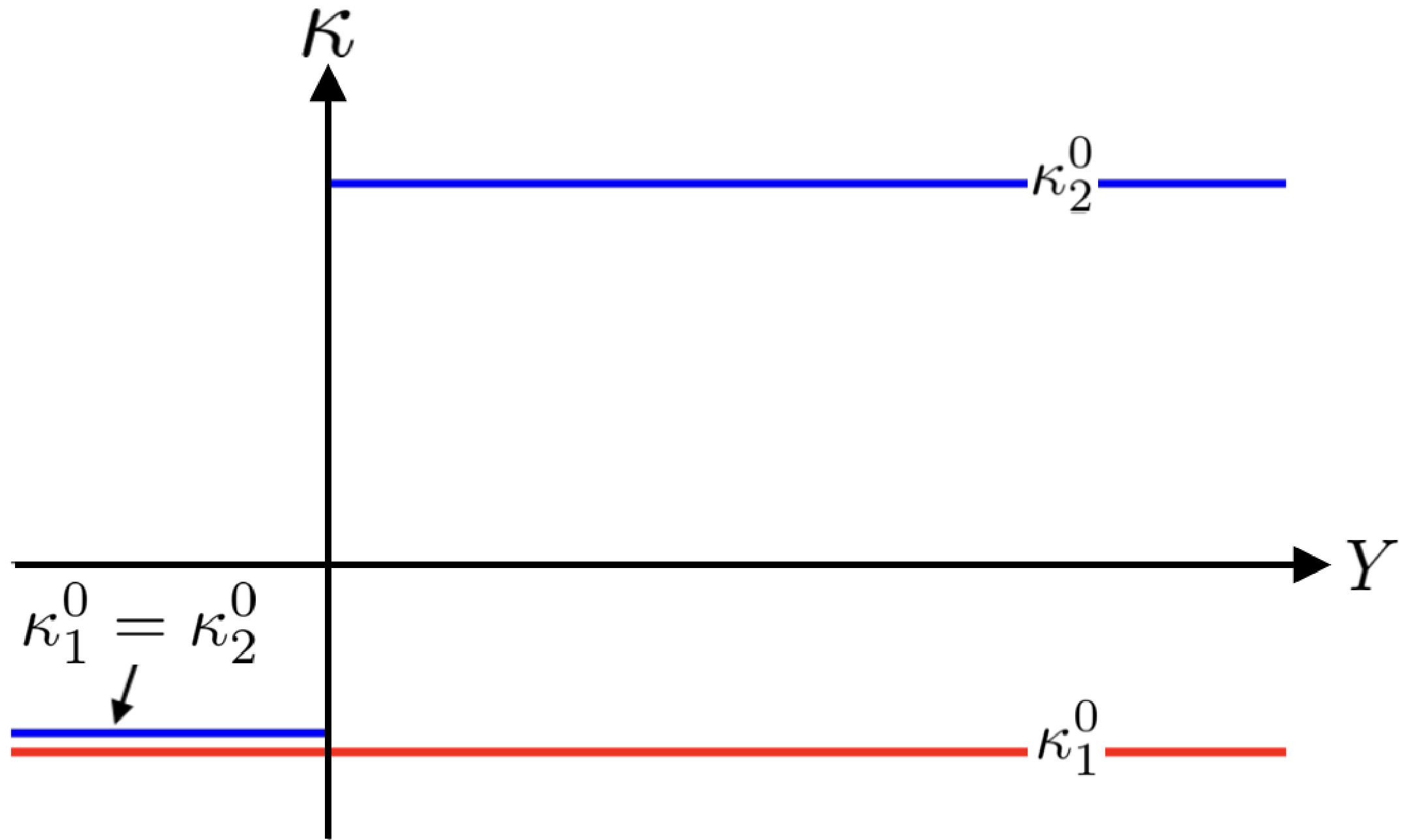}
	\end{minipage}%
 \caption {The Initial data corresponding to the half-line soliton in Figure \ref{fig12}.}
    \label{fig113}
\end{figure} 
{We remark that the extension should be consistent with the given initial data, and that the initial value problem with the extended initial data should be well-posed (admitting a global solution for $T>0$).} Indeed, we have the following proposition.
\begin{proposition}\label{p2}
	The initial value problem of the $\kappa$-system \eqref{kappa2} has a unique global solution
\begin{eqnarray}\label{12}
\kappa_{2}(Y,T)=
\left\{\begin{array}{ll}
\kappa_{1}^{0},\qquad &\text{for}\quad Y<Y_{b}(T),\\[1.0ex]
\displaystyle{\kappa_{1}^{0}+\frac{\kappa_{2}^{0}-\kappa_{1}^{0}}{Y_{a}(T)-Y_{b}(T)}(Y-Y_{b}(T))},\qquad &\text{for}\quad Y_{b}(T)<Y<Y_{a}(T),\\[2.0ex]
\kappa_{2}^{0},\qquad &\text{for}\quad Y_{a}(T)<Y,
\end{array} \right.
\end{eqnarray}
where $Y_{a}(T)=(\kappa_{1}^{0}+2\kappa_{2}^{0})T$ and $Y_{b}(T)=3\kappa_{1}^{0}T$  (note that $Y_a-Y_b=2(\kappa_2^0-\kappa_1^0)T>0$ for $T>0$).
\end{proposition}
\begin{proof}
As shown in Lemma \ref{simple}, the initial value problem has a global solution,
\[
\kappa_2(Y,T)=f(Y-(\kappa_1^0+2\kappa_2)T).
\]
Since $f(0-):=\lim_{Y\uparrow 0}f(Y)=\kappa_1^0$ and $f(0+):=\lim_{Y\downarrow 0}f(Y)=\kappa_2^0$, we have
\[
Y=Y_b(T)=3\kappa_1^0T,\quad\text{and}\quad Y=Y_a(T)=(\kappa_1^0+2\kappa_2^0)T.
\]
The solution above implies that $\kappa_2$ in the region $Y_b(T)<Y<Y_a(T)$ is linear in $Y$ for fixed $T$ (a rarefaction wave, see Lemma \ref{simple}). Then using the boundary conditions $\kappa_2(Y_b,T)=\kappa_1^0$
and $\kappa_2(Y_a,T)=\kappa_2^0$, we have the result.
\end{proof}

Now we compute the peak trajectory of the perturbed soliton in the $XY$-plane using \eqref{dX},
i.e.,
\[
\frac{\partial X}{\partial Y}=-\tan\Psi_{[1,2]}=-(\kappa_1+\kappa_2).
\]
Integrating this equation for fixed $T$, we have
	\begin{align}\label{1423}
	X(Y,T)&=\displaystyle\int_{Y}^{Y_{a}(T)}(\kappa_{1}^{0}+\kappa_{2}(\eta,T))d\eta+X_{a}(T)\\
	&=-\frac{1}{4T}(Y-Y_a(T))^2-\frac{\kappa_1^0}{2}(Y-Y_a(T))+X_a(T),\nonumber
	\end{align}
where $X_a(T)$ is the edge of the half-soliton at $Y=Y_a(T)$, i.e., from \eqref{Trajectory},
\begin{align*}
X_{a}(T)&=-\tan\Psi_{[1,2]}^0Y_a(T)+C^0_{[1,2]}T\\
&=-(\kappa_{1}^{0}+\kappa_{2}^{0})Y_{a}(T)+((\kappa_{1}^{0})^{2}+\kappa_{1}^{0}\kappa_{2}^{0}+(\kappa_{2}^{0})^2)T. 
\end{align*}
Then from Proposition \ref{p2}, we obtain
	\begin{eqnarray}\label{13}
X=\left\{\begin{array}{ll}
\displaystyle{-\frac{1}{4T}(Y+\kappa_1^0T)^2+(\kappa_1^0)^2T},\qquad & \text{for}\quad Y_{b} < Y < Y_{a},\\[2.0ex]
\displaystyle{-\tan\Psi^0_{[1,2]}Y+C^0_{[1,2]}T},\quad &\text{for}\quad Y_{a}<Y.\\
	\end{array} \right.
	\end{eqnarray}
Note here that there is no trajectory in the region $Y<Y_b$, since the amplitude of the soliton is zero in this region (i.e., $\kappa_2=\kappa_1^0$). Thus, the peak trajectory forms a parabola in the region $Y_a<Y<Y_b$, and we note that the latus rectum increases in time (i.e., the opening is getting wider in time), and  the position of the parabola depends only on $\kappa_1^0$. We call this parabolic part of a ``quasi''-soliton \emph{parabolic}-soliton, and describe it as \emph{parabolic $[1]$-soliton} or simply \emph{$[1]$-soliton} {emphasizing that it depends only on one parameter $\kappa_{1}^{0}$}. 
In the table below, we show the soliton structure, which consists of line-solitons and parabolic-soliton.
		\begin{table}[h]
			\centering
			\begin{tabular}{c|c|c|c}
				\hline
				{Interval}  & $(-\infty, Y_{b})$ & $(Y_{b}, Y_{a})$ & $(Y_{a},+\infty)$ \\  
				\hline 
				Line-soliton & & & ${[1,2]}$ \\ 
				\hline  
				Parabolic-soliton &  & ${[1]}$  & \\   
				\hline
			\end{tabular}
		\end{table}
\noindent
As we will show, in general the solution $u(x,y,t)$ consists of segments of line-solitons and parabolic-solitons. Figure \ref{fig14} shows the results of a numerical simulation of the KP equation, which are in good agreement with the results of \eqref{13}.
\begin{figure}[htbp]
	\begin{minipage}[htp]{1\linewidth}
		\centering
		\includegraphics[width=14.77cm,height=3.5cm]{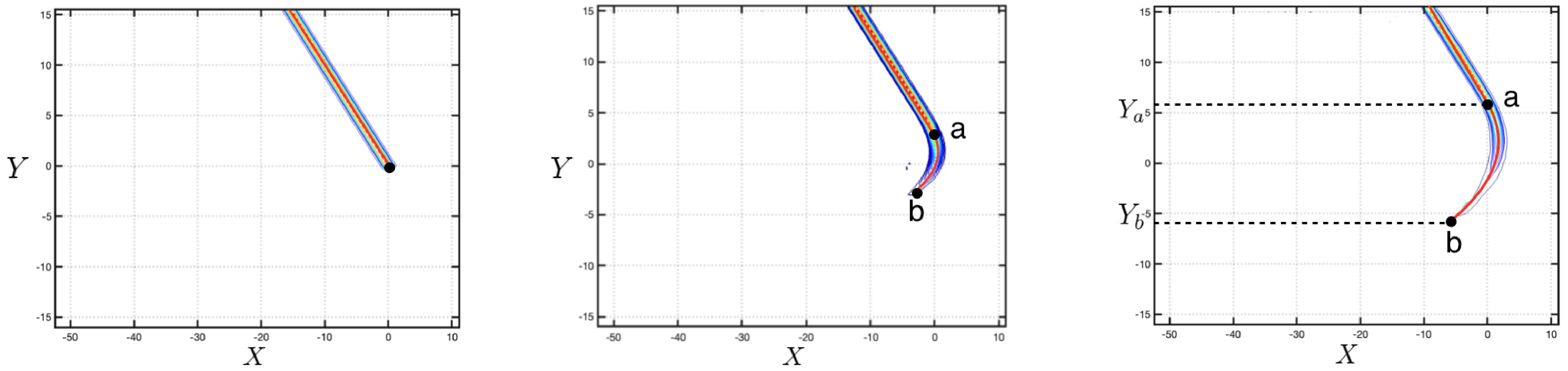}
	\end{minipage}%
	\caption{Numerical simulation and theoretical prediction: the contour plots of the numerical simulation, and the red line represents the theoretical result (\ref{13}) for  $T=0$, $T=1$, and $T=2$. {Here a and b represent the points $(X_{a},Y_{a})$ and $(X_{b},Y_{b})$, respectively.}}
	\label{fig14}
\end{figure} 

In a similar way, we can solve the initial value problem with a half-line initial soliton for $y<0$, i.e.,
the initial data is given by $u(x,y,0)=u_0(x,y)H(-y)$ with $u_0$ in \eqref{u0}.
The initial $\kappa$-parameters are given by
\begin{eqnarray}\label{10}
			\kappa_{1}=\left\{\begin{array}{ll}
				\kappa^{0}_{1},\quad\text{for}\quad &\quad Y<0,\\[1.0ex]
				\kappa^{0}_{2},\quad \text{for}\quad  &\quad Y>0,
			\end{array} \right.\qquad\text{and}\qquad 
			\kappa_{2}=\kappa^{0}_{2},\quad\text{for}\quad  Y\in\mathbb{R}.
\end{eqnarray}
We show in Figure \ref{fig84} an example of the initial data and the corresponding $\kappa$-graphs.
\begin{figure}[H]
	\begin{minipage}[t]{1\linewidth}
		\centering
		\includegraphics[width=11.9cm,height=3.5cm]{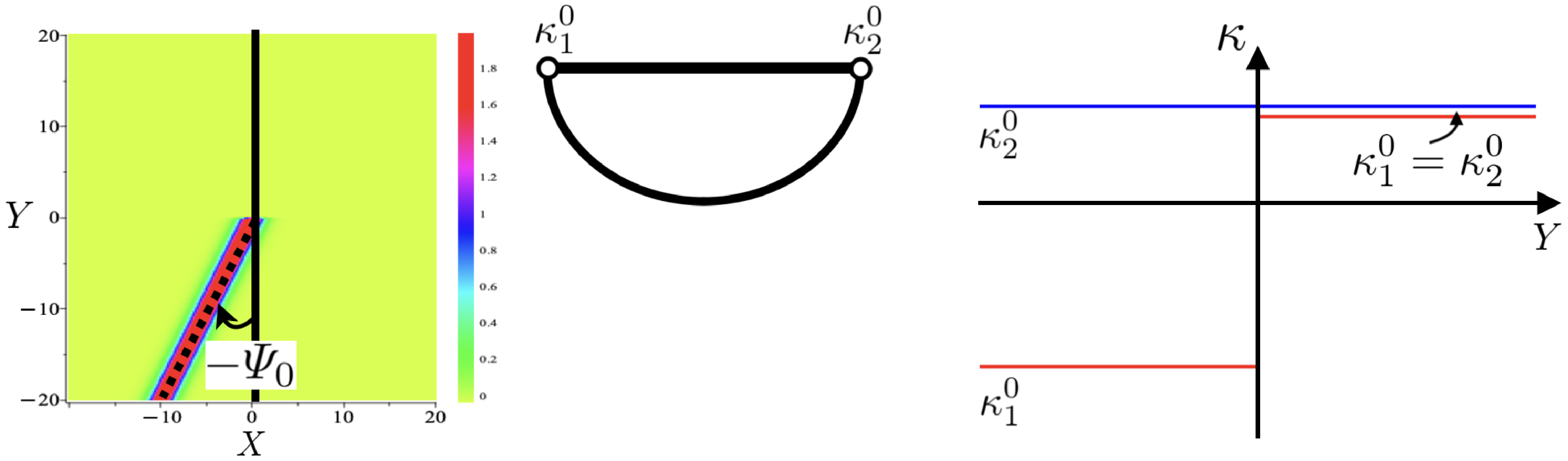}
	\end{minipage}
	\caption{Example of lower semi $(1,2)$-soliton solution, incomplete chord diagram, and the $\kappa$-graph.}
	\label{fig84}
\end{figure}

Figure \ref{fig15} shows the numerical simulation for the KP equation with the initial data \eqref{10} and the solution of the $\kappa$-system with \eqref{10}. Using \eqref{dX}, we can obtain the peak trajectory in
a similar way as before, which is
	\begin{eqnarray}\label{S-HL}
		X=
		\left\{\begin{array}{ll}
\displaystyle{-\frac{1}{4T}(Y+\kappa_2^0T)^2+(\kappa_2^0)^2T},\qquad & \text{for}\quad Y_{a} < Y < Y_{b},\\[2.0ex]
\displaystyle{-\tan\Psi_{[1,2]}^0Y+C^0_{[1,2]}T},\quad & \text{for}\quad Y<Y_{a},\\
		\end{array} \right.
	\end{eqnarray}
where $Y_a=(2\kappa_1^0+\kappa_2^0)T$ and $Y_b=3\kappa_2^0T$. Thus we have a parabolic $[2]$-soliton in the region $Y_a<Y<Y_b$ depending only on $\kappa_2^0$. Moreover, there is no soliton in the region above $Y=Y_b$.
\begin{figure}[htbp]
	\begin{minipage}[htp]{1\linewidth}
		\centering
		\includegraphics[height=3.5cm,width=12.29cm]{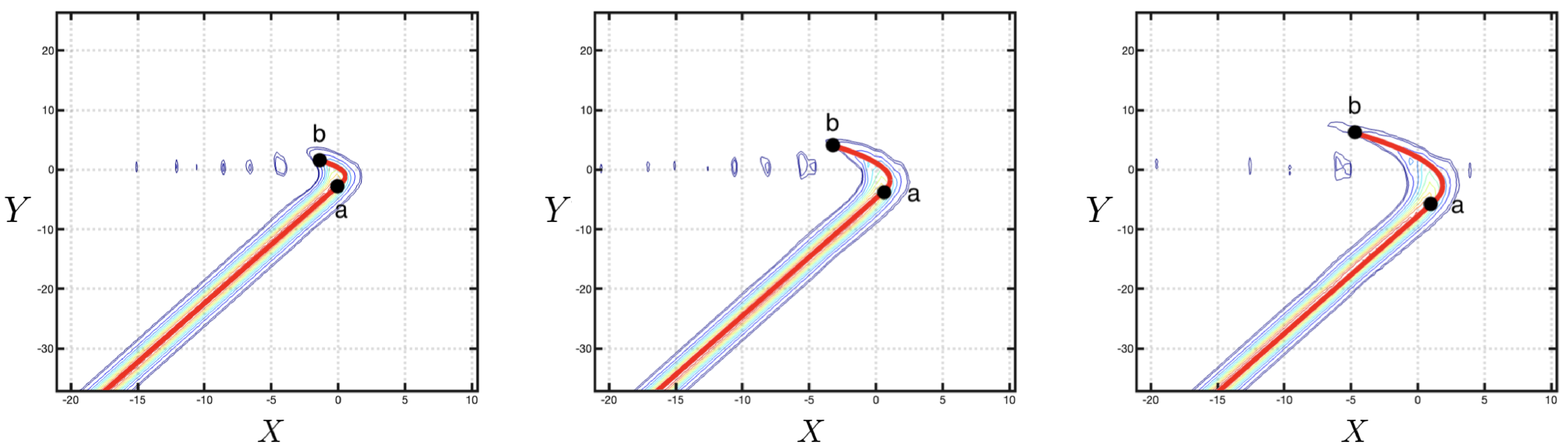}
	\end{minipage}%
	\caption{Numerical simulation and theoretical comparison: the main part is the numerical simulation results, and the red line represents the peak trajectory function. We take $(\kappa_1^0,\kappa_2^0)=(-\frac{5}{4},\frac{3}{4})$,
	and the figures are taken at $T=1,2$ and $3$.}
	\label{fig15}
\end{figure}
{Note that these results are obtained by $\kappa_{1}^{0}\leftrightarrow \kappa_{2}^{0}$ in the previous case as the half-soliton for $Y>0$.}

Before closing this section, we give the following definition for the initial $\kappa$-parameters.
\begin{definition}\label{def:fix}
For the initial $\kappa$-parameters, we define the following:
\begin{itemize}
\item[(a)]
An initial parameter $\kappa_i^0$ is a \emph{fixed} point, if $\kappa_i=\kappa_i^0=$constant
for all $T>0$.
\item[(b)]
An initial parameter $\kappa_i^0$ is a \emph{free} point, if $\kappa_i$ changes in $Y$ for $T>0$.
\end{itemize}
\end{definition}
As will be shown in the following sections, the notion of ``fixed'' and ``free'' will be useful to
describe the evolution of the chord diagram. In particular, we note that a parabola in the peak trajectory of the adiabatic soliton depends only on the fixed point.  

\begin{example} Consider an example of half-soliton with $Y<0$. Figure \ref{fig888} illustrates
the fixed and free points in the initial data.
	\begin{figure}[htbp]
		\begin{minipage}[htp]{1\linewidth}
			\centering
			\includegraphics[width=8.48cm,height=3.5cm]{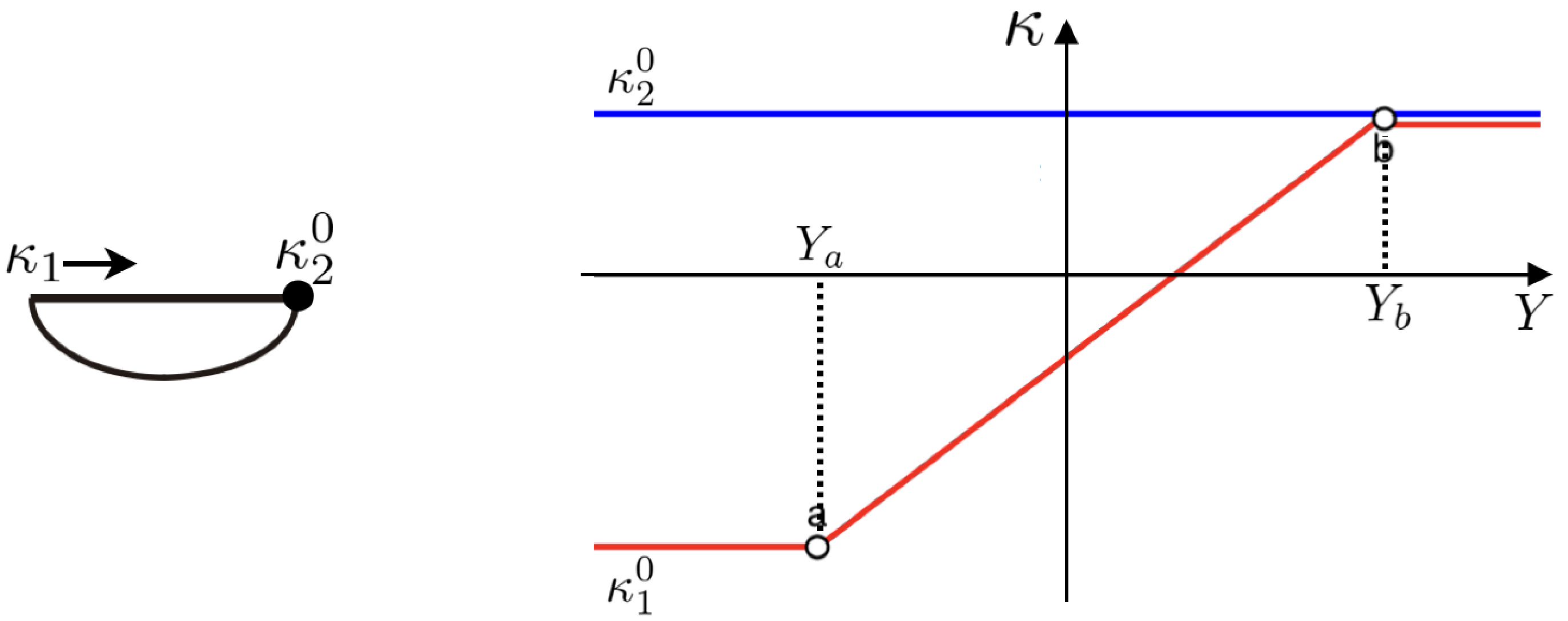}
		\end{minipage}
		\caption{The $\kappa$-graph for an initial half-soliton in $Y<0$. The left panel shows the incomplete chord diagram for the initial value problem, where the $\kappa_2=\kappa_{2}^{0}$ marked by $\bullet$ represents the fixed point. The right panel shows the $\kappa$-graph at $T>0$.} \label{fig888}
	\end{figure}
\end{example}
\section{The initial value problems with V-shaped initial data}\label{Sec:IVP-V}
In this section, {we consider the initial data \eqref{45} corresponds to the right panel of Figure \ref{fig96}}, i.e., $u(x,y,0)=u_0^+(x,y)H(y)+u_0^-(x,y)H(-y)$,
with
\begin{eqnarray}\label{5.1}
	\left\{\begin{array}{ll}
			u^{+}_{0}(x,y)=A_{0}\sech^2 \big(\sqrt{\frac{A_{0}}{2}}(x-y\tan\varPsi_{0})\big),\\[1.5ex]
			u^{-}_{0}(x,y)=2\sech^2(x+y\tan\varPsi_{0}),
		\end{array} \right.
	\end{eqnarray}
where $A_0$ and $\tan\Psi_0$ are free parameters  \cite{Kao, Kodama3}.

Using the parameters $\sqrt{2A_0}$ and $\tan\Psi_0$ {in \eqref{5.1}}, we consider all possibles of the V-shaped initial data as shown in Figure \ref{fig18} (see Chapter 6 in \cite{Kodama3}, also \cite{J:14,Kao}). In this figure,  each region is parametrized by an \emph{incomplete} chord diagram, in which the upper chord represents a half line-soliton in $Y>0$, and the lower chord represents another half line-soliton in $Y<0$. We label the edge points of the chords with
$(\kappa^0_1,\kappa^0_2,\kappa^0_3,\kappa^0_4)$. The solid lines show the cases with $\kappa^0_1=\kappa^0_2$ and $\kappa^0_3=\kappa^0_4$. Note that the case with $\kappa_1^0=\kappa_2^0$ gives the line $2\tan\Psi_0+\sqrt{2A_0}=2$, and the case $\kappa_3^0=\kappa_4^0$ gives $-2\tan\Psi_0+\sqrt{2A_0}=2$ as shown in Figure \ref{fig18}.
Also note that the dashed lines are $\kappa^0_2=\kappa^0_3$, which do not correspond to any permutation, i.e., there is no corresponding soliton solution.
\begin{figure}[H]
	\begin{minipage}[htbp]{1\linewidth}
		\centering
		\includegraphics[width=10.46cm,height=5cm]{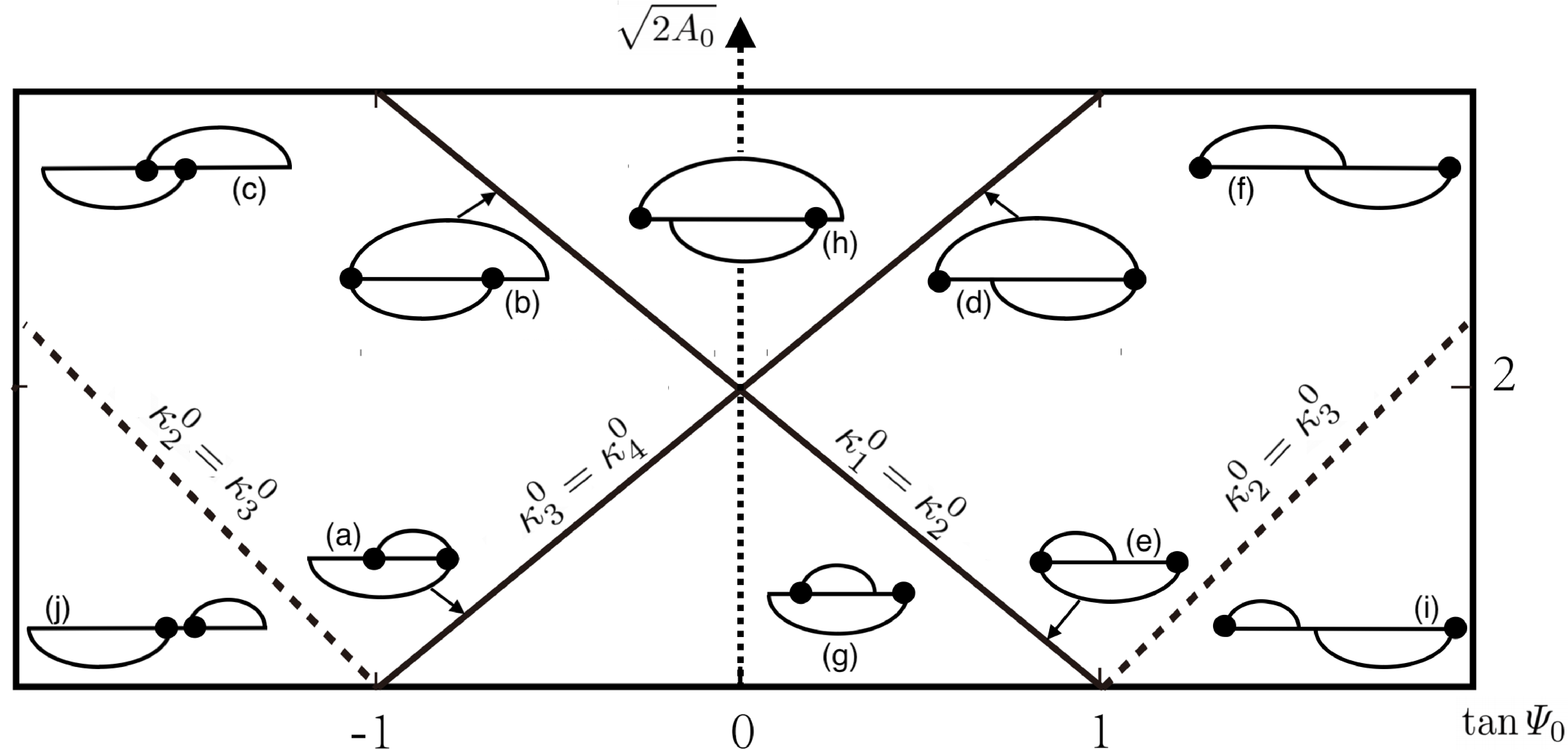}
	\end{minipage}
	\caption{All possible $V$-shaped initial data. Each region is parametrized by a unique incomplete chord diagram. Each $\bullet$ in the diagram marks the fixed point. The length (amplitude) of the lower chord is fixed as 2.} 
	\label{fig18}
\end{figure}

{The table below summarizes the relation between the parameters in \eqref{5.1} and Figure \ref{fig18}. In specific examples, it suffices that $(\kappa_1^0,\kappa_2^0,\kappa_3^0,\kappa_4^0)$ satisfy the parameter relations listed in the table; any values meeting these conditions may be chosen. In the table below, $[i,j]$ denotes the soliton determined by the spectral parameters $\kappa_i^0$ and $\kappa_j^0$. The corresponding information for $u_0^{+}$ and $u_0^{-}$ can be obtained from \eqref{01}, e.g., for the case (c), we set $u_0^{+}=u_{[2,4]}$ and $u_0^{-}=u_{[1,3]}$, where $2=\frac{(\kappa_{1}^0-\kappa_{3}^0)^2}{2}$, $A_{0}=\frac{(\kappa_{2}^0-\kappa_{4}^0)^2}{2}$, and $\tan \Psi_{0}=\kappa_{2}^{0}+\kappa_{4}^{0}=-(\kappa_{1}^{0}+\kappa_{3}^{0})$. In addition, the notes indicate the cases in which two endpoints coincide in the incomplete chord diagram.}
\vskip -0.5cm
{\begin{table}[H]
		\centering
		\begin{tabular}{c|c|c|c|c|c|c|c|c|c|c}
			\hline
			{Case}  & (a) & (b) &(c) &(d) & (e) & (f) & (g) & (h) & (i) & (j) \\  
			\hline 
			$u^{+}_{0}$ & $[2,3]$ & $[1,4]$ & $[2,4]$ & $[1,3]$& $[1,3]$& $[1,3]$ & $[2,3]$ & $[1,4]$ & $[1,2]$ &  $[3,4]$\\ 
			\hline  
			$u^{-}_{0}$ &$[1,3]$& $[1,3]$ & $[1,3]$ & $[2,3]$ & $[2,4]$ & $[2,4]$  & $[1,4]$ & $[2,3]$ & $[3,4]$ & $[1,2]$ \\ 
			\hline  
			{Note}  & $\kappa_3^0=\kappa_4^0$ & $\kappa^0_1=\kappa_2^0$ &  &$\kappa^0_3=\kappa_4^0$ & $\kappa^0_1=\kappa_2^0$ &   &  &  &  &  \\ 
			\hline
		\end{tabular}
\end{table}}

These cases in Figure \ref{fig18} have been \emph{numerically studied} in \cite{Kao}, and 
the authors predict a convergence to some exact soliton solution. The main purpose of the present paper is to study \emph{analytically} these examples, and show the convergence of the solution 
using the $\kappa$-system. In the following sections, we study all the cases in Figure \ref{fig18},
which are labeled in (a) through (j), and find the asymptotic solutions for all the cases.

\subsection{The cases (a) and (b)}
The case (a) corresponds to the case with $\kappa_3^0=\kappa_4^0$ and $\sqrt{2A_0}>2$, in which the half-line solitons
are $[2,3]$-soliton in $Y>0$ and $[1,3]$-soliton in $Y<0$. The initial data for the $\kappa$-system 
\eqref{kappa} is given by (see Figure \ref{fig251})
\begin{equation}\label{a}
\kappa_1(Y,0)=\left\{
\begin{array}{ll}
\kappa_1^0,\quad Y<0,\\
\kappa_2^0,\quad Y>0,
\end{array}\right.
\qquad\text{and}\qquad
\kappa_2(Y,0)=\kappa_3^0.
\end{equation}
\begin{figure}[htbp]
	\begin{minipage}[htp]{1\linewidth}
		\centering
		\includegraphics[height=3cm,width=8.4cm]{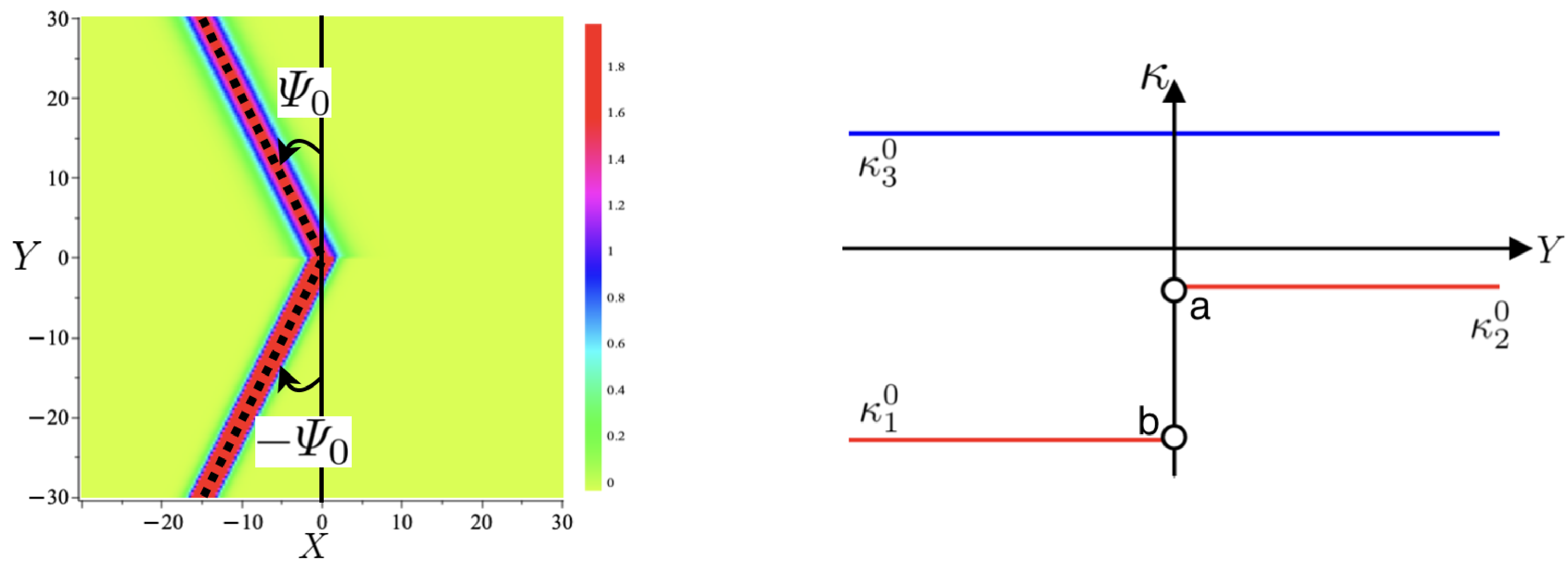}
	\end{minipage}
	\caption {$V$-shaped soliton of the case (a).}
	\label{fig251}
\end{figure}

Note that $\kappa_3^0$ is a fixed point (i.e., constant for all $Y$ at $T>0$). The system gives a simple wave,
which depends only on $\kappa_1(Y,T)$,
\[
\frac{\partial \kappa_1}{\partial T}+(2\kappa_1+\kappa_3^0)\frac{\partial\kappa_1}{\partial Y}=0.
\]
The characteristic velocity is then given by $V(\kappa_1,\kappa_2)=2\kappa_1+\kappa_3^0$.
Since the initial data of $\kappa_1$ increases in $Y$, we have a global solution,
\begin{equation}\label{Sa}
\kappa_1(Y,T)=\left\{\begin{array}{lll}
\kappa_1^0,\qquad &Y<Y_b(T),\\
\displaystyle{\kappa_1^0+\frac{\kappa_2^0-\kappa_1^0}{Y_a-Y_b}(Y-Y_b)},\qquad & Y_b(T)<Y<Y_a(T),\\
\kappa_2^0,\qquad &Y>Y_a(T),
\end{array}\right.
\end{equation}
where $Y_a(T)=(2\kappa^0_2+\kappa_3^0)T$ and $Y_b(T)=(2\kappa^0_1+\kappa_3^0)T$.
Note that $Y_a-Y_b=2(\kappa^0_2-\kappa_1^0)T$, that is, the slope of the $\kappa_1$ in the region $Y_b<Y<Y_a$ is $(2T)^{-1}$, and  $\kappa_1(Y,T)$ in the region $(Y_b,Y_a)$ is expressed by the form depending only on the fixed point $\kappa_3^0$,
\[
\kappa_1(Y,T)=\frac{1}{2T}(Y-\kappa_3^0T).
\]

The peak trajectory can be computed by integrating the slope equation $\frac{\partial X}{\partial Y}=-\tan\Psi_{[i,j]}$ as shown in the previous section, and we obtain
\begin{equation}\label{Ta}
X(Y)=\left\{
\begin{array}{lll}
-\tan\Psi_{[1,3]}^0Y+C_{[1,3]}^0T,\qquad & Y<Y_b(T),\\[1.0ex]
\displaystyle{-\frac{1}{4T}(Y+\kappa_3^0T)^2+(\kappa_3^0)^2T},\qquad & Y_b(T)<Y<Y_a(T),\\[2.0ex]
-\tan\Psi_{[2,3]}^0 Y+C^0_{[2,3]}T,\qquad & Y>Y_a(T).
\end{array}
\right.
\end{equation}
Thus, the half-solitons in $Y>Y_a$ and $Y<Y_b$ are connected through the parabolic-soliton depending  only on $\kappa_3^0$, that is, $[1,3]$-soliton and $[2,3]$-soliton are connected by $[3]$-soliton, that is, we have the following table:
		\begin{table}[h]
			\centering
			\begin{tabular}{c|c|c|c}
				\hline
				{Interval}  & $(-\infty, Y_{b})$ & $(Y_{b}, Y_{a})$ & $(Y_{a},+\infty)$ \\  
				\hline 
				Line-soliton   & $[1,3]$& & ${[2,3]}$ \\ 
				\hline  
				Parabolic-soliton &  & ${[3]}$  & \\   
				\hline
			\end{tabular}
		\end{table}

\noindent
The numerical simulation with the peak trajectory is shown in Figure \ref{fig26}.
\begin{figure}[htbp]
	\begin{minipage}[htb]{1\linewidth}
		\centering
		\includegraphics[height=3cm,width=10.36cm]{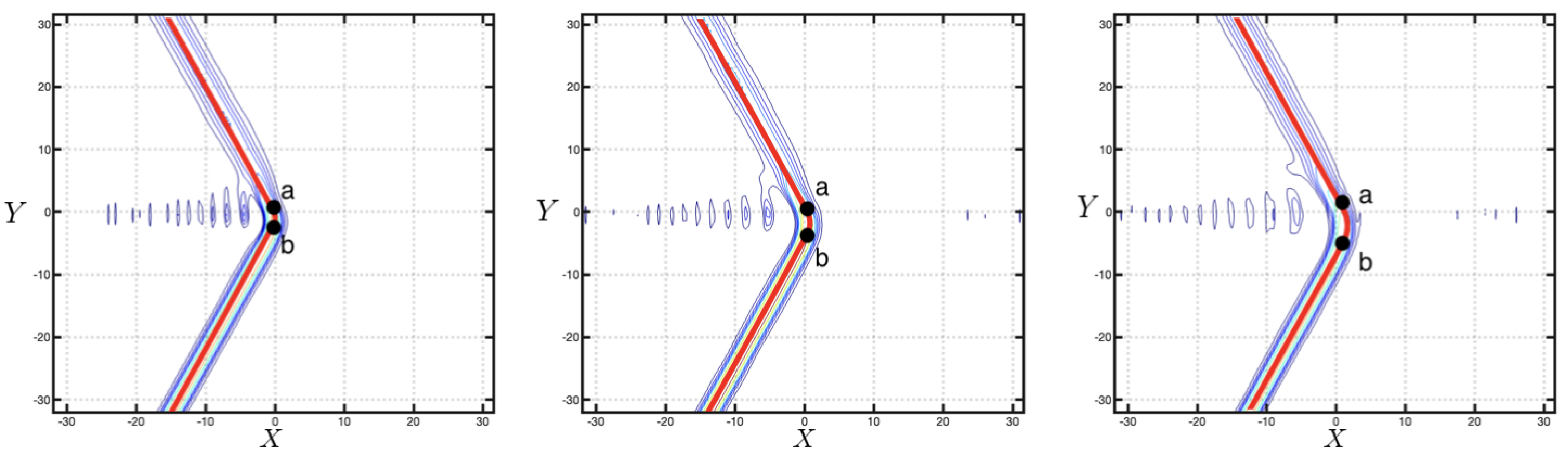}
	\end{minipage}%
	\caption {Numerical simulation for the case (a).  We take $(\kappa_1^0,\kappa_2^0,\kappa_3^0)=(-\frac{5}{4},-\frac{1}{4},\frac{3}{4})$,
	and the figures are taken at $T=1,2$ and $3$. The peak trajectories are shown as the solid curves, and the curve between points a and b is a parabola connecting upper and lower solitons.
	 }
	\label{fig26}
\end{figure}

\bigskip

 The case (b) corresponds to that of $\kappa^0_1=\kappa_2^0$ and $\sqrt{2A_0}>2$. The line-soliton of V-shape initial wave are $[1,4]$-soliton in $Y>0$ and $[1,3]$-soliton in $Y<0$. 
  We take the initial data of this case as
	\begin{eqnarray*}
			\kappa_{1}=\kappa^{0}_{1},\quad\text{for}\quad Y\in\mathbb{R},\qquad\text{and}\qquad 
			\kappa_{2}=\left\{\begin{array}{ll}
				\kappa^{0}_{3}, \quad&\text{for}\quad Y<0,\\
				\kappa^{0}_{4},&\text{for}\quad Y>0.
			\end{array} \right.
	\end{eqnarray*}
	Since $\kappa_1=\kappa_1^0$ for all $Y$, the $\kappa$-system gives a simple wave solution.
Note that the initial data $\kappa_{2}$ increases in $Y$, and the characteristic speed is $V_{2}(0+)>V_{2}(0-)$ form $\kappa_{3}^{0}<\kappa_{4}^{0}$ (i.e., $V_{2}=\kappa_1^0+2\kappa_2$). From Lemma \ref{simple}, the $\kappa$-system admits a global solution with the initial data. 
This is similar to the case ($a$), and we omit the details of the initial value problem.

\subsection{The case (c)}
The V-shape initial waves are $[2,4]$-soliton in $Y>0$ and $[1,3]$-soliton in $Y<0$ (see Figure \ref{fig132}).
\begin{figure}[htbp]
	\begin{minipage}[htp]{1\linewidth}
		\centering
		\includegraphics[height=3.5cm,width=9.71cm]{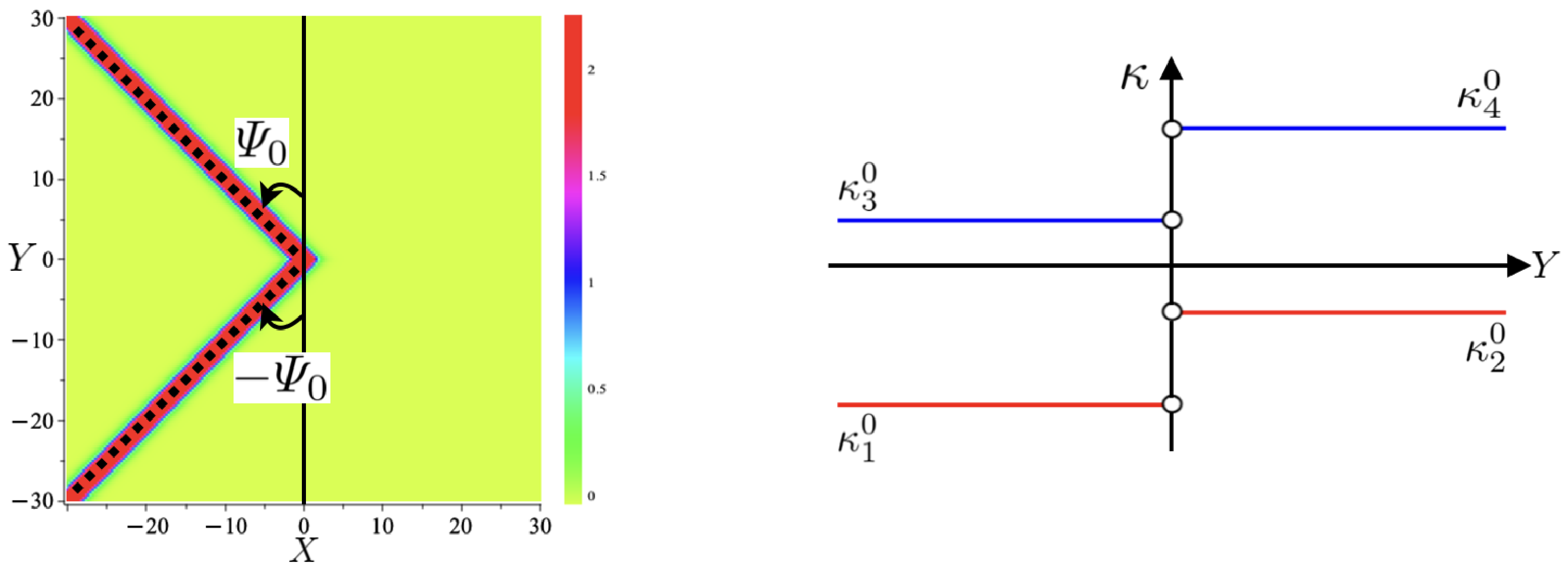}
	\end{minipage}
	\caption{The $V$-shaped initial data of the case (c). The right panel is the
	$\kappa$-graph. Here we take $(\kappa_{1}^{0},\kappa_{2}^{0},\kappa_{3}^{0},\kappa_{4}^{0})=(-\frac{3}{2},-\frac{1}{2},\frac{1}{2},\frac{3}{2})$.}
	\label{fig132}
\end{figure}
Note first that the initial data for the $\kappa$-system shown in Figure \ref{fig132} is not well-defined at $Y=0$. Then we consider the following \emph{regularization} of the initial data with a parameter
$0<\varepsilon\ll1$,
\begin{equation}\label{R-IDc}
\kappa_1=\left\{ \begin{array}{lll}
\kappa_1^0,\quad &\text{for}\quad  Y<-\varepsilon,\\
\kappa_2^0, &\text{for}\quad Y>-\varepsilon,
\end{array}\right.
\qquad \text{and}\qquad \kappa_2=\left\{ \begin{array}{lll}
\kappa_3^0,\quad &\text{for}\quad  Y<\varepsilon,\\
\kappa_4^0, &\text{for}\quad Y>\varepsilon,
\end{array}\right.
\end{equation}
(see Figure \ref{fig312}).
\begin{figure}[htbp]
	\begin{minipage}[htp]{1\linewidth}
		\centering
		\includegraphics[height=3cm,width=5.31cm]{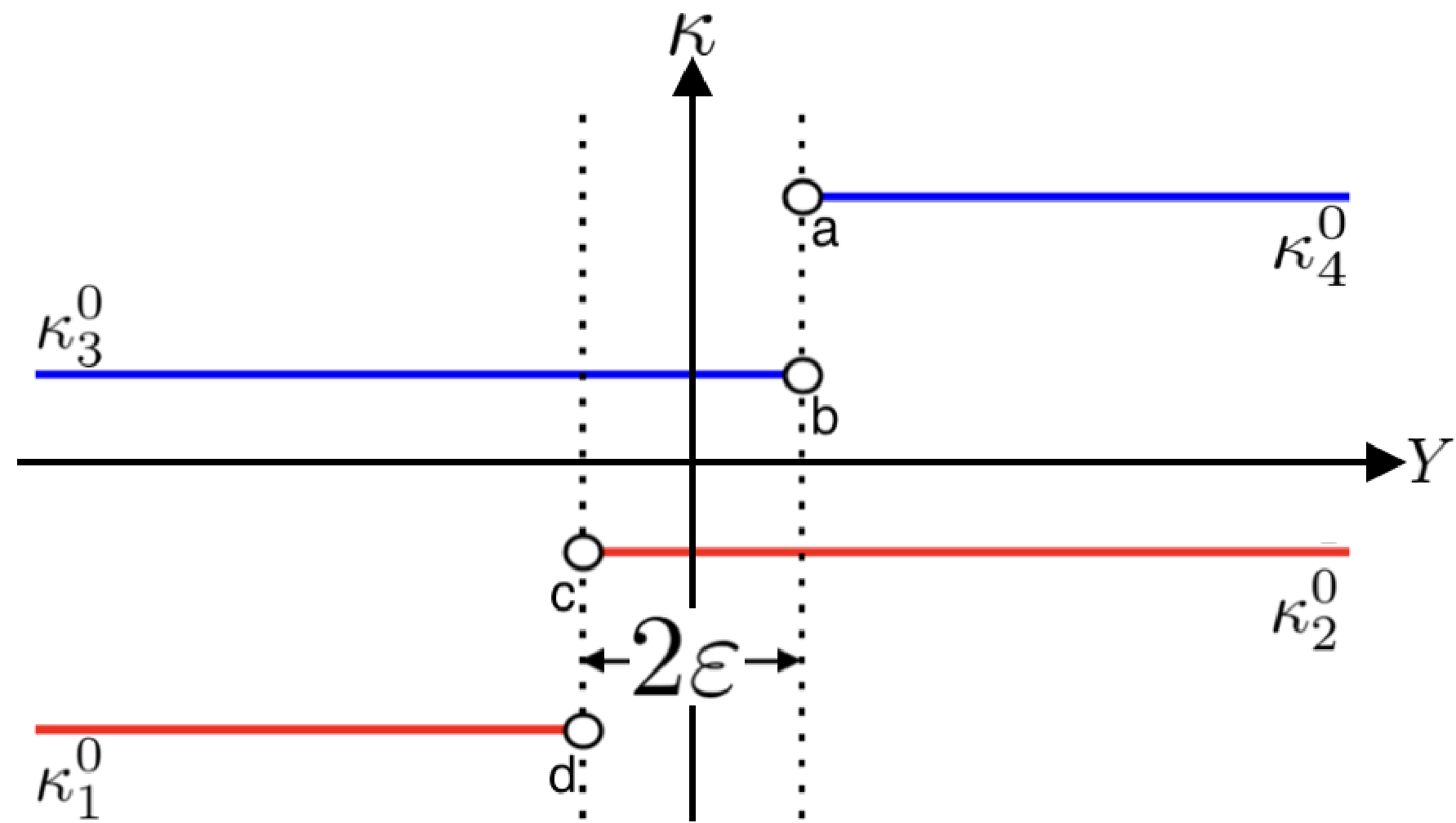}
	\end{minipage}
	\caption{Regularized initial data for the case (c). The right panel shows the initial conditions corresponding to the $\kappa$-graph.}
	\label{fig312}
\end{figure}
One should note that this regularization is to add a small soliton of $[2,3]$-type around $Y=0$, i.e., 
\[
u(x,y,0)=u_{[2,4]}^0(x,y)H(y-\delta)+u_{[1,3]}H(-(y+\delta))+u_{[2,3]}^0(x,y)\left(H(y+\delta)-H(y-\delta)\right),
\]
where $u_{[i,j]}^0(x,y)$ is the $[i,j]$-soliton at $t=0$, and a small number $\delta=\varepsilon\epsilon^{-1}\ll 1$.  Then we compute the characteristic velocities $V_1=2\kappa_1+\kappa_2$ and $V_2=\kappa_1+2\kappa_2$ at the points $Y=\pm\varepsilon$, and we obtain 
\begin{align}\label{Vc}
&V_d:=\lim_{Y\uparrow -\varepsilon}V_{1}(Y)=2\kappa_1^0+\kappa_3^0\quad <\quad V_c:=\lim_{Y\downarrow -\epsilon}V_{1}(Y)=2\kappa_2^0+\kappa_3^0
 \\
&\quad <\quad V_b:=\lim_{Y\uparrow \varepsilon}V_{2}(Y)=2\kappa_3^0+\kappa_2^0\quad <\quad V_a:=\lim_{Y\downarrow \varepsilon}V_{2}(Y)=2\kappa_4^0+\kappa_2^0.\nonumber
\end{align}
This implies that the initial value problem of the $\kappa$-system \eqref{kappa} with the initial data \eqref{R-IDc} admits the global solution (rarefaction wave). The $\kappa$-graphs and the numerical simulations are shown in Figure \ref{fig17}.
\begin{figure}[H]
	\begin{minipage}[htp]{1\linewidth}
		\centering
		\includegraphics[height=5cm,width=15.5cm]{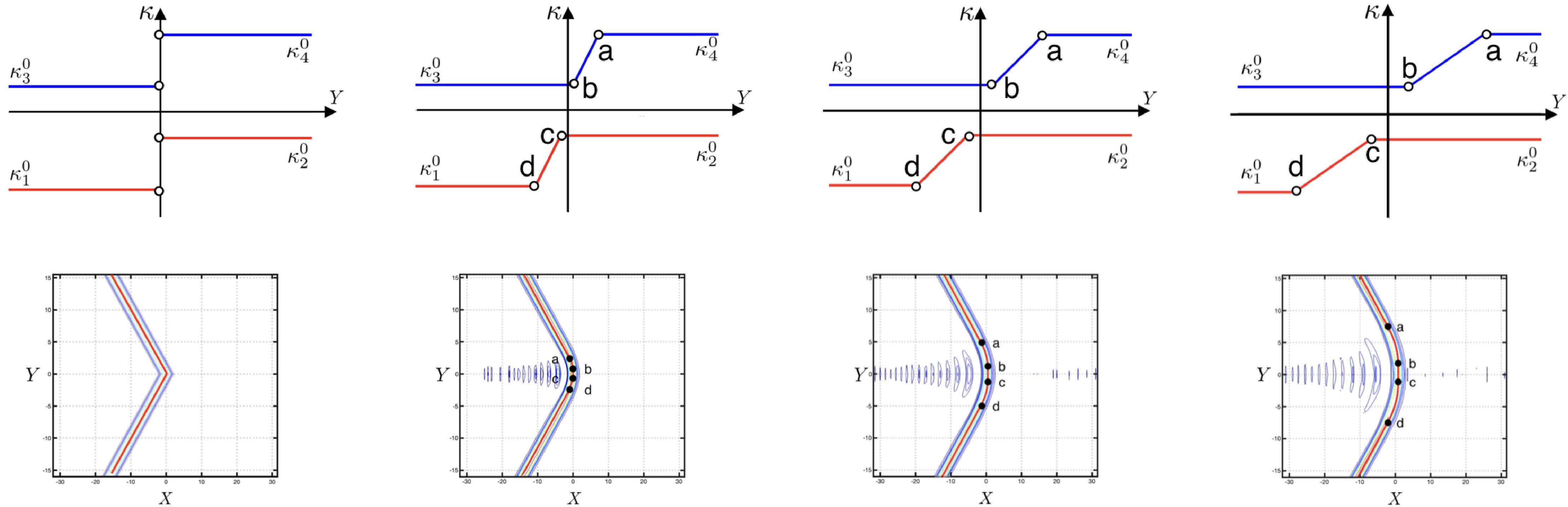}
	\end{minipage}%
	\caption{The $\kappa$-graphs and the numerical simulations for the case (c). The counter plots are obtained at $T=0,1,2,3$. The peak trajectories in the intervals  $(Y_b,Y_a)$ and $(Y_d,Y_c)$ are parabolas. }
	\label{fig17}
\end{figure}
One should note that the solution of this initial value problem depends on the $\varepsilon$,
say $\kappa_i(Y,T;\varepsilon)$, 
and the limit $\varepsilon\downarrow 0$ of the solution is well-defined. An explicit form of the peak trajectory is
given by
\begin{equation}\label{Sc}
X(Y)=\left\{\begin{array}{lll}
-\tan\Psi_{[1,3]}^0Y+C_{[1,3]}^0T,\qquad &\text{for}\quad Y<Y_d(T),\\[1.0ex]
\displaystyle{-\frac{1}{4T}(Y+\kappa_3^0 T)^2+(\kappa_3^0)^2T},  &\text{for}\quad Y_d(T)<Y<Y_c(T),\\[2.0ex]
-\tan\Psi_{[2,3]}^0Y+C^0_{[2,3]}T,  &\text{for}\quad Y_c(T)<Y<Y_b(T),\\[1.0ex]
\displaystyle{-\frac{1}{4T}(Y+\kappa_2^0T)^2+(\kappa_2^0)^2T},  &\text{for}\quad Y_b(T)<Y<Y_a(T),\\[2.0ex]
-\tan\Psi_{[2,4]}^0Y+C^0_{[2,4]}T, &\text{for}\quad Y>Y_a(T),
\end{array}\right.
\end{equation}
where $Y_\alpha(T)=V_\alpha T$ for $\alpha=a,b,c,d$.
{Note here that there are two parabolic solitons $[2]$ and $[3]$, and each parabolic-soliton tangentially connects one half-soliton to another half-soliton.} We summarize the solutions in the table:
		\begin{table}[H]
			\centering
			\begin{tabular}{c|c|c|c|c|c}
				\hline
				{Interval}  & $(-\infty, Y_{d})$ & $(Y_{d}, Y_{c})$ &$(Y_c, Y_b)$ &$(Y_b,Y_a)$ & $(Y_{a},+\infty)$ \\  
				\hline 
				Line-soliton   & $[1,3]$ &  & $[2,3]$ &   & ${[2,4]}$ \\ 
				\hline  
				Parabola-soliton &  & ${[3]}$  & & $[2]$ & \\   
				\hline
			\end{tabular}
		\end{table}
\noindent
One should also note that the $[2,3]$-soliton in the middle remains for $t\to\infty$, that is, we have
$[2,3]$-soliton as the asymptotic solution $u_0(x,y,t)$ in the sense of \eqref{localstability}.

\subsection{The cases (d) and (e)} \label{sec:d}
 The case (d) is a critical case with $\kappa_{3}^{0}=\kappa_{4}^{0}$, $\sqrt{2A_{0}}>2$. 
 The V-shape soliton and $\kappa$-graph corresponding to this initial value are shown in Fig.~\ref{fig19}.
\begin{figure}[htbp]
	\begin{minipage}[htb]{1\linewidth}
		\centering
		\includegraphics[height=3cm,width=10.32cm]{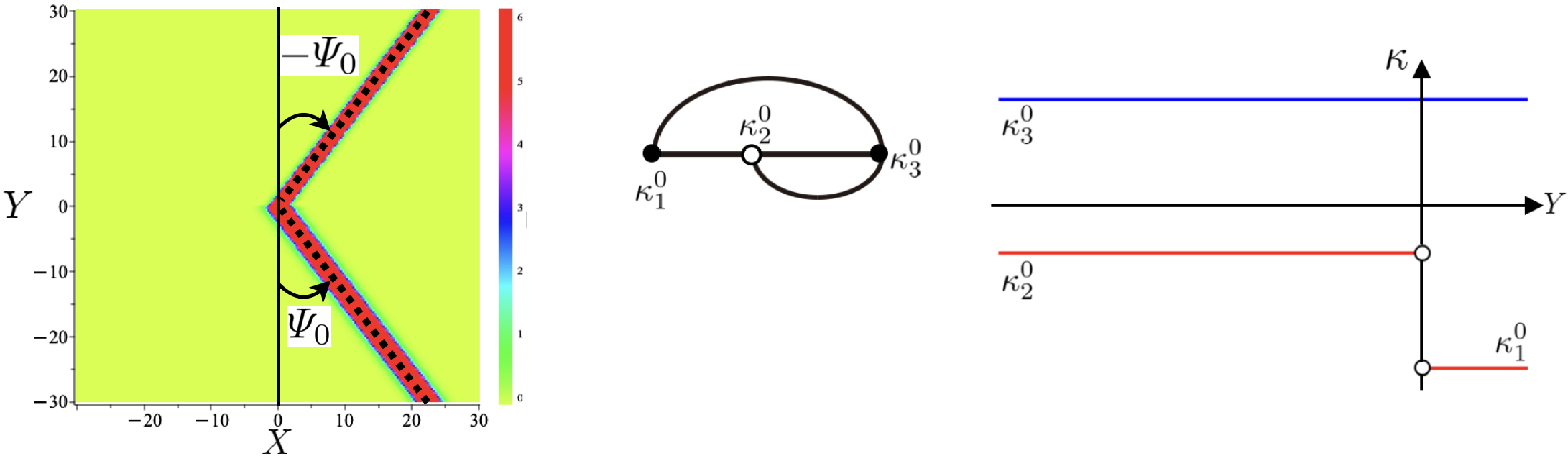}
	\end{minipage}
	\caption{The $V$-shaped initial data of the case (d). The middle panel shows the corresponding incomplete chord diagram with fixed points marked by $\bullet$.}
	\label{fig19}
\end{figure}
Since $\kappa_2^0$ is constant for all $Y$, we have a simple wave system for $\kappa_1$,
\begin{equation}\label{d}
\frac{\partial \kappa_1}{\partial Y}+(2\kappa_1+\kappa_2^0)\frac{\partial\kappa_1}{\partial Y}=0,\qquad\text{with}\qquad \kappa_1(Y,0)=\left\{\begin{array}{ll}
\kappa_2^0,\quad Y<0,\\
\kappa_1^0,\quad Y>0.
\end{array}\right.
\end{equation}
This initial value problem is not well-posed, because the initial data is not increasing (see Lemma \ref{simple}) and the characteristic velocities $V_1$ at $Y=0{\pm}$ satisfy
\[
V_1(0-)=2\kappa_2^0+\kappa_3^0\quad >\quad V_1(0+)=2\kappa_1^0+\kappa_3^0.
\]
Then the solution develops a singularity, i.e., a shock wave (see Lemma \ref{simple}).
To resolve the singularity, we propose a \emph{regularization} to the initial data as shown in Figure
\ref{fig210}.
\begin{figure}[htbp]
	\begin{minipage}[H]{1\linewidth}
	\centering
	\includegraphics[height=3.2cm,width=10.96cm]{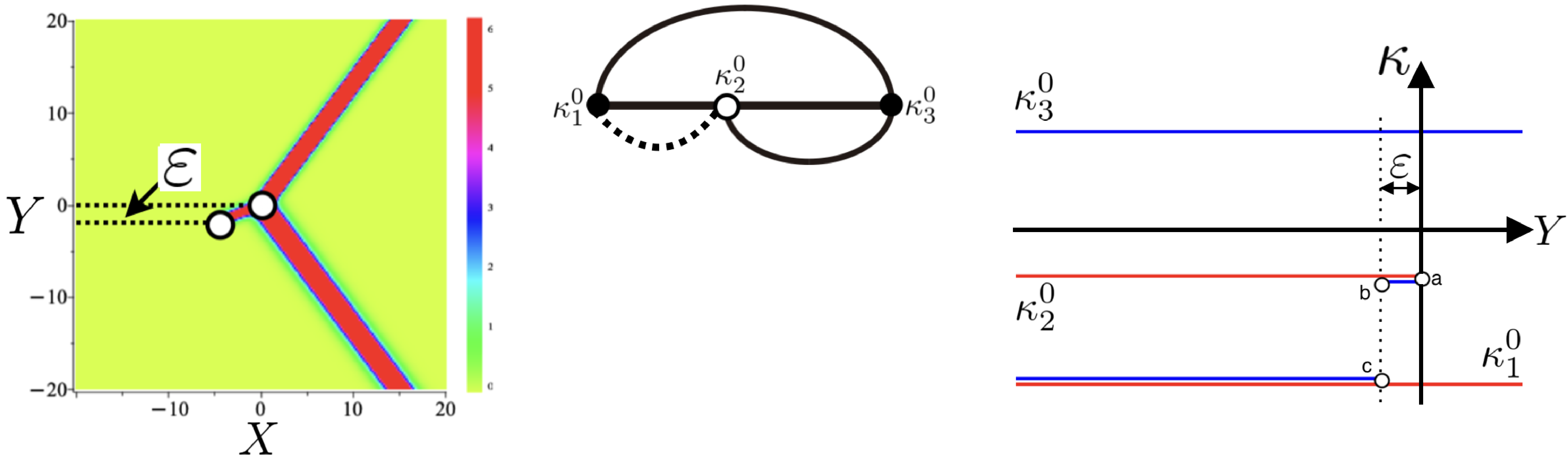}
	\end{minipage}%
	\caption{Regularization of the initial data \eqref{d}.}
	\label{fig210}
\end{figure}
The main idea of the regularization is to \emph{add} a small piece of soliton so that the intersection point forms a resonant Y-soliton. To show that this regularized initial value problem has a global solution, we first recall that the intersection point propagates with the speed $C_a=\kappa_1^0+\kappa_2^0+\kappa_3^0$ (see \eqref{89} with the slow scales, i.e., $Y_a(T)=C_aT$).
For the evolution of the small soliton with this initial data, we consider the following initial value problem for $(\tilde\kappa_1,\tilde\kappa_2)$,
\begin{equation}\label{small}
\frac{\partial\tilde\kappa_2}{\partial T}+(\kappa_1^0+2\tilde\kappa_2)\frac{\partial \tilde\kappa_2}{\partial Y}=0,\qquad -\infty<Y<Y_a(T)=C_aT,
\end{equation}
with the initial data,
\begin{equation}\label{smallID}
\tilde\kappa_1=\kappa_1^0,\quad\text{for}\quad Y<0,\qquad\text{and}\qquad
\tilde\kappa_2=\left\{\begin{array}{ll}
\kappa_1^0,\quad &\text{for}\quad Y<-\varepsilon,\\
\kappa_2^0,&\text{for}\quad -\varepsilon<Y<0.
\end{array}\right.
\end{equation}
Evaluating the characteristic velocities at the points $b$ and $c$, we can see that this has a global solution,
\begin{equation}\label{Sd}
\tilde\kappa_2=\left\{\begin{array}{ll}
\kappa_1^0,\qquad &\text{for}\quad Y<Y_c(T),\\[1.0ex]
\displaystyle{\kappa_1^0+\frac{1}{2T}(Y-Y_c(T)),} &\text{for}\quad Y_c(T)<Y<Y_b(T),\\[2.0ex]
\kappa_2^0, &\text{for}\quad Y_b(T)<Y<Y_a(T).
\end{array}\right.
\end{equation}
Here we have
\[
Y_c(T)=3\kappa_1^0T-\varepsilon\quad<\quad Y_b(T)=(2\kappa_2^0+\kappa_1^0)T-\varepsilon\quad<\quad Y_a(T)=(\kappa_1^0+\kappa_2^0+\kappa_3^0)T.
\]
It is obvious that the limit $\varepsilon\to 0$ gives the solution of the original problem \eqref{d}.
Figure \ref{fig21} shows the $\kappa$-graph and the numerical simulation for $T>0$.
One should note that the Y-soliton of type $\pi=(3,2,1)$ appears as a resonance at the point a in the region $Y>Y_b(T)$,
and this Y-soliton gives the asymptotic solution $u_0(x,y,t)$ in the sense of \eqref{localstability} (see
Figure \ref{figY}). 
    \begin{figure}[htbp]
	\begin{minipage}[htp]{1\linewidth}
		\centering
		\includegraphics[width=9cm,height=3.2cm]{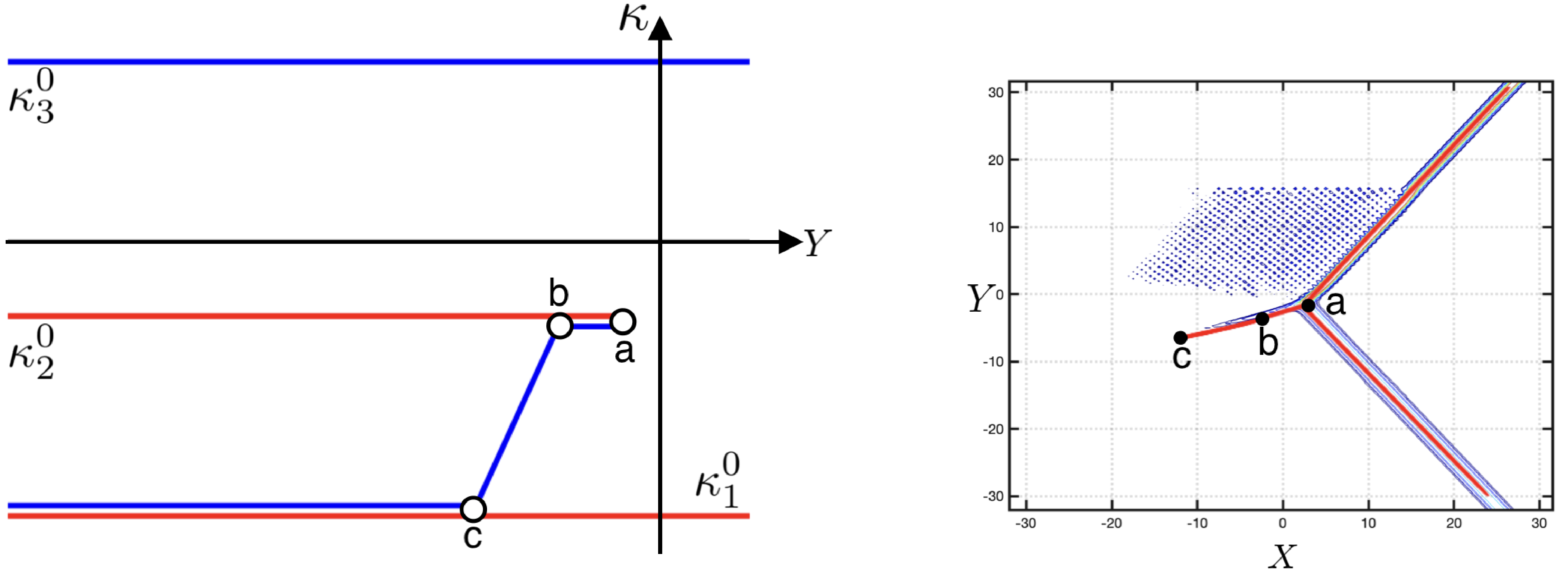}
	\end{minipage}
	\caption{The $\kappa$-graph and the numerical simulation for the case (d).
We take $(\kappa_1,\kappa_2,\kappa_3)=(-\frac{17}{8},-\frac{5}{8},\frac{11}{8})$.
The counter plot at the right panel is at $T=1$. We have $V_c<V_b<V_a$.}
	\label{fig21}
    \end{figure}

The peak trajectory in the region between points b and c is given by
a parabola,
\[
X(Y)=-\frac{1}{4T}(Y+\kappa_1^0T)^2+(\kappa_1^0)^2T.
\]
Note here that $\kappa_1^0$ is a fixed point.

Before closing the case (d), we remark that the incomplete chord diagram in Figure \ref{fig210} has two fixed points (see Definition \ref{def:fix}). Then we show that the global solution contains an additional (resonant) soliton which has the parameter $\kappa_2^0$. This point is another type of the initial point, and we define the following.
\begin{definition}
We define another type of initial points in addition to the points in Definition \ref{def:fix}:
\begin{itemize}
\item[(c)]
A point $\kappa_j^0$ is ``singular'', if there exists a half $[i,k]$-soliton with $\kappa^0_i<\kappa^0_j<\kappa^0_k$.
\end{itemize}
\end{definition}
Note that the singular point is a resonant point, which appears as a \emph{cusp} point of the complete chord diagram.

\begin{remark}
Our regularization is similar to the dispersive regularization used in the Whitham theory for slowly modulated solutions in nonlinear dispersive wave equations (see \cite{BK:94, Kodama5} and also Appendix \ref{A-KW} for the details). 
\end{remark}

The case (e) is also the degenerate case with $\kappa_{1}^{0}=\kappa_{2}^{0}$ and $\sqrt{2A_0}<2$.  The initial half-solitons are $[1,3]$-soliton in $Y>0$ and $[1,4]$-soliton in $Y<0$. 
This is similar to the case (d), and the $\kappa$-system \eqref{kappa} is reduced to a simple wave system for the $\kappa_2$. The initial value problem for $\kappa_2$ is then given by
\begin{equation}\label{e}
\frac{\partial \kappa_2}{\partial T}+(2\kappa_1^0+\kappa_2)\frac{\partial \kappa_2}{\partial Y}=0,
\quad\text{with}\quad \kappa_2(Y,0)=\left\{\begin{array}{ll}
\kappa_4^0,\quad &\text{for}\quad Y<0,\\
\kappa_3^0, &\text{for}\quad Y>0.
\end{array}\right.
\end{equation}
Since the initial data of $\kappa_2$ is decreasing, the system develops a shock singularity. 
The regularization can be done in a similar way as in the case (d) (see Figure \ref{fig22}).
\begin{figure}[htbp]
	\begin{minipage}[t]{1\linewidth}
	\centering
	\includegraphics[height=3cm,width=9.9cm]{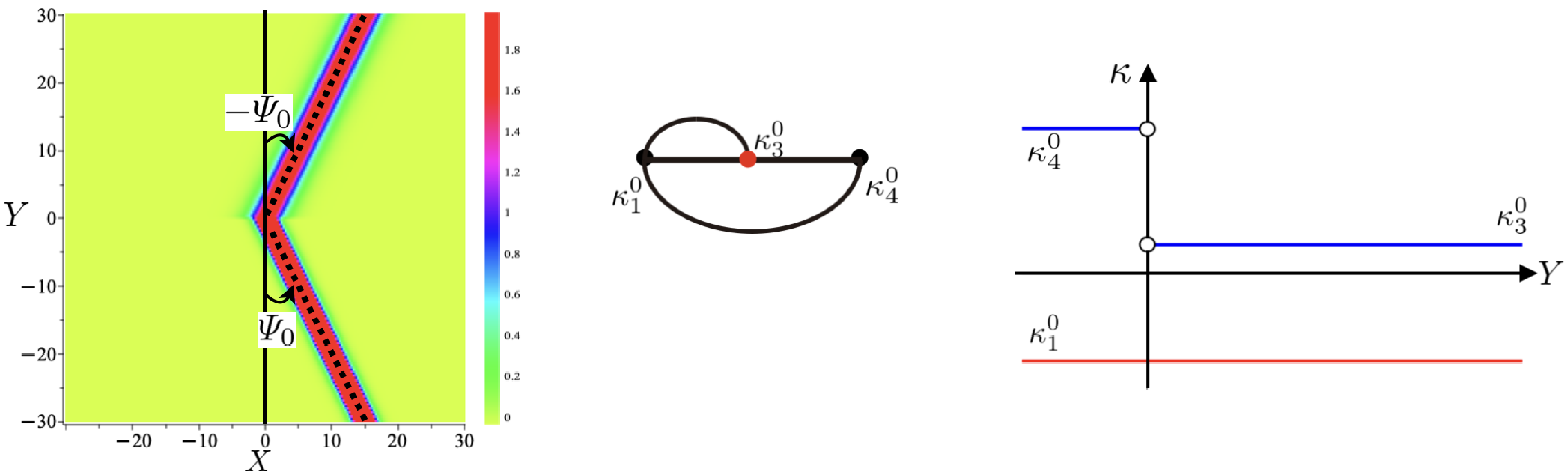}
	\end{minipage}%
	\caption{The $V$-shaped initial data of the case (e). In the incomplete chord diagram (middle panel), black dots are fixed points, while the red dot is a singular point.}
	\label{fig22}
\end{figure}
Figure \ref{fig23} shows the global solution of the $\kappa$-system \eqref{e} and the numerical simulation.
\begin{figure}[htbp]
	\begin{minipage}[htp]{1\linewidth}
		\centering
		\includegraphics[width=11.9cm,height=3cm]{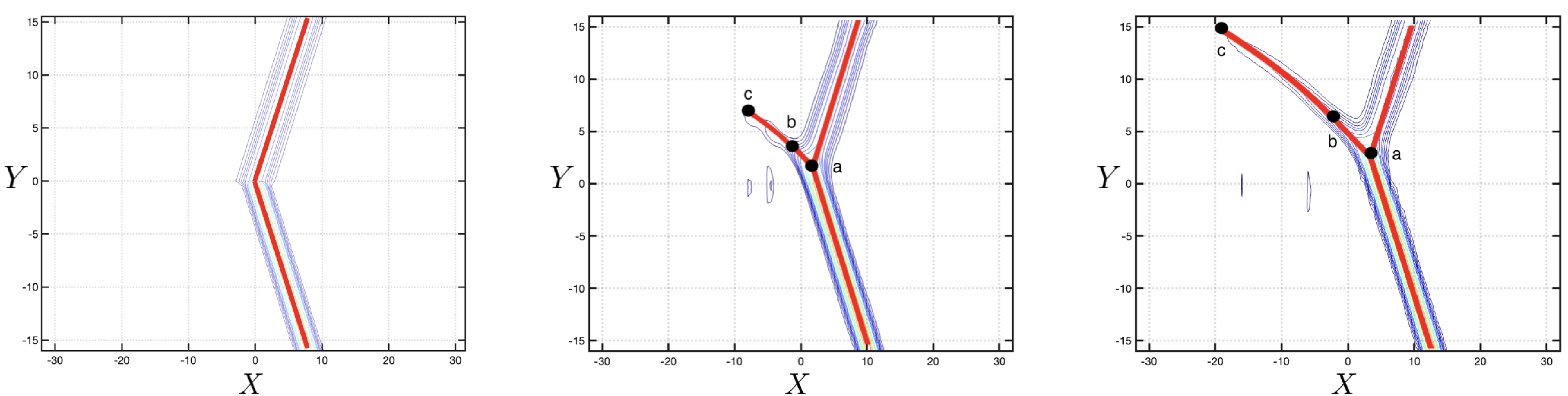}
	\end{minipage}
	\caption{The global solution of the $\kappa$-system \eqref{e} and the numerical simulation with the peak trajectory. We take $(\kappa_1^0=\kappa_2^0,\kappa_3^0,\kappa_4^0)=
	(-\frac{3}{4},\frac{1}{4},\frac{5}{4})$, and the numerical result are at $T=0$, $T=2$ and $T=4$.}
	\label{fig23}
\end{figure} 

Here the asymptotic solution $u_0(x,y,t)$ in \eqref{localstability} is given by the Y-soliton of type $\pi=(2,3,1)$ (see Section \ref{sec:Y}).

\subsection{The case (f)}
The initial half-solitons of this case are $[1,3]$-soliton for $Y>0$ and $[2,4]$-soliton for $Y<0$, i.e.,
\begin{equation}\label{IDf}
\kappa_1=\left\{\begin{array}{ll}
\kappa_2^0,\quad &\text{for}\quad Y<0,\\
\kappa_1^0, &\text{for}\quad Y>0,
\end{array}\right.\qquad\text{and}\qquad
\kappa_2=\left\{\begin{array}{ll}
\kappa_4^0,\quad &\text{for}\quad Y<0,\\
\kappa_3^0, &\text{for}\quad Y>0.
\end{array}\right.
\end{equation}
Note that both $\kappa_1^0$ and $\kappa_2^0$ are decreasing, and the $\kappa$-system 
\eqref{kappa} develops shock waves. 
The initial profile and the regularized initial data are shown in Figure \ref{fig35}.

\begin{figure}[htbp]
	\begin{minipage}[htp]{1\linewidth}
		\centering
		\includegraphics[height=3.3cm,width=11.55cm]{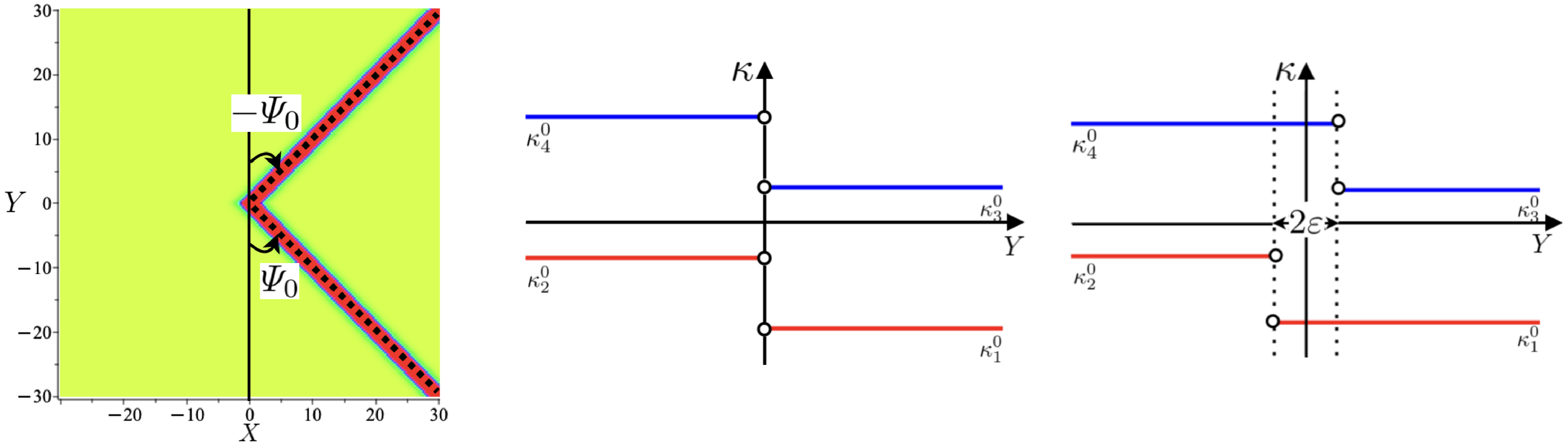}
	\end{minipage}
	\caption{The initial profile and the regularized initial data for the case (f).}
	\label{fig35}
\end{figure}

First note that $\kappa_1^0$ and $\kappa_4^0$ are fixed points, and $\kappa_2^0$ and $\kappa_3^0$ are singular points.
Then, following the arguments in the cases (d) and (e), we obtain a global solution of the $\kappa$-system. 
{This solution has the following structure consisting of multiple linear solitons and parabolic solitons:}
\begin{table}[htbp]
\begin{center}
		\begin{tabular}{c |c |c |c| c| c |c |c}
			\hline
			{Interval}& $(-\infty, Y_{f})$ & $(Y_{f}, Y_{e})$  & $(Y_{e}, Y_{d})$ & $(Y_{d}, Y_{c})$  & $(Y_{c}, Y_{b})$ & $(Y_{b}, Y_{a})$ & $(Y_{a},+\infty)$ \\[1.0ex]
			\hline
			Line-soliton & {${[2,4]}$} &{${[2,4]}$}& {${[1,2]}$, ${[2,4]}$} &{${[1,4]}$} & {${[1,3]}$, ${[3,4]}$} & {${[1,3]}$} & {${[1,3]}$}\\[1.0ex]
			\hline
			Parabolic-soliton & & {${[1]}$}& & & & {${[4]}$} & \\[1.0ex]
			\hline
		\end{tabular}
\end{center}		
\end{table}

Here $Y_\alpha(T)$ for the points $\alpha=a,b,\ldots,f$ are given by
\begin{align*}
Y_a=3\kappa^0_4T\quad &>\quad Y_b=(2\kappa_3^0+\kappa_4^0)T\quad>\quad Y_c=(\kappa_1^0+\kappa_3^0+\kappa_4^0)T\quad>\\\
&>\quad Y_d=(\kappa^0_1+\kappa_2^0+\kappa_4^0)T \quad >\quad 
Y_e=(2\kappa_2^0+\kappa_1^0)T\quad>\quad Y_f=3\kappa_1^0T.
\end{align*}
Figure \ref{fig37} shows the evolution of the $\kappa$-graph and the solution $u(x,y,t)$ of the numerical simulation.
\begin{figure}[H]
	\begin{minipage}[htb]{1\linewidth}
		\centering
		\includegraphics[height=5.5cm,width=11.5cm]{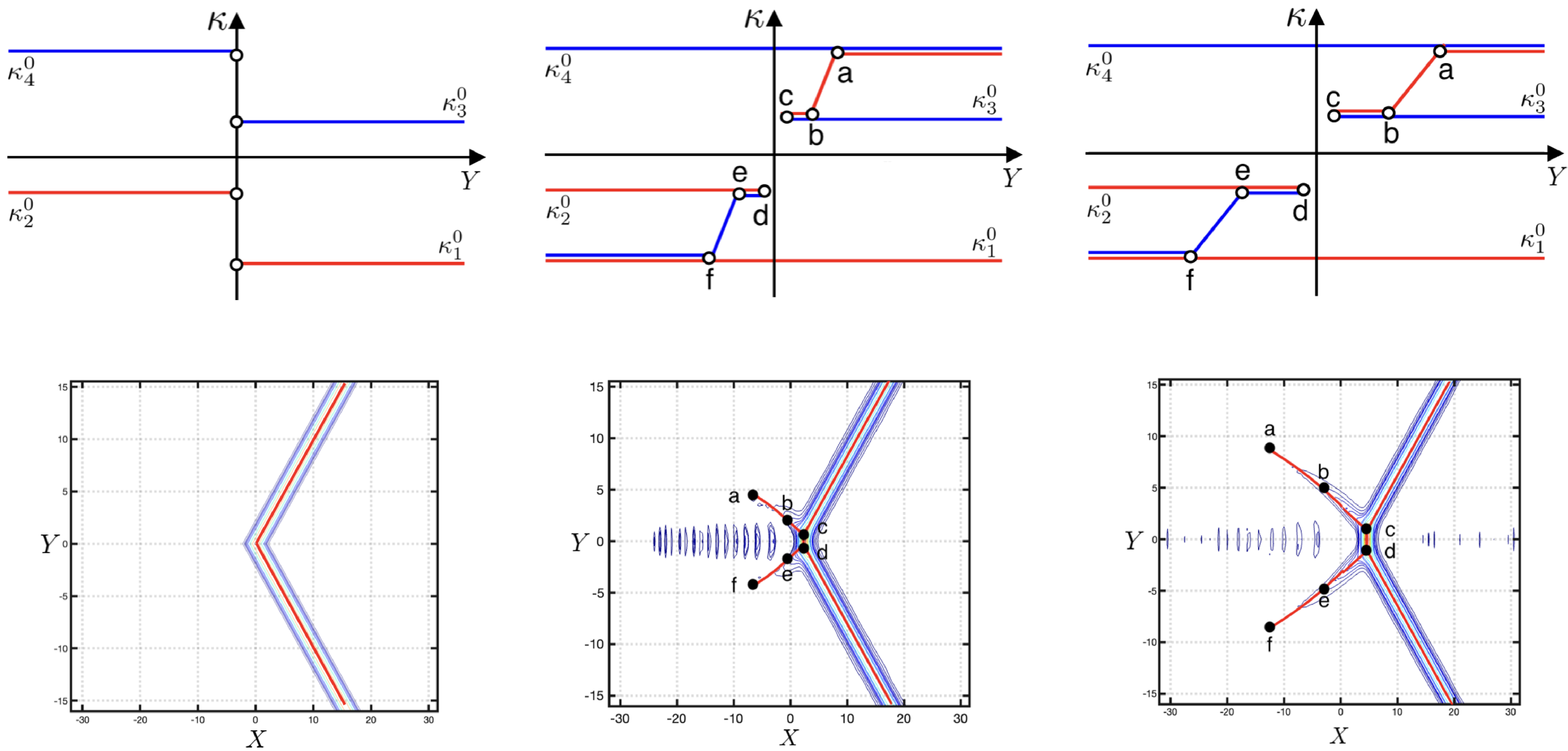}
	\end{minipage}%
	\caption {The $\kappa$-graphs and the numerical simulation for the case (f).
We take $(\kappa_1^0,\kappa_2^0,\kappa_3^0,\kappa_4^0)=(-\frac{3}{2},-\frac{1}{2},\frac{1}{2},\frac{3}{2})$. The contour plots of the numerical solution are at $T=0$, $T=1$, and $T=2$.}
	\label{fig37}
\end{figure}

The asymptotic solution $u_0(x,y,t)$ in the sense of local stability \eqref{localstability}
is given by the KP soliton of type $\pi=(3,1,4,2)$, whose $\tau$-function is given by
$\tau=|AE^T|$ with
\[
A=\begin{pmatrix}
1&a&0&-c\\
0&0&1&b
\end{pmatrix}, \qquad \text{and}\qquad E=\begin{pmatrix}
E_1 &E_2&E_3 &E_4\\
\kappa_1E_1&\kappa_2E_2&\kappa_3E_3&\kappa_4E_4
\end{pmatrix},
\]
where the constants $a$, $b$, $c$ are given by
\[
a=\frac{\kappa_1-\kappa_3}{\kappa_2-\kappa_3},\quad b=\frac{\kappa_1-\kappa_3}{\kappa_1-\kappa_4},\quad \text{and} \quad c=\frac{\kappa_1-\kappa_3}{\kappa_3-\kappa_4}.
\]

\subsection{The cases (g) and (h)}
{Since the case (h) is just the upside-down version of the case (g), we discuss only the case (g).}

The initial line-solitons are $[2,3]$-soliton for $Y>0$ and $[1,4]$-soliton for $Y<0$.
The initial data for the $\kappa$-system \eqref{kappa} is then given by
\begin{equation}\label{IDg}
\kappa_1(Y,0)=\left\{\begin{array}{ll}
\kappa_1^0,\quad &\text{for}\quad Y<0,\\
\kappa_2^0, &\text{for}\quad Y>0,
\end{array}\right. \quad\text{and}\quad
\kappa_2(Y,0)=\left\{\begin{array}{ll}
\kappa_4^0,\quad&\text{for}\quad Y<0,\\
\kappa_3^0, &\text{for}\quad Y>0.
\end{array}\right.
\end{equation}
Figure \ref{fig25} shows the initial data for a symmetric choice of the parameters
$(\kappa_1^0,\kappa_2^0,\kappa_3^0,\kappa_4^0)=(-1,-\frac{1}{2},\frac{1}{2},1)$.
\begin{figure}[htbp]
	\begin{minipage}[t]{1\linewidth}
		\centering
		\includegraphics[height=3.5cm,width=13cm]{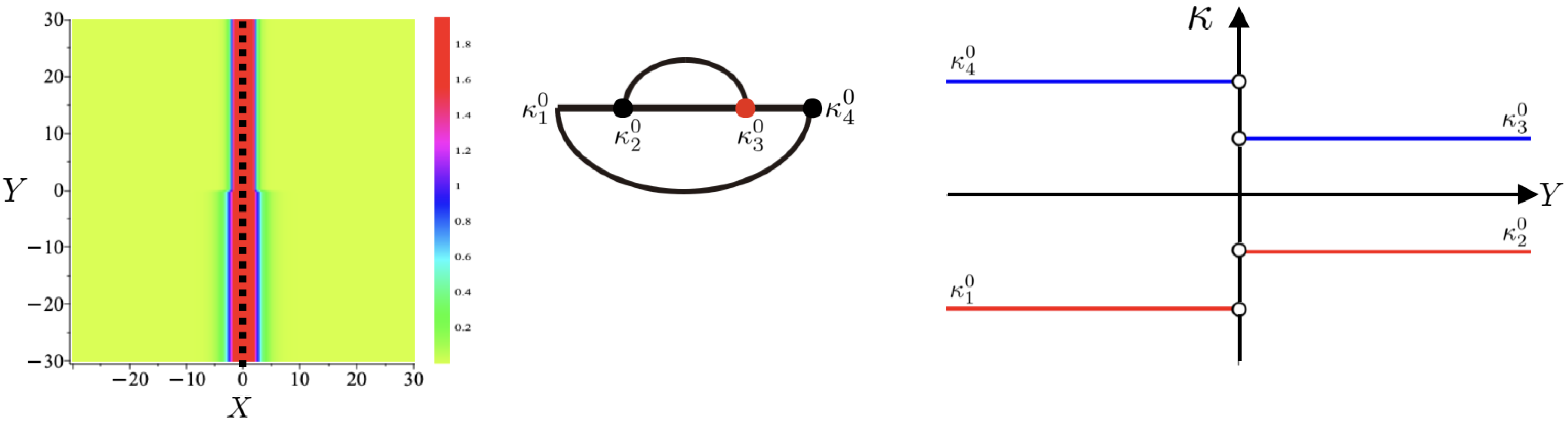}
	\end{minipage}
\caption{The initial data for the case (g). The parameters $\kappa_2^0$ and $\kappa_4^0$ are fixed points, and
$\kappa_3^0$ is a singular point. The $\kappa_1^0$ is a free point.
}
\label{fig25}
\end{figure}
Note that the $\kappa_1$ is increasing, but $\kappa_2$ is decreasing. This implies that 
the $\kappa_1$ is a rarefaction wave, and the $\kappa_2$ is a shock wave. The regularization of the initial data is then given in Figure \ref{fig77}.
\begin{figure}[H]
	\begin{minipage}[t]{1\linewidth}
		\centering
		\includegraphics[height=3.5cm,width=10cm]{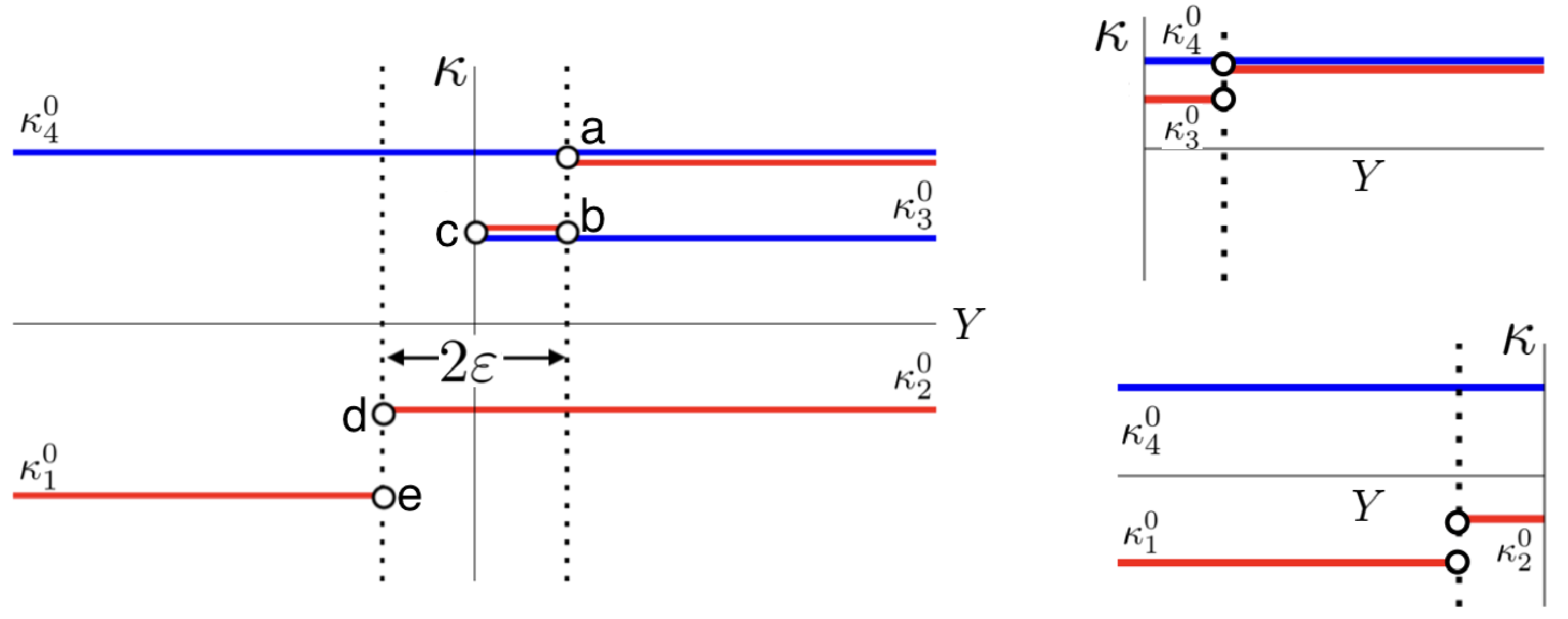}
	\end{minipage}
	\renewcommand\figurename{Figure}
	\caption{The regularized initial data for the case (g). }
	\label{fig77}
\end{figure}
The characteristic velocities at the points from a through e are given by
\begin{align}\label{Vg}
V_a=3\kappa_4^0\quad &>\quad V_b=2\kappa_3^0+\kappa_4^0\quad >\quad V_c=\kappa_2^0+\kappa_3^0+\kappa_4^0\quad>\\
&>\quad V_d=2\kappa_2^0+\kappa_4^0\quad >\quad V_e=2\kappa_1^0+\kappa_4^0.\nonumber
\end{align}
Thus, all points are separated with increasing distances between them. This implies that 
the system with this regularized initial data has a global solution (see Figure \ref{fig27}).
The numerical simulations are shown in Figure \ref{fig27}.
\begin{figure}[htbp]
	\begin{minipage}[t]{1\linewidth}
		\centering
		\includegraphics[height=5.5cm,width=13cm]{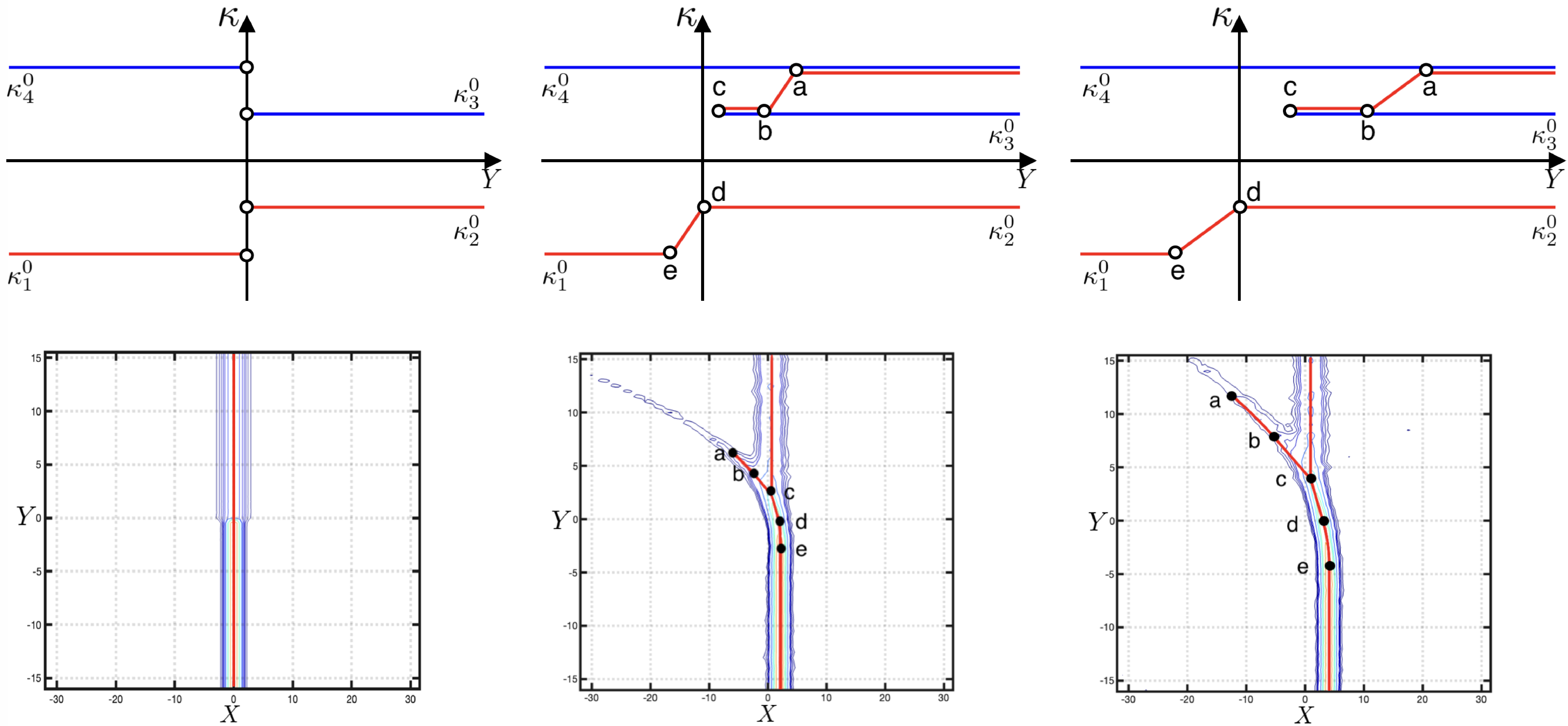}
	\end{minipage}%
	\caption {The $\kappa$-graphs and numerical simulation for the case (g). 
We take $(\kappa_1^0,\kappa_2^0,\kappa_3^0,\kappa_4^0)=(-1,-\frac{1}{2},\frac{1}{2},1)$,
and the $\kappa$-graphs are obtained at $T=0$, $T=2$, and $T=4$.}
	\label{fig27}
\end{figure}
The solution $u(x,y,t)$ consists of line-solitons and parabolic-solitons listed in the table below.
In the table, $Y_\alpha$ for the points $\alpha=a,b,\ldots, e$ are given by $Y_\alpha(T)=V_{\alpha}T$
with $V_\alpha$ in \eqref{Vg}. In particular, $Y_c$ corresponds to the resonant interaction point. It is also interesting to note that two parabolic-solitons
in $(Y_e,Y_d)$ and $(Y_b,Y_a)$ are the same $[4]$-soliton, and the part of the parabola is replaced by two line-solitons, $[2,4]$ and $[3,4]$ in the Y-soliton.

\medskip
\begin{center}
		\begin{tabular}{c |c |c |c| c| c |c }
			\hline
			{Interval}& $(-\infty, Y_{e})$  & $(Y_{e}, Y_{d})$ & $(Y_{d}, Y_{c})$  & $(Y_{c}, Y_{b})$ & $(Y_{b}, Y_{a})$ & $(Y_{a},+\infty)$ \\[1.0ex]
			\hline
			Line-soliton & {${[1,4]}$}&  &{${[2,4]}$} & {${[2,3]}$, ${[3,4]}$} & {${[2,3]}$} & {${[2,3]}$}\\[1.0ex]
			\hline
			Parabolic-soliton & & {${[4]}$}&  &  & {${[4]}$} & \\[1.0ex]
			\hline
		\end{tabular}
\end{center}		
Noting that $Y_e(T)\to-\infty$ and $Y_b(T)\to+\infty$ as $T\to\infty$, we can see that the asymptotic solution $u_0(x,y,t)$ in \eqref{localstability} is given by the Y-soliton of
type $\pi=(2,3,1)$ (see Section \ref{sec:Y}).

\subsection{The case (i)}
The initial half-solitons are $[1,2]$-soliton for $Y>0$ and $[3,4]$-soliton for $Y<0$.
We first note that $\kappa_1^0$ and $\kappa_4^0$ are fixed points. Then 
we take the regularized initial data as shown in Figure \ref{fig28}.
 \begin{figure}[htbp]
	\begin{minipage}[htb]{1\linewidth}
		\centering
		\includegraphics[height=3.2cm,width=5cm]{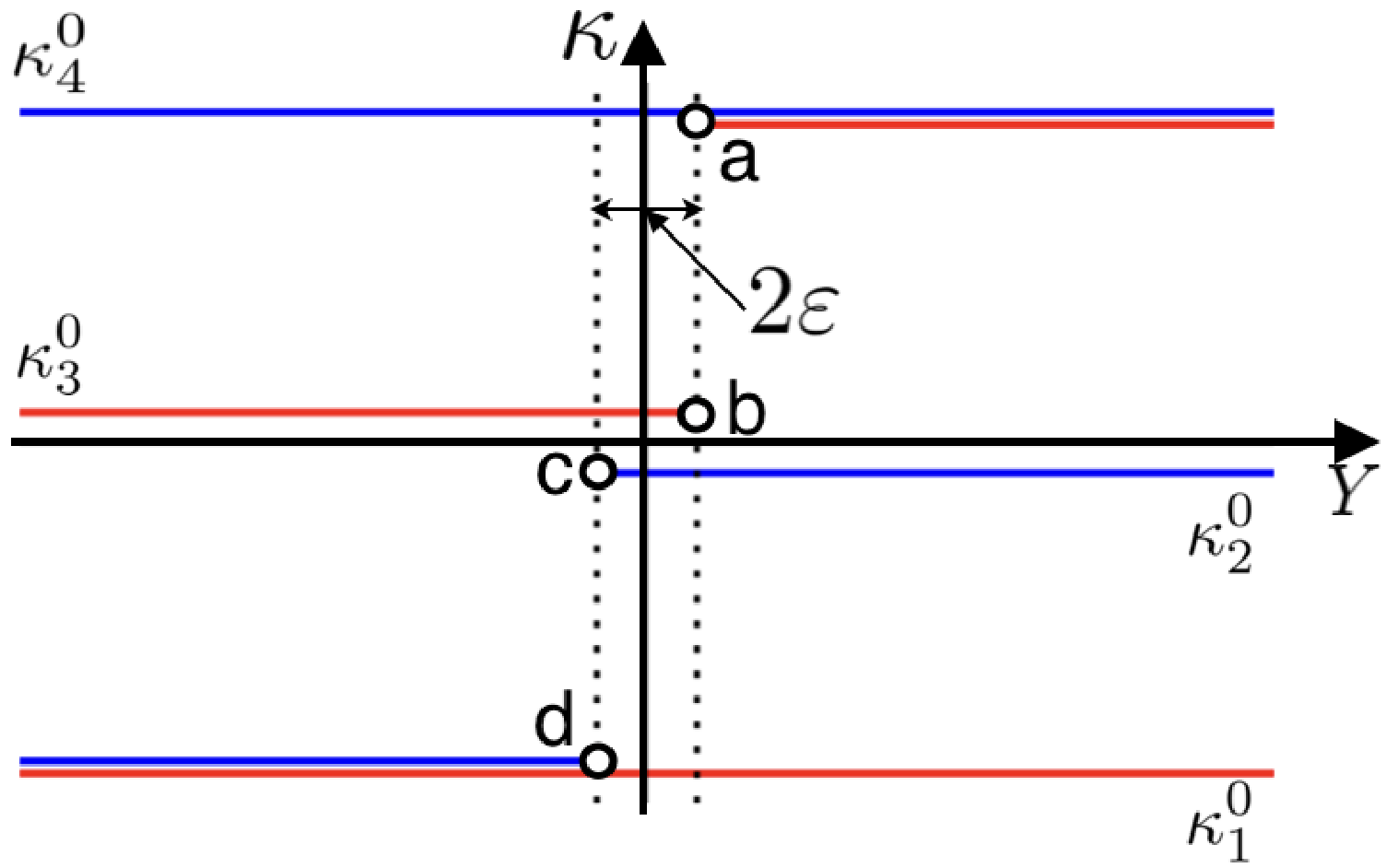}
	\end{minipage}%
	\caption{The regularized initial data for the case (i).}
	\label{fig28}
\end{figure}
That is, the regularization is to add a piece of O-type soliton in a small region $(-\varepsilon,\varepsilon)$
for $0<\varepsilon\ll 1$.
These two line-solitons in the O-type soliton do not have resonance, and they do not have any 
interaction in this region, that is, we can consider them as independent half-solitons (see section \ref{sec:OY}).
We also remark that the phase shift in the O-type soliton is ignored in the slow scales.
Then using the results of the section \ref{sec:H}, we find that the characteristic velocities at the points $a, b, c$ and $d$ are given by
\begin{equation}\label{Vi}
V_a=3\kappa_4^0\quad >\quad V_b=2\kappa_3^0+\kappa_4^0\quad>\quad 
V_c=2\kappa_2^0+\kappa_1^0\quad>\quad V_d=3\kappa_1^0.
\end{equation}
{Thus, all points are separated by increasing distances as $T$ increases.}
This implies that the $\kappa$-system has a global solution for $T>0$. Figure \ref{fig29}
shows the global solution of the $\kappa$-system and the numerical simulation.
 \begin{figure}[htbp]
	\begin{minipage}[htb]{1\linewidth}
		\centering
		\includegraphics[height=5.5cm,width=12.35cm]{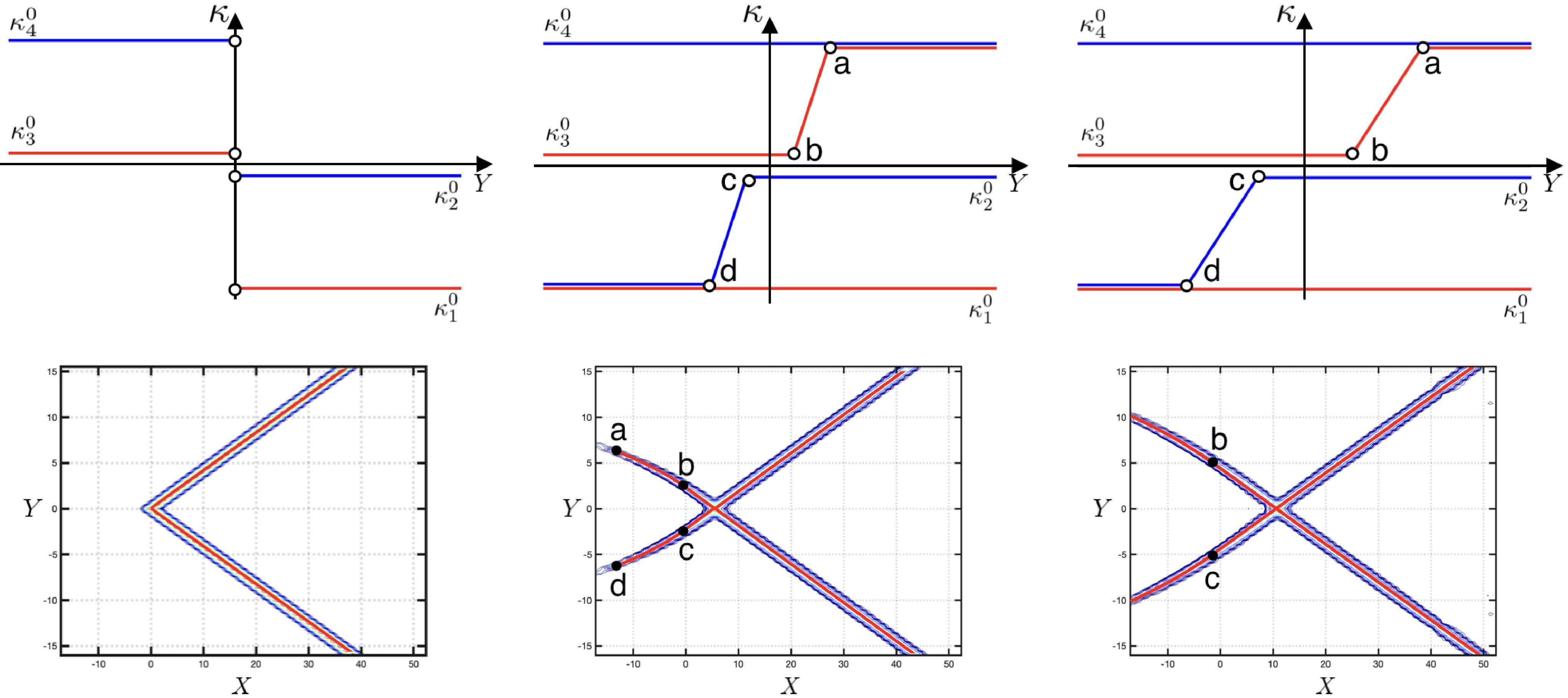}
	\end{minipage}%
	\caption {The $\kappa$-graphs and numerical simulation for the case (i).
	We take $(\kappa_1^0,\kappa_2^0,\kappa_3^0,\kappa_4^0)=(-\frac{11}{5},-\frac{1}{5},\frac{1}{5},\frac{11}{5})$, and the simulations are at $T=0,1$ and $T=2$. The endpoints a and d in the left panel cannot be displayed in this numerical simulation because the values of $X_a$ and $X_d$ are out of range.}
	\label{fig29}
\end{figure}

In the figure, $Y_\alpha(T)$ for the points $\alpha=a, b, c, d$ are given by $Y_\alpha(T)=V_\alpha T$
with $V_\alpha$ given in \eqref{Vi}, and we have parabolic $[4]$-soliton in the region $(Y_b,Y_a)$ and
parabolic $[1]$-soliton in the region $(Y_d,Y_c)$, see the table below.
\begin{table}[H]
	\centering
	\begin{tabular}{c|c|c|c|c|c}
		\hline
	{Interval}& $(-\infty, Y_{d})$ & $(Y_{d}, Y_{c})$  & $(Y_{c},Y_{b})$ & $(Y_{b}, Y_{a})$ & $(Y_{a},+\infty)$ \\  
		\hline 
		Line-soliton & $[3,4]$ & $[3,4]$  &  ${[1,2]}$, ${[3,4]}$ & ${[1,2]}$ & ${[1,2]}$  \\ 
		\hline  
		Parabolic-soliton &  & $[1]$ &  &$[4]$ & \\  
		\hline
	\end{tabular}
\end{table}

\noindent
The asymptotic solution $u_0(x,y,t)$ in \eqref{localstability} is given by O-type (see Section \ref{sec:O}).
\subsection{The case (j)}
The initial half-solitons are $[1,2]$-soliton for $Y<0$ and $[3,4]$-soliton for $Y>0$. In this case, $\kappa_2^0$ and $\kappa_3^0$ are the fixed points. The other two points are free points. We then have the regularized initial data shown in Figure \ref{fig31}.
\begin{figure}[htbp]
	\begin{minipage}[htb]{1\linewidth}
		\centering
		\includegraphics[height=3cm,width=5cm]{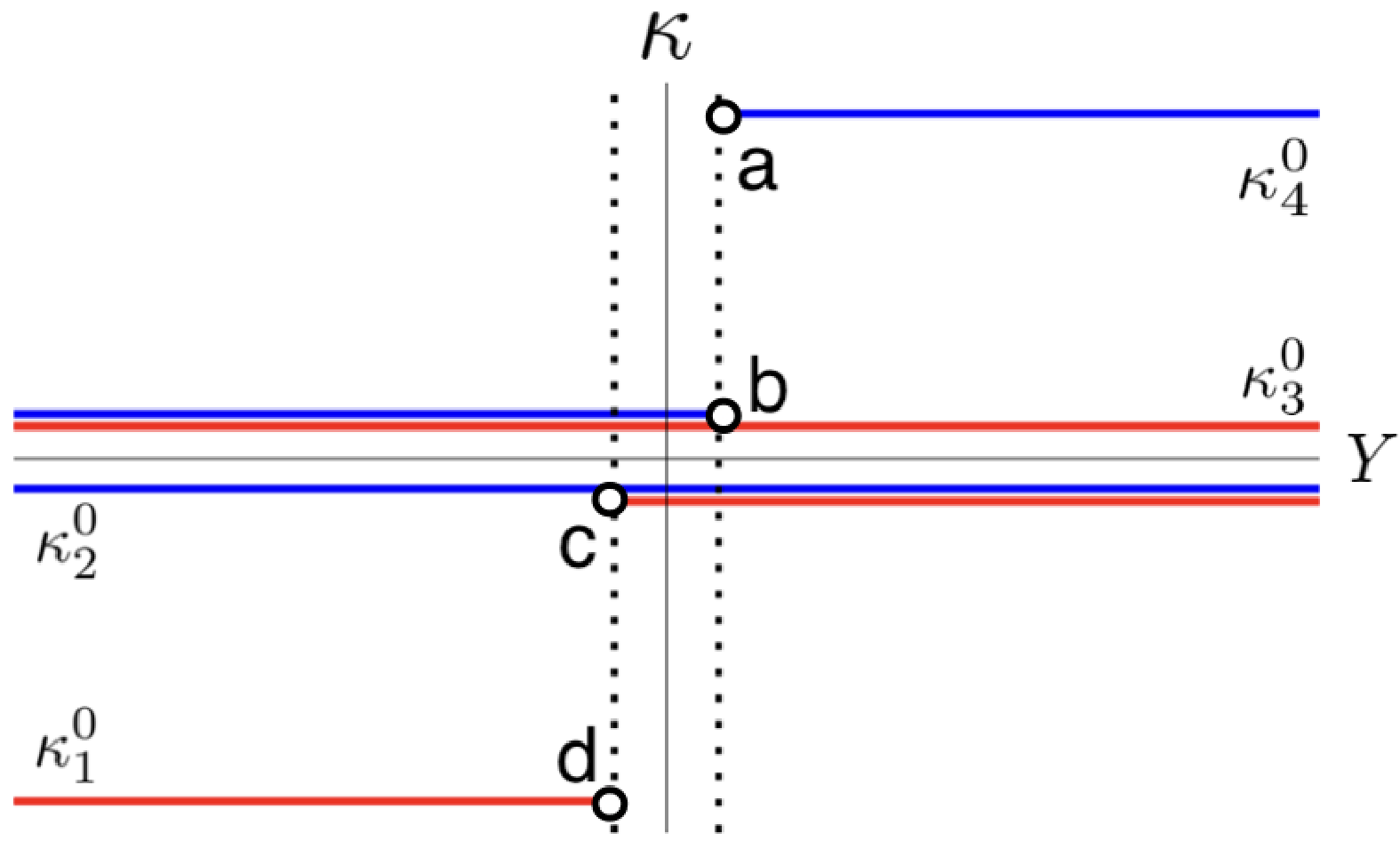}
	\end{minipage}%
	\caption{The regularized initial data for the case (j).}
	\label{fig31}
\end{figure}
In a case similar to case (i), we find that the characteristic velocities at the
points a, b, c, and d are given by
\begin{equation}\label{Vj}
V_a=2\kappa_4^0+\kappa_3^0\quad >\quad V_b=3\kappa_3^0\quad>\quad 
V_c=3\kappa_2^0\quad>\quad V_d=2\kappa_1^0+\kappa_2^0.
\end{equation}
Again, we note that all points are separated with increasing distances between them.
The $\kappa$-graphs and the results of numerical simulation are shown in Figure \ref{fig45}.
\begin{figure}[htbp]
	\begin{minipage}[htb]{1\linewidth}
		\centering
		\includegraphics[height=6cm,width=13.4cm]{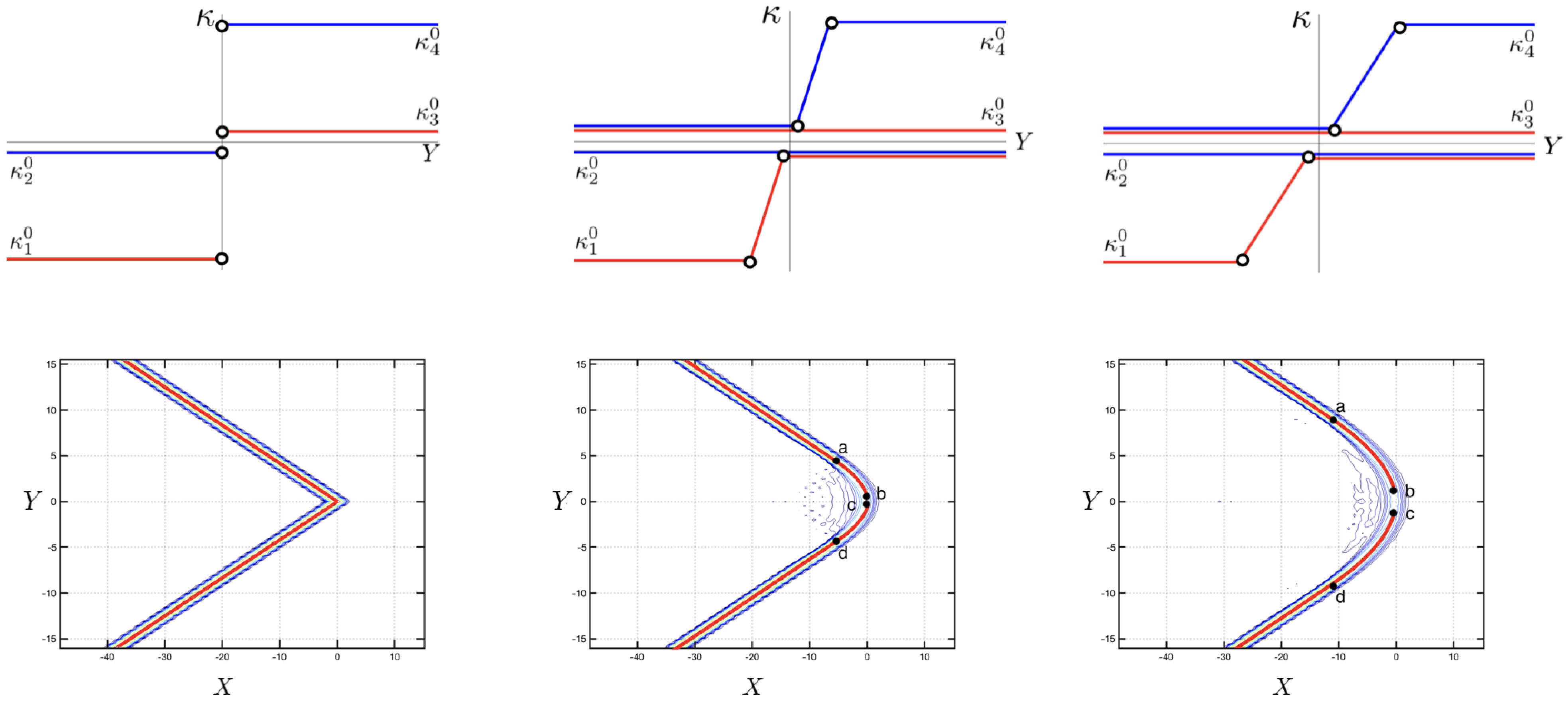}
	\end{minipage}%
	\caption{The $\kappa$-graphs and the  numerical simulation for the case (j).
	We take $(\kappa_1^0,\kappa_2^0,\kappa_3^0,\kappa_4^0)=(-\frac{11}{5},-\frac{1}{5},\frac{1}{5},\frac{11}{5})$, and the simulations are at $T=0,1$ and $T=2$. }
	\label{fig45}
\end{figure}
In the figure, $Y_\alpha$ for the points $\alpha=a, b, c$ and d are given by $Y_\alpha(T)=V_\alpha T$ with $V_\alpha$ in \eqref{Vj}, and the peak trajectories in the regions $(Y_b,Y_a)$ and $(Y_d,Y_c)$ are 
parabolic $[3]$-soliton and $[2]$-soliton, respectively. We summarize the result in the table:
\begin{table}[H]
	\centering
	\begin{tabular}{c|c|c|c|c|c}
		\hline
	{Interval}& $(-\infty, Y_{d})$ & $(Y_{d}, Y_{c})$  & $(Y_{c}, Y_{b})$ & $(Y_{b}, Y_{a})$ & $(Y_{a},+\infty)$ \\  
		\hline 
		Line-soliton & $[1,2]$ &   &  &  & ${[3,4]}$  \\ 
		\hline  
		Parabolic-soliton &  & $[2]$ & &$[3]$ & \\  
		\hline
	\end{tabular}
\end{table}

\noindent
Notice that the asymptotic solution $u_0(x,y,t)$ in the sense of \eqref{localstability} is just zero,
since $Y_c\to-\infty$ and $Y_b\to+\infty$ as $T\to\infty$.
\subsection{Summary}\label{summary}
In this section, we studied the initial value problem of the $\kappa$-system 
with V-shape initial data. {We summarize our main results and discuss possible directions for future research as follows.}
For given V-shape initial data,  we started with the following  steps:
\begin{itemize}
\item[(1)] Draw an incomplete chord diagram for each initial data with the soliton parameters $(\kappa_1^0,\kappa_2^0,\kappa_3^0,\kappa_4^0)$.
\item[(2)] Identify the type for each parameter $\kappa^0_i$ according to the results of the half-soliton problem in Section \ref{sec:H}.
\item[(3)] Based on the types of the parameters, give a regularization for the initial data, so that
the initial value problem for the $\kappa$-system admits a global solution.
\end{itemize}
Then we obtained the following theorem as a summary of our results.
\begin{theorem}\label{main}
For the initial value problem with V-shape initial data, the asymptotic solution can be characterized as follows. In the incomplete diagram given by the initial data, 
\begin{itemize}
\item[(a)] each singular point corresponds to a shock singularity, which generates an additional soliton,
\item[(b)] each free point corresponds to a parabolic-soliton,
\item[(c)] each fixed point gives a parameter of the parabolic-soliton.
\end{itemize}
Then there exists regularized initial data, and the (asymptotic) solution consists of line-solitons and parabolic-solitons. In particular, between two line-solitons, there is a parabolic-soliton, which connects tangentially to these line-solitons.
\end{theorem}
 Figure \ref{fig32} shows the asymptotic solitons and the corresponding complete chord diagrams for the initial data given in Figure \ref{fig18}.
\begin{figure}[htbp]
	\begin{minipage}[htb]{1\linewidth}
		\centering
		\includegraphics[width=11cm,height=5cm]{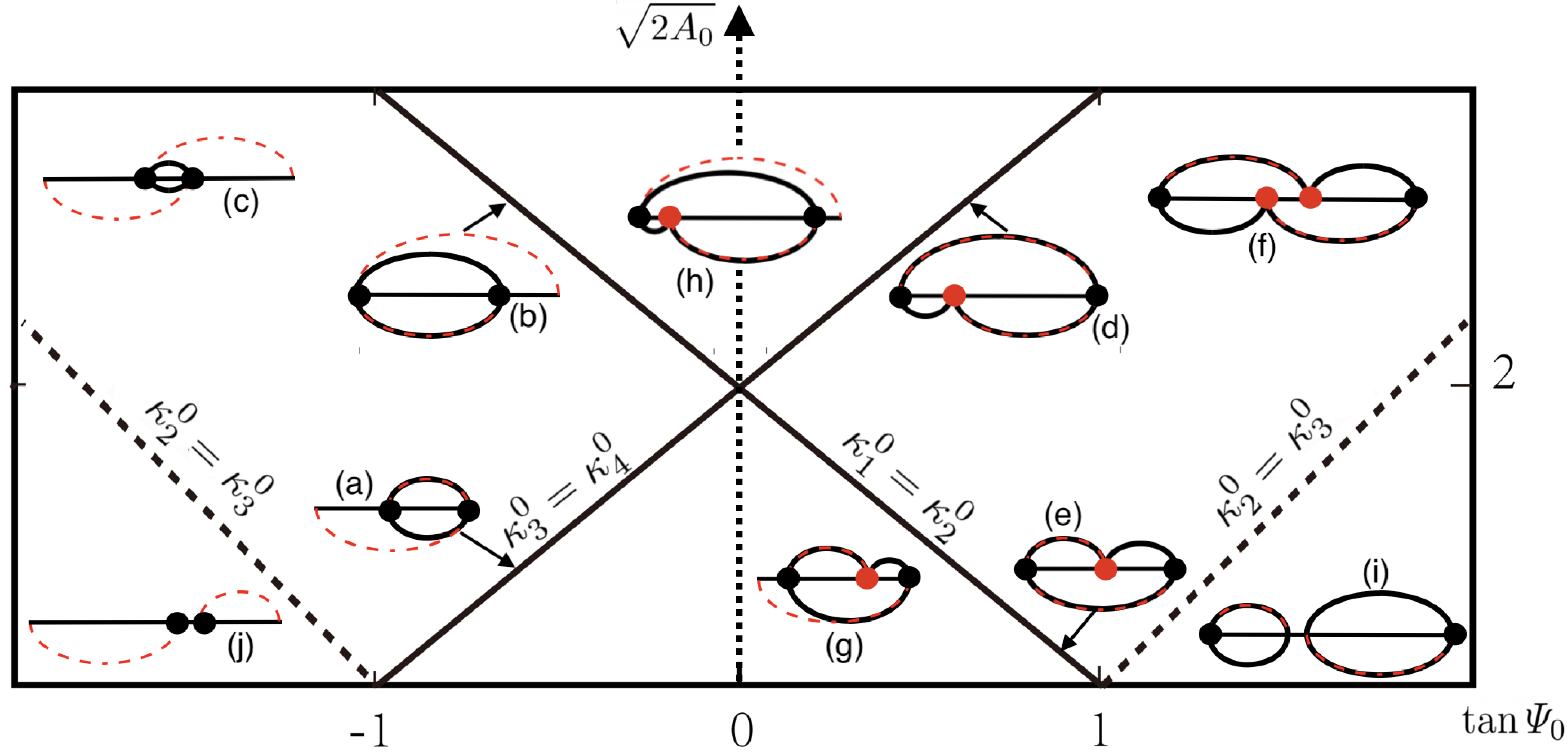}
	\end{minipage}
	\caption{The complete chord diagrams for the initial value problem of the $\kappa$-system with the V-shape initial data in Figure \ref{fig18}. This shows the asymptotic solutions $u_0(x,y,t)$ in the sense of local stability \eqref{localstability}. The dashed curves in the chord diagrams show the initial chords.}
	\label{fig32}
\end{figure}

As a final remark, we add the following result of the initial value problem of \eqref{kappa} for the initial data \eqref{45} given by
\begin{equation*}
\left\{\begin{array}{lll}
u^+_0(x,y)=A_0\sech^2\sqrt{\frac{A_0}{2}}(x+\tan\Psi_0^+y),\\
u^-_0(x,y)=A_0\sech^2\sqrt{\frac{A_0}{2}}(x+\tan\Psi_0^-y).
\end{array}\right.
\end{equation*}
{Note that the above initial data can be viewed as an extension of the bent soliton initial data studied in \cite{Ryskamp1}, allowing an asymmetric V-shaped configuration: the amplitudes of $u^{+}$ and $u^{-}$ are the same, while their slopes may differ. As shown in Figure \ref{fig33}, once the structure of the incomplete chord diagram is determined, the corresponding global solution for the asymmetric V-shaped initial data can be obtained from the known symmetric cases (see Figures \ref{fig96}, \ref{fig18}, and \ref{fig32}), yielding the appropriate combination of line and parabolic solitons, e.g., the case (c) corresponds to the left panel of Figure \ref{fig33}, while the case (f) corresponds to the right panel.}  There are four different asymptotic solutions $u_0(x,y,t)$:
\begin{itemize}
\item[(a)] Above the line $\kappa_2^0=\kappa_3^0$ crossing $\tan\Psi_0^+=4$,  $u_0(x,y,t)=0$ (cf. the case (j)).
\item[(b)] Between the line in (a) and $\tan\Psi_0^+=\tan\Psi_0^-$, $u_0(x,y,t)$ is a line-soliton with the amplitude $0<A<2$ (cf. the case (c)).
\item[(c)] Between the line $\tan\Psi_0^+=\tan\Psi_0^-$ and the line $\kappa_2^0=\kappa_3^0$ crossing $\tan\Psi_0^-=4$, $u_0(x,y,t)$ is of $(3,1,4,2)$-type (cf. the case (f)).
\item[(d)] Below the line $\kappa_2^0=\kappa_3^0$ crossing $\tan\Psi_0^-=4$,  $u_0(x,y,t)$ is O-type (cf. the case (i)).
\end{itemize}

\begin{figure}[htbp]
	\begin{minipage}[t]{1\linewidth}
		\centering
		\includegraphics[height=6.56cm,width=15cm]{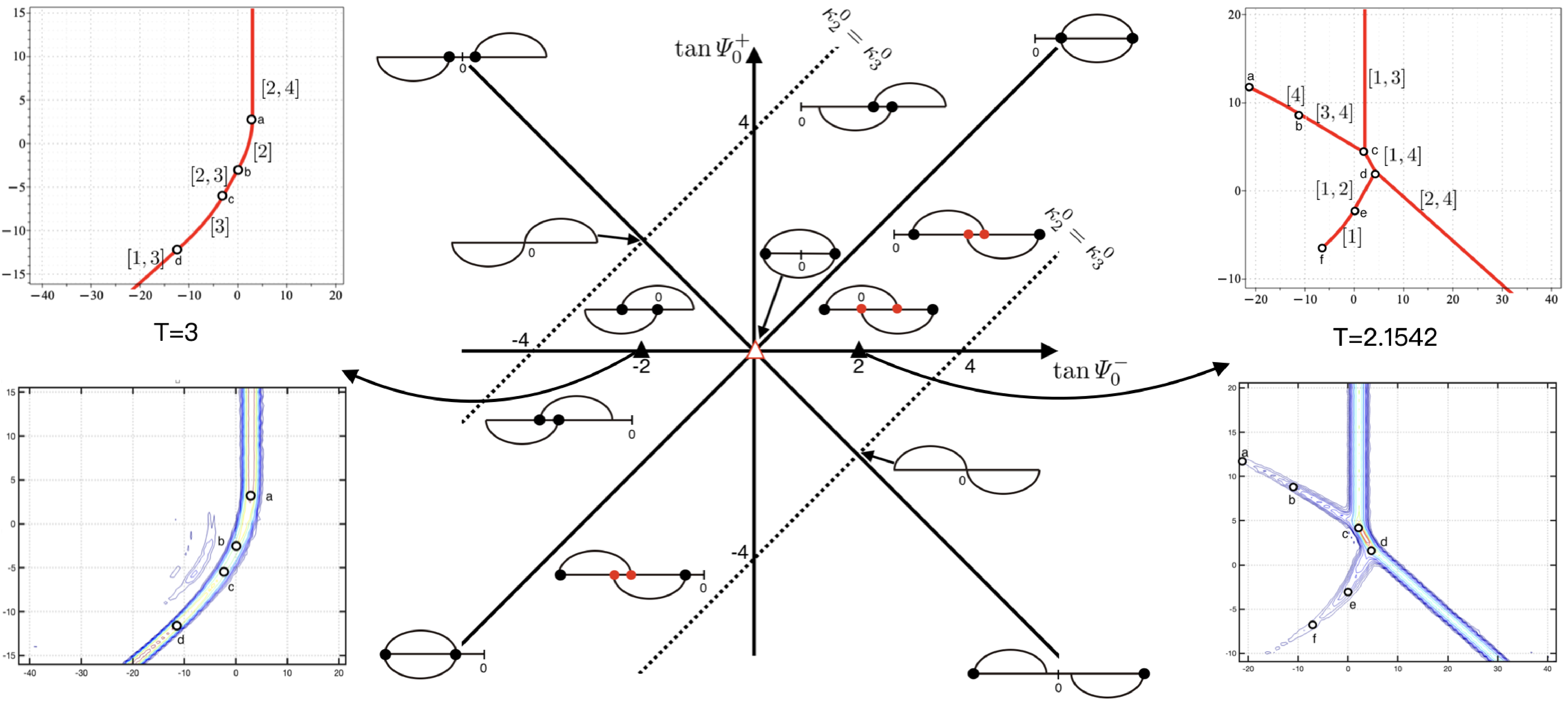}
	\end{minipage}%
	\caption{The solutions of the initial value problems with the initial data consisting of two
	half-solitons with the same amplitude $A_0$ but the different angles $\tan\Psi_0^{\pm}$.}
	\label{fig33}
\end{figure}

{Before concluding, we briefly summarize the main contributions of this work and discuss several possible directions for future research.}

{In this work, we employ an asymptotic perturbation approach to develop a precise, tractable analytical framework for describing the evolution of the V-shaped line-soliton initial-value problem. This approach provides a systematic, concise analytical classification of the phenomena observed in previous numerical studies \cite{Kao,McDowell}.	Although our discussion focuses on the evolution of line solitons in the KP equation, the method can be naturally extended to other nonlinear wave equations with similar structures. An important advantage of the asymptotic perturbation theory is that it does not rely strongly on the integrability of the equation, as shown in Appendix \ref{A-kappa}.}

{The present approach may apply to a broader class of systems, e.g., the Benney-Luke equation \cite{Luo:25}. Another interesting research direction concerns more complicated initial-value problems consisting of multiple line solitons, such as the configuration shown in Figure 6.15 of \cite{Kodama3}. We will further discuss these problems in future work.}



\bigskip
\noindent
{\bf Acknowledgements.} 
The authors would like to thank Harry Yeh for critical reading of the manuscript. They also appreciate a research fund from Shandong University of Science and Technology. One of the authors (C.L) is supported by National Natural Science Foundation of China (Grant No. 12071237).


\appendix
\section{The KP solitons}\label{A-KP}
{In this appendix, we briefly review KP solitons. Most of the material presented here can be found in \cite{Kodama3}.} The solution of the KP equation is commonly expressed in the form 
	\begin{eqnarray}\label{120}
		u(x,y,t)=2\left(\ln\tau(x,y,t)\right)_{xx},
	\end{eqnarray}
where the function $\tau(x,y,t)$ is called the \emph{tau function}. For the soliton solutions, the $\tau$ function is written in the following form
	\begin{eqnarray}\label{7}
		\tau(x,y,t) =\text{Wr}(f_{1},\cdots,f_{N}),
	\end{eqnarray}
where $\text{Wr}(f_{1},\cdots,f_{N})$ is the Wronskian of the functions $\{f_{i}(x,y,t):i=1,\cdots,N\}$ with respect to the $x$-variable, and the functions $f_{i}$'s are given by 
	\begin{eqnarray*}
		f_{i}(x,y,t)=\sum_{j=1}^{M} a_{i,j} E_{j}(x,y,t),\qquad\text{with}\quad E_{j}(x,y,t)=\exp(\kappa_j x+\kappa_j^2y-\kappa_j^3 t),
	\end{eqnarray*}
in which $A:=(a_{i,j})$ is an $N\times M$ matrix with $N<M$. Note that if $N=M$, the solution becomes trivial $u=0$. We assume that t parameters $\{\kappa_j:j=1,\ldots,M\}$ are ordered as
\begin{equation}\label{order}
\kappa_1~<~\kappa_2~<~\cdots~<~\kappa_M.
\end{equation}
Then the $\tau$-function \eqref{7} can be written in the determinant form $\tau=|AE^T|$ with the $N\times M$ matrix defined by
\[
E(x,y,t):=\begin{pmatrix}
E_1 & E_2 &\cdots & E_M \\
\kappa_1 E_1&\kappa_2 E_2&\cdots &\kappa_M E_M\\
\vdots &\vdots &\ddots &\vdots \\
\kappa_1^{N-1}E_1&\kappa_2^{N-1}E_2 &\cdots &\kappa_M^{N-1}E_M
\end{pmatrix},
\]
where $E^T$ denotes the transpose of the matrix $E$. Using the Binet-Cauchy lemma for the determinant, the solution (\ref{7}) can be expressed in the form
	\begin{eqnarray}\label{26}
		\tau(x,y,t)=|AE(x,y,t)^T|=\sum_{1\le i_{1}<\cdots<i_{N}\le M} \Delta_{i_{1},\cdots,i_{N}}(A)E_{i_{1},\cdots,i_{N}}(x,y,t),
	\end{eqnarray}
in which $\Delta_{i_{1},\cdots,i_{N}}(A)$ is the $N\times N$ minor of the matrix $A$ whose columns are labeled by the index set $I=\{i_{1}<\cdots<i_{N}\}$, and
	\begin{eqnarray*}
		E_{i_{1},\cdots,i_{N}}(x,y,t):=\Wr(E_{i_1},E_{i_2},\ldots,E_{i_N})=\prod _{l<m} (\kappa_{i_m}-\kappa_{i_l})E_{i_1}\cdots E_{i_N}.
	\end{eqnarray*}
With the ordering \eqref{order}, the exponential functions are positive definite, i.e., $E_{i_1,\ldots,i_N}(x,y,t)>0$. It was then shown in \cite{KW:13,KW:14} that the $\tau$-function \eqref{7} is positive definite if and only if the matrix $A$ is totally nonnegative (TNN). This implies that the soliton solution generated by the $\tau$-function above is regular, if and only if all the minors $\Delta_{i_1,\ldots,i_N}(A)\ge 0$. The real and regular soliton solution of the KP equation is referred to as the KP soliton. All the KP solitons are classified in terms of the parameters \eqref{order} and the TNN matrix $A$ (see \cite{Kodama3} and the references therein). To state the classification theorem, we first assume that the matrix $A$ of rank$(A)=N$ is \emph{irreducible}, meaning that the row reduced echelon form (RREF) of $A$ has no zero column and no row containing only pivot.
\begin{theorem}[\cite{Chakravarty,Kodama3}]\label{classification}
Let $\{i_1,i_2,\ldots,i_N\}$ be the pivot set, and $\{j_1,j_2,\ldots,j_{M-N}\}$ be the non-pivot set of the irreducible TNN matrix $A$ in RREF. Then the solution generated by the $\tau$-function \eqref{7} 
can be parametrized by a unique permutation $\pi\in S_M$, the symmetric group $S_M$ of $M$ numbers, in the sense that the KP soliton has the following asymptotic structure.
\begin{itemize}
\item[(a)] For $y\gg0$, there exist $N$ line-solitons of $[i_k,\pi(i_k)]$-type for some $i_k<\pi(i_k)\le M$ and $k=1,\ldots,N$.
\item[(b)] For $y\ll 0$, there exist $(M-N)$ line solitons of $[\pi(j_l),j_l]$-type for some $1\le \pi(j_l)<j_l$ and $l=1,\ldots,M-N$.
\end{itemize}
{(see also \cite{B:06} for asymptotic solitons).}
\end{theorem}

We define the \emph{chord} diagram associated with the permutation $\pi\in S_M$ (see \cite{Kodama3} and the references therein).
\begin{definition}
Consider a line segment on the real line with $M$ marked points labeled by the $\kappa$-parameters
$\{\kappa_1<\kappa_2<\cdots<\kappa_M\}$. The chord diagram associated with the permutation 
(derangement) $\pi\in S_M$ is defined by
\begin{itemize}
\item[(a)] if $i<\pi(i)$ (pivot index), then draw a chord joining $\kappa_i$ and $\kappa_{\pi(i)}$ on the upper part of the line, and
\item[(b)] if $j> \pi(j)$ (non-pivot index), then draw a chord joining $\kappa_j$ and $\kappa_{\pi(j)}$ on the lower part of the line.
\end{itemize}
With the classification theorem \ref{classification}, we let $\pi(A)$ denote the corresponding permutation of the TNN matrix $A$.
\end{definition}

\begin{example}[Example 5.6 in \cite{K:10}]
Consider the following $3\times 6$ TNN matrix,
\[
A=\begin{pmatrix}
1&0&-a&-b&0&c\\
0&1&d&e&0&-f\\
0&0&0&0&1&g
\end{pmatrix},
\]
where all the parameters $a, b,\ldots, g>0$ with $bd-ae>0$ and $cd-fa>0$.
Then we have 
\[
\pi(A)=(4,5,1,2,6,3),
\]
which implies that the corresponding KP soliton has asymptotically
\begin{itemize}
\item[(a)] for $y\gg0$, three solitons of $[1,4]$-, $[2,5]$-, and $[5,6]$-type, 
\item[(b)] for $y\ll 0$, three solitons of $[1,3]$-, $[2,4]$-, and $[3,6]$-type.
\end{itemize}
The KP soliton and the corresponding chord diagram are shown in Figure \ref{fig135}.
 \begin{figure}[htbp]
				\begin{minipage}[t]{1\linewidth}
					\centering
					\includegraphics[height=3.5cm,width=11.5cm]{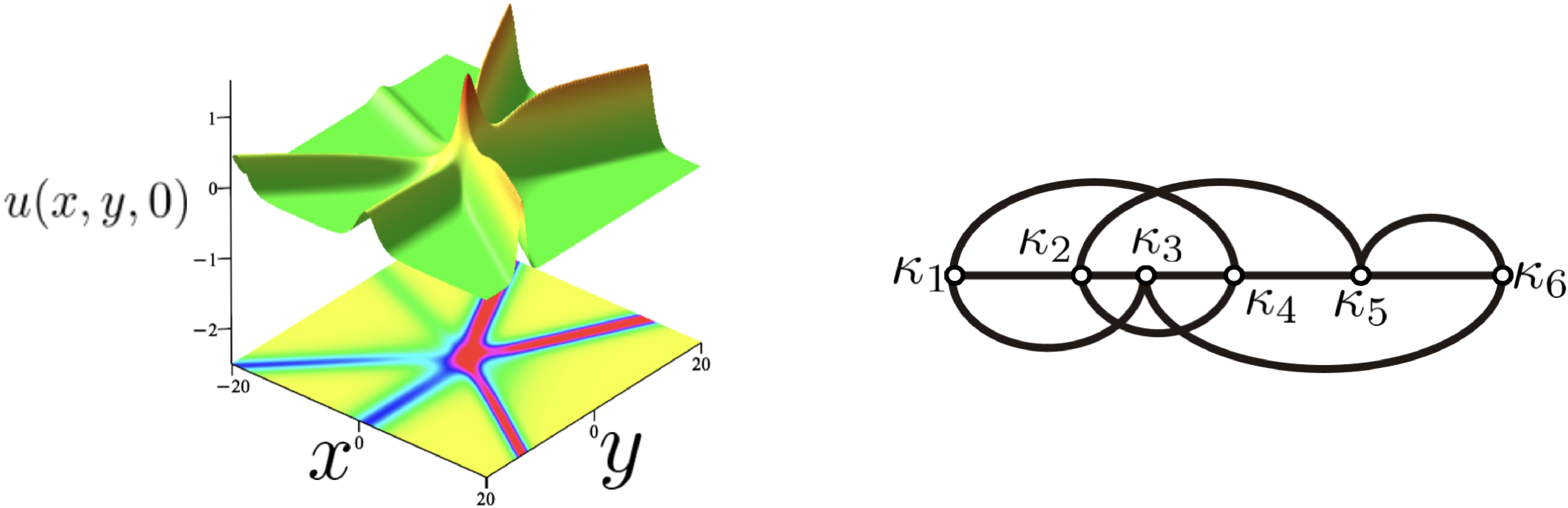}
				\end{minipage}%
				\caption{Example of the KP soliton with $\pi=(4,5,1,2,6,3)$.}
				\label{fig135}
			\end{figure}

\end{example}
{In this paper, we consider deformation (dynamics) of the  chord diagrams to describe the evolution of the perturbed solitons (the chord dynamics).}
\section{The {$\kappa$}-system}\label{A-kappa}
In this appendix, we derive the $\kappa$-system (\ref{6}) for the parameters $\{\kappa_{1},\kappa_{2}\}$ of $[1,2]$-soliton using an asymptotic perturbation theory. Although the system has been derived in \cite{GKP:18,Ryskamp}, we here give a much more elementary derivation using a standard perturbation method. We also emphasize that our slow variables are just $Y=\epsilon y$ and $T=\epsilon t$.

First writing $u=2\phi_{x}$, we have the potential form of the KP equation, 
		\begin{eqnarray}\label{1}
			4\phi_{xt}+12\phi_{x}\phi_{xx}+\phi_{xxxx}+3\phi_{yy}=0.
		\end{eqnarray}
This equation admits a shock type solution in $\phi$, i.e.,
		{\begin{eqnarray}\label{2}
				\phi(x,y,t)=\frac{\kappa_1+\kappa_2}{2}
				+\frac{\kappa_1-\kappa_2}{2}
				\tanh\!\left(\frac{\kappa_2-\kappa_1}{2}\,\xi(x,y,t)\right),
		\end{eqnarray}}where $\xi(x,y,t)=x+Q y-Ct$ with $Q:=\kappa_1+\kappa_2$ and $C:=\kappa_1^2+\kappa_1\kappa_2+\kappa_2^2$. 
Note for $\kappa_{1}<\kappa_{2}$ that the solution behaves as
	\begin{eqnarray*}
		\phi(x,y,t)\quad\longrightarrow\quad \left\{\begin{array}{ll}
			\kappa_{1},\quad &\text{for}\quad x\ll 0,\\
			\kappa_{2},&\text{for}\quad x\gg 0.
		\end{array} \right.
	\end{eqnarray*}
We study a perturbation problem of the one-soliton solution under the assumption of adiabatic
modulation of the parameters $\{\kappa_1,\kappa_2\}$. We then introduce the slow scales
with a small parameter $0<\epsilon\ll 1$,
\begin{equation}\label{slow}
Y=\epsilon y,\qquad\text{and}\qquad T=\epsilon t.
\end{equation}
Note here that we consider $x$ to be a fast scale with $\xi=x+Qy-Ct$. {That is, we consider the charge of the coordinates $(x,y,t)$ to $(\xi,y,t)$.} With the new variables $(\xi, Y, T)$, we have
\begin{equation}\label{SlowP}
\frac{\partial}{\partial x}=\frac{\partial}{\partial \xi},\qquad
\frac{\partial}{\partial y}=Q\frac{\partial}{\partial \xi}+\epsilon\frac{\partial}{\partial Y},\qquad
\frac{\partial}{\partial t}=-C\frac{\partial}{\partial \xi}+\epsilon\frac{\partial}{\partial T},
\end{equation}
where $Q$ and $C$ are function of $(Y,T)$. Then one can see that these variables $(\xi,Y,T)$ are compatible only if the following is satisfied
\begin{equation}\label{compatibility}
\frac{\partial Q}{\partial T}+\frac{\partial C}{\partial Y}=0,
\end{equation}
which is sometimes referred to as the conservation of wave numbers (see e.g. \cite{Ryskamp}).
This equation is derived from the compatibility of the original variables $(x,y,t)$, i.e.,
\[
			\frac{\partial^2}{\partial x \partial y}=\frac{\partial^2}{\partial y \partial x},\qquad
			\frac{\partial^2}{\partial y \partial t}=\frac{\partial^2}{\partial t  \partial y},\qquad
			\frac{\partial^2}{\partial x \partial t}=\frac{\partial^2}{\partial t \partial x}.   
\]
Using \eqref{SlowP}, the KP equation (\ref{1}) becomes
\begin{eqnarray}\label{PKP}
			-4C\phi_{\xi\xi}+12\phi_\xi\phi_{\xi\xi}+\phi_{\xi\xi\xi\xi}+3Q^2\phi_{\xi\xi}
			+\epsilon\left(4\phi_{\xi T}+3Q\phi_{\xi Y}+3(Q\phi_\xi)_Y\right)+3\epsilon^2\phi_{YY}=0.
		\end{eqnarray}
Now we assume the following asymptotic form of the solution,
\begin{equation}\label{eigenF}
			\phi(x,y,t)=\phi^{(0)}(\xi,Y,T)+\epsilon\phi^{(1)}(\xi,Y,T)+\mathcal{O}(\epsilon^2).
\end{equation}
Inserting \eqref{eigenF} into \eqref{PKP}, we have, at the leading order  $\epsilon=0$,
		\begin{eqnarray}\label{4}
			-4C\phi^{(0)}_{\xi\xi}+12\phi^{(0)}_{\xi}\phi^{(0)}_{\xi\xi}+\phi^{(0)}_{\xi\xi\xi\xi}+3Q^2\phi^{(0)}_{\xi\xi}=-(\kappa_{1}-\kappa_{2})^2\phi^{(0)}_{\xi\xi}+12\phi^{(0)}_{\xi}\phi^{(0)}_{\xi\xi}+\phi^{(0)}_{\xi\xi\xi\xi}=0,
		\end{eqnarray}
where we have used $4C-3Q^2=(\kappa_1-\kappa_2)^2$, and the solution is given by \eqref{2}. At the order $\epsilon$, we have
\begin{eqnarray}\label{O1}
			\mathcal{L}[\phi^{(0)}]\phi^{(1)}=-4\phi^{(0)}_{\xi T}-3Q\phi^{(0)}_{Y \xi}-3(Q \phi^{(0)}_{\xi })_{Y},		
\end{eqnarray}
where $\mathcal{L}[\phi^{(0)}]$ is the linearization operator for Eq.~(\ref{4}), i.e.,
\[
\mathcal{L}[\phi^{(0)}]=-(\kappa_{1}-\kappa_{2})^2\frac{\partial^2}{\partial\xi^2}+12\frac{\partial}{\partial\xi}\cdot\phi^{(0)}_{\xi}\cdot\frac{\partial}{\partial \xi}+\frac{\partial^4}{\partial \xi^4}.
\]
Note that this operator is (formally) self-adjoint in the space of bounded functions. Since $\phi^{(0)}_{\xi} \in \ker \mathcal{L}[\phi^{(0)}] \cap L^2(\mathbb{R})$ and $\phi^{(1)}$ is assumed to be bounded, it follows that
{\begin{align*}
		\langle \phi^{(0)}_{\xi},\mathcal{L}[\phi^{(0)}]\phi^{(1)}\rangle
		:= \int_{\mathbb{R}} \phi^{(0)}_{\xi}
		\mathcal{L}[\phi^{(0)}]\phi^{(1)}\, d\xi 
		=  \langle \mathcal{L}[\phi^{(0)}]\phi^{(0)}_{\xi},\phi^{(1)}\rangle =0 .
\end{align*}}This implies 
{\[\int_{\mathbb{R}}\left(
-2\frac{\partial}{\partial T}(\phi^{(0)}_{\xi})^2
-3\frac{\partial}{\partial Y}\bigl(Q(\phi^{(0)}_{\xi})^2\bigr)
\right)d\xi=\int_{\mathbb{R}} \phi^{(0)}_{\xi} (-4\phi^{(0)}_{\xi T}-3q\phi^{(0)}_{Y \xi}-3(q \phi^{(0)}_{\xi })_{Y}) d\xi=0,\]}which leads
\begin{equation}\label{5}
 \frac{\partial }{\partial T}(\kappa_{1}-\kappa_{2})^3+\frac{3}{2}\frac{\partial }{\partial Y}\left(Q(\kappa_{1}-\kappa_{2})^3\right)=0.
\end{equation} 
Then, together with \eqref{compatibility}, we obtain the $\kappa$-system (\ref{6}).

\section{Regularization in the KdV-Whitham equation}\label{A-KW}
{Here, we show that our regularization, discussed in Section \ref{sec:d}, is similar to that used in the Whitham averaging theory \cite{BK:94, Kodama5, Wh:74}.}

Let us first review an elliptic solution of the KdV equation,
\begin{equation}\label{KdV}
4u_t+6uu_{x}+u_{xxx}=0,
\end{equation}
which is the KP equation under the assumption $u_y=0$.
It is well known that the equation admits a periodic solution given by
\begin{eqnarray}\label{110}
u(x,t)&=r_{2}+r_{3}-r_{1}-2(r_2-r_{1})\,\text{sn}^{2}\big(\sqrt{r_{3}-r_{1}}(x-\frac{1}{2}(r_{1}+r_{2}+r_{3})t),m\big),\\
&=r_1+r_2-r_3+2(r_3-r_1)\,\text{dn}^2\big(\sqrt{r_{3}-r_{1}}(x-\frac{1}{2}(r_{1}+r_{2}+r_{3})t),m\big),\nonumber
\end{eqnarray}
where $r_1<r_2<r_3$ are parameters, $\text{sn}(z,m)$ and $\text{dn}(z,m)$ are the Jacobi elliptic functions with $m:=\tfrac{r_2-r_1}{r_3-r_1}$.
The average value of $u(x,t)$ over the period $L=\frac{2K(m)}{\sqrt{r_3-r_1}}$ is
\begin{equation}\label{average}
\bar{u}:=\frac{1}{L}\displaystyle\int_{0}^{L}u(x,t)\,dx=r_1+r_2-r_3+2(r_3-r_1)\frac{E(m)}{K(m)},
\end{equation}
where $K(m)$ and $E(m)$ denote the complete elliptic integrals of the first and second kinds, respectively. In the Whitham theory, the parameters $(r_1,r_2,r_3)$ depend on the slow scales
$(X=\epsilon x,T=\epsilon t)$, that is, $\bar{u}=U(X,T)$.

Let us suppose that the initial data $u(x,0)$ depends only on the slow scale, i.e., $u(x,0)$ changes slowly and $u(x,0)=U(X,0)$, no rapid oscillation. The function $U(X,T)$ then satisfies the \emph{dispersionless} KdV equation,
\begin{equation}\label{dKdV}
4U_T+6UU_X=0,
\end{equation}
that is, we have ignored the dispersion term in \eqref{KdV}. 
Now we consider the following step initial data,
\begin{eqnarray}\label{96}
		u(x,0)=U(X,0)=\left\{\begin{array}{ll} 
			a,\quad &\text{for}\quad X<0,\\
			b, &\text{for}\quad X>0.
		\end{array} \right.
\end{eqnarray}
The (implicit) solution of \eqref{dKdV} is given by
\[
U(X,T)=F(X-\frac{3}{2}UT),
\]
where $F(X)=U(X,0)$, the initial function. As discussed in Section \ref{Sec:kappa},
if $a<b$, then we have a global solution corresponding to a rarefaction wave.
And, if $a>b$, the equation \eqref{dKdV} develops a shock singularity.
Then, the KdV equation \eqref{KdV} with the initial data \eqref{96} develops
a so-called \emph{dispersive shock wave}, which can be considered as a slow modulation
of the periodic solution \eqref{110} (see \cite{Gurevich,Kamchatnov,Wh:74}). The Whitham equation is then given
by a quasilinear system of equations for the parameters $(r_1,r_2,r_3)$. That is, 
for the case with $a>b$, we need to consider the Whitham equation for $(r_1,r_2,r_3)$
instead of the single equation \eqref{dKdV}. Then the initial data \eqref{96} should be
expressed in terms of these variables. This is the regularization considered in \cite{BK:94, Kodama5}.
The initial data for the parameters $(r_1,r_2,r_3)$ is then given by
\begin{align}\label{111}\left\{\begin{array}{lll}
r_1(X,0)=b,\quad\quad\text{for}\quad X\in\mathbb{R},\\[1.0ex]
r_2(X,0)=\left\{\begin{array}{ll}
b,\quad&\text{for}\quad X<0,\\
a,\quad &\text{for}\quad X>0,
\end{array}\right.\\[2.0ex]
r_3(X,0)=a,\quad\quad\text{for}\quad X\in\mathbb{R}.
\end{array}\right.
\end{align}
This initial data corresponds to the following limits of the averaged function $\bar{u}=U(X,0)$,
\begin{equation}\label{limits}
U(X,0)=\left\{\begin{array}{ll}
\displaystyle{\lim_{r_2\to r_1}\bar{u}=r_3=a},\quad &\text{for}\quad X<0,\\[2.0ex]
\displaystyle{\lim_{r_2\to r_3}\bar{u}=r_1=b},\quad &\text{for}\quad X>0.
\end{array}\right.
\end{equation}
Here we have used the following limits,
\[
\lim_{r_2\to r_1}E(m)=\lim_{r_2\to r_1}K(m)=\frac{\pi}{2},\qquad
\lim_{r_2\to r_3}E(m)=1,\quad \lim_{r_2\to r_3}K(m)=\infty.
\]
In terms of the periodic solution \eqref{111}, we have
\begin{itemize}
	\item[(a)]  the limit $r_{2}\to r_{3}$  ($m\to1$) gives
		\begin{eqnarray*}
			u(x,t)\quad\longrightarrow\quad  r_{1}+2(r_3-r_1)\sech^2(\sqrt{r_{3}-r_{1}}(x-\frac{1}{2}(r_{1}+2r_{3})t)),
		\end{eqnarray*}
	where we have used $\text{dn}^2(z,m)\to\sech^2z$. This is called the \emph{soliton limit}.
	\item[(b)]  the limit $r_{2}\to r_{1}$ ($m\to0$) gives
		\begin{eqnarray*}
			u(x,t)\quad\longrightarrow\quad r_{3}+2(r_2-r_1)\sin^2(\sqrt{r_{3}-r_{1}}(x-\frac{1}{2}(2r_{1}+r_{3})t))\approx r_{3},
		\end{eqnarray*}
	where we have used $\text{sn}(z,m)\to\sin z$. This is called the \emph{linear limit}.
\end{itemize}
We illustrate the regularization in Figure \ref{fig85}.
\begin{figure}[htbp]
	\begin{minipage}[t]{1\linewidth}
		\centering
		\includegraphics[height=2.5cm,width=12.8cm]{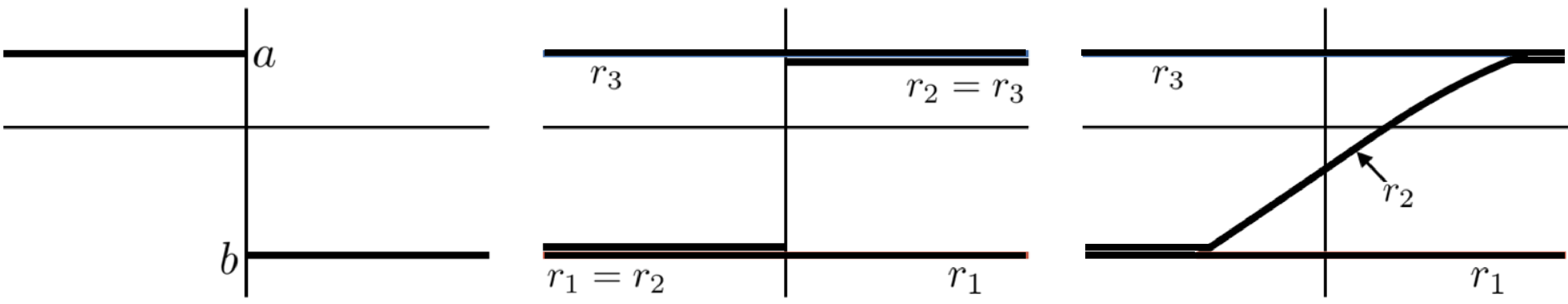}
	\end{minipage}%
	\caption{The regularization in the KdV-Whitham equation.
	The left panel is $u(x,0)$ in \eqref{96}. The middle panel is the regularized initial data
	for $(r_1,r_2,r_3)$. The right panel shows the solution for $T>0$.}
	\label{fig85}
\end{figure}


\bigskip
 \noindent
 {\bf Declarations.}
 \begin{itemize}
\item Data sharing not applicable to this article as no datasets were generated or analyzed during the current study.\\
 \item The authors declare no conflicts of interest associated with this manuscript.
 \end{itemize}



\begin{thebibliography}{99}

\bibitem{BK:94} A. M. Bloch and Y. Kodama, The Whitham equation and shocks in the Toda lattice. In Singular limits of dispersive waves, Eds. N. M. Ercolani, I.R. Gabitov, C.D. Levermore, D. Serre. (1994), 1-19.

\bibitem{B:06}
{G. Biondini and S. Chakravarty,  Soliton solutions of the Kadomtsev-Petviashvili II equation. J. Math. Phys. {\bf  47} (2006) 033514.}


\bibitem{Chakravarty} S. Chakravarty and Y. Kodama, Classification of the line-soliton solutions of KP II. J. Phys. A: Math. Theor. {\bf 41} (2008) 275209.


\bibitem{CK:09} {S. Chakravarty and Y. Kodama},
{Soliton solutions of the KP equation and application to shallow water waves}. 
Stud. Appl. Math. {\bf 123} (2009) 83-151.

\bibitem{GKP:18}
T. Grava, C. Klein and G. Pitton,
Numerical study of the Kadomtsev–Petviashvili equation
and dispersive shock waves. Proc. R. Soc. A {\bf 474} (2018) 20170458.

\bibitem{Gurevich} A.V. Gurevich and L.P. Pitayevsky, 
Nonstationary structure of a collisionless shock wave. 
Sov. J. Exp. Theor. Phys {\bf 38} (1974) 291-297.

\bibitem{Luo:25}
{L. Hu, X.D. Luo and Z. Wang, Obliquely interacting solitary waves and wave wakes in free-surface flows. J. Fluid Mech. {\bf 1011} (2025) A8:1-35.}

\bibitem{J:14}
{Y.H. Jia},
{\em Numerical study of the KP solitons and higher order Miles theory of the Mach reflection in shallow water},  Ph.D. Thesis, The Ohio State University (2014).

\bibitem{Kamchatnov} A.M. Kamchatnov, 
{\it Nonlinear periodic waves and their modulations: an introductory course}.
(World Scientific, Singapore, 2000).

\bibitem{Kao} C.Y. Kao and Y. Kodama, 
Numerical study of the KP equation for non-periodic waves, 
Math. Comput. Simulation {\bf 82} (2012) 1185-1218.


\bibitem{Kodama5} Y. Kodama, 
The Whitham equations for optical communications: Mathematical theory of NRZ. 
SIAM J. Appl. Math. {\bf 59} (1999) 2162-2192.

\bibitem{Kodama} Y. Kodama, M. Oikawa and H. Tsuji, 
Soliton solutions of the KP equation with V-shape initial waves. 
J. Phys. A: Math. Theor. {\bf 42} (2009)  312001.

\bibitem{K:10}
{Y. Kodama}, {KP solitons in shallow water}, 
J. Phys. A: Math. Theor. {\bf 43} (2010) 434004.


\bibitem{Kodama3} Y. Kodama,
 {\it Solitons in two-dimensional shallow water}.
CBMS-NSF, {\bf 92}, (SIAM, Philadelphia, 2018).

\bibitem{KW:13} {Y. Kodama and L. Williams}, 
{The Deodhar decomposition of the Grassmannian and the regularity of KP solitons}. 
Adv. Math. {\bf 244} (2013) 979-1032.


\bibitem{KW:14} {Y. Kodama and L. Williams}, 
{KP solitons and total positivity for 
the Grassmannian}. Invent. Math. {\bf 198} (2014) 637-699.


\bibitem{KY:16}
 {Y. Kodama and H. Yeh},
 {The KP theory and Mach reflection},
 J. Fluid Mech. {\bf 800} (2016) 766-786.



\bibitem{LYK:11}
{W. Li, H. Yeh and Y. Kodama}, {On the Mach reflection of a solitary wave: revisited}.
J. Fluid. Mech. {\bf 672} (2011) 326-357.


\bibitem{McDowell} T. McDowell, M. Osborne, S. Chakravarty and Y. Kodama, 
On a class of initial value problems and solitons for the KP equation: A numerical study. 
Wave Motion {\bf 72} (2017) 201-227.

\bibitem{Miles2} J.W. Miles. 
Resonantly interacting solitary waves. 
J. Fluid Mech. {\bf 79} (1977) 171-179.


\bibitem{Mizumachi} T. Mizumachi, 
{\it Stability of line solitons for the KP II equation in $\mathbb{R}^2$}. 
Memoirs of the American Mathematical Society, {\bf 238} (AMS, Providence, RI, 2015).


\bibitem{PTLO}
  {A. V. Porubov, H. Tsuji, I. L. Lavrenov and M. Oikawa},
 {Formation of the rogue wave due to non-linear two-dimensional waves interaction}.
  {Wave Motion} {\bf 42} (2005) 202-210.


\bibitem{Ryskamp1} S.J. Ryskamp, M.D. Maiden, G. Biondini and M A. Hoefer, 
Evolution of truncated and bent gravity wave solitons: the Mach expansion problem. 
J. Fluid Mech. {\bf 909} (2021) A24:1-33.


\bibitem{Ryskamp} S.J. Ryskamp, M.A. Hoefer and G. Biondini, 
Modulation theory for soliton resonance and Mach reflection. 
Proc. Royal Soc. A {\bf 478} (2022) 20210823.



\bibitem{Wang} C. Wang and R. Pawlowicz, 
Oblique wave-wave interactions of nonlinear near-surface internal waves in the Strait of Georgia. 
J. Geophys. Res.: Oceans {\bf 117} (2012) C06031.


\bibitem{Yuan2} C.X. Yuan, R. Grimshaw, E. Johnson and Z. Wang, 
Topographic effect on oblique internal wave-wave interactions. 
J. Fluid Mech. {\bf 856} (2018) 36-60.

\bibitem{Wh:74}
{G. B. Whitham}, {\it Linear and nonlinear waves}. 
(Wiley-Interscience, New York, 1974).



\end{thebibliography}
\end{document}